\documentclass[11pt,oneside]{report}

\usepackage[parfill]{parskip} 
\usepackage[a4paper,includefoot,nohead]{geometry}
\usepackage{graphicx}
\usepackage{sidecap}  
\usepackage[small,bf]{caption}

\usepackage[colorlinks=true, pdfstartview=FitV, linkcolor=blue, citecolor=blue, urlcolor=blue]{hyperref}
\usepackage[toc,page]{appendix}

\usepackage{amsmath}
\usepackage{amsthm}
\usepackage{amssymb}
\usepackage{bbm}

\usepackage{algorithm}
\usepackage{algorithmic}
\usepackage{setspace}

\newtheorem{theorem}{Theorem}[section]
\newtheorem{definition}[theorem]{Definition}
\newtheorem{lemma}[theorem]{Lemma}

\newtheorem{proposition}[theorem]{Proposition}
\newtheorem{corollary}[theorem]{Corollary}
\newtheorem{conjecture}[theorem]{Conjecture}

\def\BQP{\mathbf{BQP}}
\def\BPP{\mathbf{BPP}}
\def\PP{\mathbf{PP}}
\def\NP{\mathbf{NP}}
\def\P{\mathbf{P}}
\def\L{\mathbf{L}}
\def\ParL{\mathbf{\oplus L}}
\def\FH{\mathbf{FH}}

\def\NC{\mathbf{NC}}
\def\QNC{\mathbf{QNC}}

\def\Bb#1#2{{BP}_{#2}(#1)}

\def\Postb#1#2{{Post}_{#2}(#1)}

\def\DQC#1{\mathbf{BQ}[#1]\mathbf{P}}

\def\PDQC#1{\mathbf{PQ}[#1]\mathbf{P}}

\def\H{\mathbf{H}}

\def\0{\mathbf{0}}
\def\a{\mathbf{a}}
\def\b{\mathbf{b}}
\def\c{\mathbf{c}}
\def\d{\mathbf{d}}
\def\e{\mathbf{e}}

\def\p{\mathbf{p}}

\def\r{\mathbf{r}}
\def\s{\mathbf{s}}
\def\t{\mathbf{t}}

\def\x{\mathbf{x}}
\def\X{\mathbf{X}}
\def\y{\mathbf{y}}
\def\Y{\mathbf{Y}}
\def\z{\mathbf{z}}

\def\cFS{\mathcal{FS}}
\def\IQP{\mathcal{IQP}}

\def\cC{\mathcal{C}}
\def\Code{\mathcal{C}}
\def\cE{\mathcal{E}}
\def\cG{\mathcal{G}}
\def\cL{\mathcal{L}}
\def\cM{\mathcal{M}}
\def\cO{\mathcal{O}}
\def\cP{\mathcal{P}}
\def\cQ{\mathcal{Q}}

\def\cT{\mathcal{T}}
\def\cU{\mathcal{U}}
\def\cV{\mathcal{V}}
\def\cX{\mathcal{X}}
\def\cZ{\mathcal{Z}}

\def\CC{\mathbbm{C}}
\def\FF{\mathbbm{F}}
\def\NN{\mathbbm{N}}

\def\RR{\mathbbm{R}}
\def\ZZ{\mathbbm{Z}}

\def\Pr{\mathbbm{P}}
\def\Ex{\mathbbm{E}}

\def\ketbra#1#2{\vert#1\rangle\langle#2\vert}

\def\ket#1{\left|#1\right>}
\def\bra#1{\langle#1\vert}

\def\span#1{\left<~#1~\right>}

\def\eps{\varepsilon}

\def\rhoS{\rho_{\rm start}}

\def\rhoS{\rho_{\mbox{start}}}

\def\naive{\emph{na\"\i{}ve}}
\def\etc{\emph{\&c.}}
\def\ie{\emph{i.e.}}
\def\eg{\emph{e.g.}}

\def\cf{\emph{cf.}}
\def\viz{\emph{viz}}

\begin{document}

\pagestyle{empty}

\begin{center}

\vspace{20mm}

{\huge \bf Quantum Complexity~:\\[2mm] restrictions on algorithms\\ and architectures}

\vspace{70mm}

{\Large \textbf{Daniel James Shepherd}}

\vspace{80mm}

A dissertation submitted to the University of Bristol in accordance with the requirements of the degree of Doctor of Philosophy (PhD) in the Faculty of Engineering, Department of Computer Science, July 2009.

\vspace{8mm}

\begin{flushright}40,000 words\end{flushright}

\end{center}


\pagebreak


\begin{center}
\vspace{16mm}

{\huge \bf Quantum Complexity~:\\[2mm] restrictions on algorithms\\ and architectures}

\vspace{8mm}

{\textbf{ \Large{Daniel James Shepherd,} \normalsize{MA (Cantab)}}}

\copyright ~2009

\vspace{7mm}

{\Large \bf Abstract}

\end{center}

\bigskip
\vbox{
We study discrete-time quantum computation from a theoretical perspective.

We describe some frameworks for universal quantum computation with limited control---namely the one-dimensional cellular automaton, and the qubit spin-chain with all access limited to one end of the chain---and we obtain efficient constructions for them.  These two examples help show how little control is necessary to make a universal framework, provided that the gates which are implemented are noiseless.  It is hoped that the latter example might help motivate research into novel solid-state computing platforms that are almost totally isolated from couplings to the environment.

Recalling concepts from the theory of algorithmic complexity and quantum circuits, we show some positive results about the computational power of the so-called ``one clean qubit'' model and its ensuing naturally defined bounded-probability polytime complexity class, showing that it contains the class $\ParL$ and giving oracles relative to which it is incomparable with $\P$.

We study quantum computing models based on the Fourier Hierarchy, which is a conceptually straightforward way of regarding quantum computation as a direct extension of classical computation.  The related concept of the \emph{Fourier Sampling oracle} provides a subtly different perspective on the same mathematical constructions, again establishing the centrality of the Hadamard transform as one way of extending classical ideas to quantum ones.  We examine quantum algorithms that naturally employ these concepts, recasting some well-known number-theoretic algorithms into these models.  In particular, a detailed example is given of how, using only gates that would preserve the computational basis, it is possible to render a version of Shor's algorithm where initialisation and readout are performed in the Hadamard basis.

In a similar vein, we study models based on the Clifford-Diagonal Hierarchy, and introduce the \emph{$\IQP$ oracle}, illustrating that temporal complexity need not be necessary for some notions of quantum complexity.  We examine simple protocols that arise from these notions, in particular providing an example of a protocol which it is hoped could be of significant use in testing quantum computers that are rather limited in terms of computational depth.  We also provide some analysis of the classical techniques that approximate the signalling required within such protocols, arguing that some specifically `quantum' complexity can appear in the absence of temporal structure.
}


\pagebreak


\begin{center}
\vspace{16mm}

{\Large \bf Acknowledgements}

\end{center}

\bigskip

I'd like to thank Richard Jozsa, my supervisor at Bristol, who has provided continual support and encouragement throughout my time here.  My time at the university has been made especially enjoyable by the friendship and interesting discussions with many university colleagues, and special mention should go to Michael Bremner, Sean Clark, Rapha\"el Clifford, Toby Cubitt, Aram Harrow, Nick Jones, Richard Low, Will Matthews, Ashley Montanaro, Tobias Osborne, and Tony Short, for patiently putting up with (and even encouraging) so many of my crazy suggestions and distractions.

Thanks also go to my colleagues at CESG, GCHQ, the Heilbronn Institute, and elsewhere, some of whom made possible my research through arranging funding and allowing me much liberty in direction, and many more of whom offered significant support and friendship throughout.  Peter Smith in particular deserves a special mention for introducing me to quantum algorithmics in the first place.

I'm grateful to all those who have invited me to speak on the subjects of this research.  Naturally, my co-authors have been a big help in the development of ideas and in the process of writing~: I thank Torsten Franz also for some of the images used in \S\ref{sect:QCA}.  Thanks are due to Chris Major, who taught me to code in Java and helped with the software I developed for graphical representation of the architecture described in \S\ref{sect:spinChains}.  

Of course, this research would have been frustrating without LANL's quant-ph archive and Google's search engine, and the write-up would not have been so straightforward without Donald Knuth's \LaTeX{} typesetter.

\bigskip
This thesis is dedicated to my sister, Anna, with much love.

\medskip
\hfill  $\mathcal{DJS}$, 2009.


\pagebreak


\begin{center}
\vspace{20mm}

{\Large \bf Author's Declaration}

\end{center}
\vspace{5mm}

I declare that the work in this dissertation was carried out in accordance with the Regulations of the University of Bristol. The work is original, except where indicated by special reference in the text, and no part of the dissertation has been submitted for any other academic award. Any views expressed in the dissertation are those of the author.
\\[20mm]

\begin{tabular}{ll}
SIGNED:&.........................................\\[10mm]
DATE:  &.........................................
\end{tabular}


\pagebreak


\onehalfspacing
\pagestyle{plain}

\tableofcontents  \label{sect:contents}
\setcounter{chapter}{-1}

\clearpage
\pagebreak
\setboolean{@twoside}{true}

\cleardoublepage

\chapter{Preface}  \label{chap:intro}

\section{Overview}  \label{sect:overview}

Quantum algorithmics became widely recognised as a subject in its own right after the publication in 1994 of Shor's Algorithm \cite{lit:Shor95} and later Grover's Algorithm \cite{lit:Grover96} in 1996, these in turn having been inspired by the Deutsch-Jozsa algorithm of 1992.  These ideas underpin algorithmic primitives that illustrate a superiority of quantum information processing over classical information processing for solving certain problems whose statements and solutions are definable in purely classical terms, that is, without reference to the theory of Quantum Information Processing.  

This subject relates to many other disciplines, and as such enables many different valid approaches to be made to the potential `real-world solution' of such problems.
We employ a mathematical methodology (based on complexity theory) to formulate some new algorithmic constructions.  It is worthwhile briefly exploring some of the more philosophical issues surrounding the subject of quantum information and algorithmics before engaging with the mathematics.
Thus Chapter~\ref{chap:approach} is written in a non-rigorous style, freely borrowing notions from a range of authors, simply to put in place a few of the concepts that will be referred to in the more formally written, mathematically oriented, later sections.  

Following that, our goal is to cast some quantum algorithms into particular structures or frameworks that reflect some kind of physical limitation.  The motivation for doing this is to obtain new insight into such questions as
\begin{itemize}
  \item 
Which physical limitations do not significantly inhibit quantum computation?  What is the `simplest' architecture for a quantum computer?  (Chapter~\ref{chap:UCLC}.)

  \item
How useful are mixed states in quantum computing?  (Chapter~\ref{chap:PMC}.)

  \item
Which physical limitations enable a ready comparison with classical computation?  Which quantum algorithms are `close' to being classical?  What is the `simplest' quantum subroutine?  (Chapter~\ref{chap:FH}.)

  \item
Which physical limitations correspond to natural structures within the underlying mathematics?  What is the `simplest' fundamentally quantum protocol?  (Chapter~\ref{chap:IQP}.)
\end{itemize}

The answers we give to these questions are by no means complete, unconditional, or uncontroversial; rather, they present a particular style and particular ways of thinking about quantum complexity that may be useful in the development of new algorithms or in the future design of quantum computing architectures.

\newpage
\section{Previous publications}  \label{sect:previous}

Much of the content of this dissertation has been published previously, and some of it is joint work.

\subsubsection*{Chapter \ref{chap:UCLC}}

The first part is based on my paper \cite{me:UPQCA}, ``Universally programmable quantum cellular automaton''.  This was joint work with Torsten Franz and Reinhard Werner.  It was the first paper to give an explicit construction for a `universal' one-dimensional cellular automaton, in the physically-motivated sense introduced by Schumacher and Werner (2004).

\subsubsection*{Chapter \ref{chap:PMC}}

The example of the ``One pure qubit'' model contains material from an unpublished paper of mine that is available from the quant-ph archive, \href{http://arXiv.org/abs/quant-ph/0608132}{quant-ph/0608132}.

\subsubsection*{Chapter \ref{chap:FH}}

This draws heavily from my paper \cite{me:RQLN}, ``On the role of Hadamard gates in quantum circuits''.  The theme of the `Fourier Hierarchy', introduced by Yaoyun Shi (2003) and used in that paper, is developed further in this dissertation.  Following \cite{lit:BV97}, the term \emph{Fourier Sampling Oracle} is used here.

\subsubsection*{Chapter \ref{chap:IQP}}

This is based on my paper \cite{me:IQC}, ``Temporally unstructured quantum computation''; a joint work with Michael Bremner.  In that paper, we introduced a new quantum protocol using the `IQP model', on which basis the term \emph{IQP Oracle} is defined.  The work is recounted in this dissertation, with a slightly different emphasis. 



\newpage
\section{Notations}  \label{sect:notations}

Here we collect together some of the notations used in the dissertation that are perhaps not standard.

Pauli matrices are denoted $X$, $Y$ and $Z$.  These often occur as unitary transforms, but they are also Hermitian operators which may be used to define projectors.  For example, $\frac{1+Z}2$ is a projector from a two-dimensional space to a one-dimensional space, written in the Dirac notation as $\ketbra00$.
\begin{eqnarray*}
  X := \left( \begin{array}{cc} 0&1\\1&0 \end{array} \right),~~
  Y := \left( \begin{array}{cc} 0&-i\\i&0 \end{array} \right),~~
  Z := \left( \begin{array}{cc} 1&0\\0&-1 \end{array} \right).
\end{eqnarray*}
Subscripts on such symbols are used to indicate which qubits they pertain to.  For example, $X_a$ can refer to a unitary transform that `flips' bit $a$ according to $\ket0 \leftrightarrow \ket1$ by applying $X$ to the qubit labelled by $a$, or alternatively $X_a$ can refer to an Hermitian operator that `observes' the qubit labelled by $a$.
The Hadamard operator (or matrix) is given by
\begin{eqnarray*}
  H  &:=&  (~X+Z~)/{\sqrt2}.
\end{eqnarray*}

Many of the unitary operators we employ are `controlled gates', whose matrix representation with respect to the computational basis takes the form
\begin{eqnarray*}
  \left( \begin{array}{cc} 1&0\\0&U \end{array} \right).
\end{eqnarray*}
Such a matrix may be written $\Lambda(U)$, where the $\Lambda$ symbol denotes `control' (informally, ``apply gate $U$ to some qudit conditioned on some other control qubit being set'').  
When there are multiple controls, we use a superscript to count them, so $\Lambda^2(X) = \Lambda(\Lambda(X))$ for example would denote the \emph{Toffoli} gate.
Superscripts on unitaries (\emph{e.g.} $U^2$) denote sequential application of a unitary transform to a qudit, which is the same concept as raising to a power, algebraically.  Where parallel application is intended, we write $U^{\otimes 2}$ to mean $U$ applied on both of two qudits in parallel.
Again, subscripts are generally used to indicate which qudits are acted to be on by $U$ and which qubits are `controls' for the gate (the ordering of control qubits is immaterial).  Ranges of qudits may be specified for large unitaries, \eg{} $\Lambda^4_{[1..4]}(U_{[5..7]})$ would denote applying a three-qudit unitary $U$ across sites numbered 5 through 7, conditional on the four qubits in sites 1 through 4 being in the state $\ket{1111}$.

\cleardoublepage

\chapter{Approach to Complexity}  \label{chap:approach}

This chapter presents a little background to quantum complexity~: it is by no means a complete introduction.  It can readily be skipped by the reader already familiar with the field.

\section{Ontology}  \label{sect:ontology}

There is famously much diversity in describing what various quantum statements might really be saying about the world.  Bearing in mind the basic notions of statistical mechanics and of Everett's interpretation (see \cite{lit:Everett}), we will begin by providing a brief sketch of \emph{how one might choose to understand the ontology of quantum processes}, with the hope that this may help provide clarity for some of the phraseology used later in this dissertation.  (Our goal in this opening section is to `tell one story about reality', rather than contrast the various options.)

\subsection{A model for dynamics}  \label{sect:model-dynamics}

Let TIME be modelled as a real number line, parameterised by $t$.
Consider a \emph{state vector} $\Psi = \Psi(t)$, a mathematical function of TIME, whose role is to encapsulate the total description of all that may be said about a system; that is, a complete objective description of a system.
One may wish to ask questions about how the system \emph{evolves} with time, and this line of thinking we refer to as DYNAMICS.
\begin{eqnarray*}
   \Psi' &=& \frac{d\Psi}{dt}.
\end{eqnarray*}
There is an underlying anticipation that this model should provide a way of approximating reality, using smooth functions for the state vector.  Of course there are plenty of reasons to think that the model is far too \naive{} to capture anything like the `whole' of physics, not least because the model for TIME here is entirely non-relativistic.
These concerns aside, we see that the immediate \emph{ontological} problem intrinsic to the model arises from the notion of \emph{linearity}.  That is, since (real) analysis makes it clear that the derivative operator is linear, we have it axiomatically that
\begin{eqnarray*}
   \left( \Psi + \Phi \right)' &=& \Psi' + \Phi',
\end{eqnarray*}
whatever $\Psi$ and $\Phi$ might really be; but the model so far says nothing about the meaning of the $+$ signs on either side of this equation.  
Put another way, there is nothing mysterious about why DYNAMICS should be linear, but linearity itself does not come automatically equipped with a physical (semantic) interpretation.

If the role of $\Psi$ is taken to be the encapsulation of all that describes a system at a given time, and if $\Phi$ is supposed to be of the same category, then we see that these symbols are indicating potential different \emph{possible} configurations of the same system.  Then it is most natural to infer that $\Psi + \Phi$ might encode a state of the system `being $\Psi$ and/or also being $\Phi$', and so we might associate linearity with the intuitive (yet quantitative) notion of \emph{probability}.

The underlying philosophical notions of Probability Theory are notoriously difficult to define rigrously (\cf \cite{lit:Fuchsthesis}).  With the `Frequentists', we could attempt to reject all notions of probability that do not ultimately depend upon the counting of ontologically real conditions or events.
Alternatively, with the `Bayesians', we could adopt Probability Theory as a means of describing relationships between prior and posterior subjective states, according to experimental data.  A problem with the first approach is that it is not especially powerful, since it is limited in scope to cases where there is something definite to count.  A problem with the second approach is that it is hard to find a meaning for a prior distribution, and even then, the resulting posterior distribution depends heavily on the choice of experiments made and data collected.
Instead, we could choose to overlook the precise meaning of probability for the time being, identifying it simply as \emph{propensity} in the modern sense (\cf~\cite{lit:Alb06}), and simply ask about the possible forms of solutions to the equation above.

\subsection{Classical computation}  \label{sect:classical-comp}

For a \emph{classical} model, one would take $\Psi$ and $\Phi$ to be stochastic vectors (assumed finite-dimensional for this discussion) and, on discretising TIME, take the set of allowable linear transformations to be the stochastic linear maps.
Then the allowable linear combinations would be the convex ones, and the interpretation of the state vectors themselves would be as belonging to a vector space having a basis that constitutes the possible `actual' (\ie{} objective) configurations of the system; we'll call this the \emph{computational basis}.  The convex combination of `actual' states then merely denotes a probabilistic mix (a derived concept subject to whatever we later decide `probability' means).

The dynamics of this kind of computation come with no guarantee of time symmetry, since many stochastic maps do not possess an inverse.  Thus there is the possibility of `computational heating' of a state vector (increasing Shannon entropy, perhaps by setting a bit of memory to be random, for example) or `computational cooling' of a state vector (\eg{} perhaps by resetting a bit of memory to 0).  
This provides something of a backdrop for classical randomised computation.  Indeed, the classical theory of Turing machines requires little more than this for an ontological framework.  (By restricting to the rational field instead of the real field, we can even make a \emph{purely} Frequentist interpretation, because then one can normalise the vectors onto the integer lattice and speak reasonably unambiguously about \emph{counting computational paths}, interpreting probabilities in the ordinary fashion.)

There is also a notion of \emph{Reduction}, which is a subjective operation to be applied on a state vector, to reduce it (stochastically) to a computational basis vector within its support, that is, to choose to `realise' one of the possibilities for the state. 
This operation is subjective not least because it is non-linear (and therefore cannot be part of DYNAMICS), but also because its meaning depends on what we decide \emph{probability} really is.  One can think of \emph{Reduction} (also called `state collapse' in the quantum world) as a spontaneous change in the scope of what is actually being modelled, rather than a change or evolution of objective state itself, as when, for example, one chooses to consider a single possibility instead of considering many possibilities at the same time.
Yet the most important aspect of the classical model---as contrasted with the quantum alternative---is that this subjective Reduction effectively commutes with the objective DYNAMICS of the model.  In symbols, if $R$ denotes the subjective stochastic choosing of a computational basis vector, and $S$ denotes any objective stochastic transformation, then
\begin{eqnarray*}
  R( S \cdot \Psi ) &=& R( S \cdot R( \Psi ) ).
\end{eqnarray*}
The implication of this is that it makes no difference to the meaning of the computation how we understand $R$, because the action of $R$ can always be pushed through to the end of the computational procedure, and therefore effectively ignored.  This seems to be commensurate with the common understanding of what randomness really is, \ie{} a purely subjective uncertainty that can be effectively ignored until required.
The relevant maxim here is, ``Classical computation paths do not interfere.''

\subsection{Quantum computation}  \label{sect:quantum-comp}

The \emph{quantum} model, as we describe it, takes a different approach. 
For reasons of quantification, it is still appropriate to think of vector spaces with a metric.  This time, we expect $\Psi$ and $\Phi$ to be normalised vectors, \ie{} having unit Euclidean length.  The set of transformations that preserves this property is constituted by the orthogonal group, or more generally the unitary group.  
The symbol $U$ denotes for us an arbitrary unitary transform, and replaces the stochastic transform $S$ of classical dynamics.
Now there is time-symmetry, in the sense that, being a group, every transformation possesses a valid inverse.  
To ensure that the Lie algebra is algebraically closed, we may as well take the underlying field to be complex, in which case, the appropriate kind of (positive-definite) metric is the Hermitian inner product.  With respect to this inner product, the computational basis is taken to be orthonormal.  The inner product between an `object' vector and a `reference' unit vector is called an \emph{amplitude} (and we speak of the amplitude of the object in the direction of the reference).

The `specialness' of the computational basis is no longer geometrically significant~: without reference to a specific physical model, there is much symmetry within the Lie group of transformations, and so no particular reason to prefer one orthonormal basis over another.
And so, within a closed system, there is no notion of `heating' or `cooling', because the (von Neumann) entropy of a state vector is always zero.  The full unitary group acts transitively on the projective space, and so no one state is intrinsically different from any other.  Notions of entropy, entanglement, and mixedness, do not properly arise until we consider dividing a system into two parts, or consider the meta-system of `system-plus-environment'.  

By taking a tensor decomposition of a finite-dimensional system into two (or more) parts, it is well understood that notions of entaglement (Everett's ``relative states'', \cite{lit:Everett}) are possible.  Furthermore, by imposing limits on the allowable dynamics \emph{across} the two parts, and by considering one's computation to take place on the smaller of the two parts, it is possible to recreate all of the features of the classical model within the smaller component, including heating, cooling, and irreversibility.  This phenomenon of \emph{decoherence}\footnotemark{} goes some way to justify the claim that classical mechanics is a `subset' of the more complete quantum mechanics.
\footnotetext{The term \emph{decoherence} is used a little differently in so-called `non-Everettian' interpretations of quantum mechanics.}
Decoherence is both necessary for quantum computation (for example to enable a system to be cooled into an initial starting state), and yet also problematic (because it can inhibit the `quantum' features of the computation, if not precisely handled).

The notion of Reduction for a quantum system turns out to be closely related to the notion of decoherence.  
Everett derives the Born rule from simple considerations about normalisation; this being the only stochastic measure valid for all states.  This rule tells us that the action of Reduction on a quantum state consists in the stochastic choice of a computational basis vector according to the square of the modulus of the amplitude of that state in the direction of the choice.
\begin{eqnarray*}
  R(\Psi) &\rightarrow& \mathbf{b}_j  \mbox{~~~~w.p.~~~~} 
  \left| \left< \Psi, \mathbf{b}_j \right> \right|^2.
\end{eqnarray*}
Everett writes \cite{lit:Everett}, ``In other words, pure [unitary] wave mechanics, without any initial probability assertions, leads to all the probability concepts of the familiar formalism.''  
As with the classical case, this Reduction operation is not linear, and so not a part of DYNAMICS, \ie{} not to be taken as intrinsically objective.
The sometimes counterintuitive aspect of quantum information processing could then be said to derive from the fact that this Reduction does \emph{not} commute with the set of allowable transformations, in sharp contrast to the more familiar classical theory. 
This means that one is not at liberty to apply this simplifying Reduction at arbitrary stages of processing.  Indeed, \emph{a priori} one is not at liberty to apply this simplifying Reduction \emph{at all}, without a clear justification.  

The justification we need, at least for making sense of ideas within the field of computation and algorithmics, comes from \emph{measurement}.  Besides initial input preparation, measurement is the main place where decoherence becomes an essential physical feature of the description of quantum information processing.  \emph{Complete measurement} describes the action of choosing an orthonormal basis for a target system and then applying an entangling operation between that system, with respect to that basis, and a suitably prepared external system (a measurement apparatus), before separating the two systems to prevent further interaction.
The entangling operation is simply a unitary map that will have the effect of simulating decoherence of the target system in the required basis when it is subsequently considered separately from the measurement apparatus.  To be precise, measurement does more than introduce decoherence, because it also records between the target system and the measurement device a quantum correlation pertaining to the information which has been \emph{ipso facto} measured.  

The important point here is that measurement \emph{supervenes} Reduction, so that (in symbols) if $M$ denotes measurement (with respect to an unrepresented measurement apparatus) and $R$ denotes subjective Reduction, then
\begin{eqnarray*}
  M \cdot R(\Psi) = M \cdot \Psi = R( M \cdot \Psi ).
\end{eqnarray*}
Note that $M$ is objective; it is perfectly linear (indeed unitary) on the space of the system \emph{tensored} with the measurement apparatus, though it acts non-linearly if considered on the target system alone.  Thus, we use external measuring systems---and the decoherence they bring---in order to give a physical and operational meaning to the otherwise vague notion of Reduction.  When measurement happens, states `collapse', whether we like it or not.

If a quantum information process begins with a suitably decoherent initialisation procedure (such as the preparation of an array of qubits into separate unentangled computational basis states) and ends with a (complete) measurement, then it makes perfect sense to regard the Reduction process $R$ as happening at the beginning and end of the computation, where the system is `apparently classical'.  But since $R$ and $U$ do not necessarily commute, we must avoid imposing $R$ at intermediate points within the computation, in between unitary dynamics.  Pragmatically this means that we must avoid `accidental measurement', or indeed any undesired decoherence of `important' data, throughout the lifetime of the computational process.  This is what serves to distinguish a quantum computer from a classical one.  (See the lecture transcripts at \cite{lit:Democritus} for a gentle---yet remarkably effective---introduction to this kind of abstract approach.)

This provides enough of an ontological framework for the definition and analysis of quantum Turing machines and the various other conceptual devices one comes across in the theory of quantum computational complexity theory, without direct recourse to the physics of quantum mechanics itself.

\subsection{Refining the model}  \label{sect:refine-model}

\subsubsection*{Dimensions}

For convenience, we have been restricting attention to finite dimensional vector spaces, and will continue to do so for studying algorithmics, for the most part.  
On the few occasions where infinite dimensional vector spaces are more appropriate, a sufficient mathematical treatment will be given.

\subsubsection*{Time}

As well as thinking of TIME as a real line, and DYNAMICS as proceeding via the unitary Lie algebra of Hamiltonian actions on the vector space, we have also found it often convenient to \emph{discretise} TIME, working with (a countable subgroup of) the Lie group of unitary gates as though they were `atomic' transformations.  The study of algorithmics uses both notions, continuous and discrete, usually depending upon assumptions about the underlying physical architecture.  

None of our treatments uses a \emph{generally covariant} treatment of TIME as an aspect of SPACETIME, since relativistic effects are considered unlikely to be of significant philosophical relevance to complexity theories underpinned by a pragmatic control theory, and such notions require a distinctly deeper ontology to make sense.  See \eg{} \cite{book:Penrose} for a thoroughgoing guide to the geometric principles involved in `quantizing gravity'.

Mostly we shall prefer the discrete picture for TIME, since it is well adapted to discussing both classical and quantum models, whereas the continuous picture does not apply so well in the classical case.  In fact, this goes some way to illustrate how the classical model is simply not `native' to the set of assumptions that we began with in \S\ref{sect:model-dynamics}.
In the discrete picture, an `atomic' dynamic component or evolution is called a \emph{gate}, for both classical and quantum computing.

\subsubsection*{Computational paths}

It is usual within the study of classical algorithms to speak of computational paths, as alluded to previously.  Whenever a discrete time model is employed, one may understand a \emph{computational path}, in a counterfactual sense, to be the series of states that would be followed by the DYNAMICS of the computation were Reductions (with respect to the computational basis) to be made before and after each gate.  (One sometimes speaks of the Universe ``splitting into many worlds'' in this context, though this is really an artefact of the subjective Reduction process.)
For quantum algorithmics, the non-commutativity of gates with Reductions is tantamount to the maxim, ``Quantum computational paths can interfere.''

For quantum computing, it is more appropriate to regard a computational path as tracking not only the computational basis vector `realised' (counterfactually) at each point in (discrete) time, but also the amplitude that the objective state vector holds in that direction at that time~: in general, both of these kinds of information are relevant to a computational process.  The square of the modulus of the amplitude then provides the path with its own `weight', and there is also a \emph{phase}, which is the argument of the amplitude.  Significantly, phases may be negative as well as positive (and if we use an algebraically closed field, we may take them to be complex also).  Then the `interference' between computational paths derives precisely from the fact that when a set of paths is regarded together for Reduction (\ie{} when a measurement is made), it is the linear combination of paths terminating with the same computational basis vector that determines the probabilities relevant to the stochastic choice of an `output'.  In other words, we consider that there is no canonical `actual history' to a particular computation~: nothing of the sort, ``This is \emph{the path} which the computation took.''  Rather, two paths with the same final computational basis state will \emph{constructively interfere} or \emph{destructively interfere} according as to whether they have the same phase (are `in phase') or have opposite phases (are `out of phase'). 

Thus, computational paths have a distinctly more subjective flavour within a quantum process than within a classical one, arising from the fact that the computational basis is arbitrary within a quantum process, rather than a part of the objective description within a classical one.  Nonetheless, the concept has remained firmly entrenched within the conceptual framework of quantum algorithmics, and has its uses within some of the non-physical definitions in the field of algorithmic complexity.

\section{Complexity}  \label{sect:complexity}

What counts as (quantum) information, and how do we decide whether the processing that it has been subject to is `quantum'? 
To date, no truly convincing quantum computer of significant computational power has been presented, but many physical experiments have shed light on what quantum information processing might mean.

The popular method for giving quantitative rigour to the various notions of quantum information processing involves \emph{asymptotic computational complexity analysis}, which involves finding upper- and lower bounds on the resource requirements of certain algorithmic tasks (usually classically defined) in the asymptotic limit of arbitrarily large problem instances, when certain constraints apply.  Resource requirements can include a range of parameterisable constraints, most notably TIME and SPACE, in some sense.  It is appropriate that conditions for algorithmic tasks may also include certain non-physical constraints, such as quantified bounds on success probability, or costed (and well-defined) oracular access to particularly relevant mathematical functions.  Much of the literature on quantum algorithmic complexity derives from similar notions and results from the classical theory of algorithmic complexity, and many of the notions carry over (`quantize') very naturally.  

In this section, we briefly recall some of the various different elements and notions that will be useful in the forthcoming discussion.  
In particular, some mention is made of the Turing machine model and the circuit model, as these concepts are referred to throughout the dissertation; but nowhere do we use the \emph{random access memory} model, this latter being more relevant to the kinds of highly complex `large-SPACE' algorithms that are not the subject of this study. 
For more background on the classical concepts, we recommend reference to \cite{book:Papa}, and for the quantum ones, see \cite{book:NandC}.

\subsection{Architectures}  \label{sect:architectures}

Many architectures have been proposed for the construction of a quantum computer.  The earliest algorithms were considered in a model based on networks of small unitary gates, but recent years have seen ideas like the one-way quantum computer in which the non-unitary acts of measurement play a key role for data processing, or adiabatic computing in which continuous time dynamics are used.   
Studies showing how different quantum computational models can simulate each other are valuable in constructing universal paradigms.  They also provide perhaps the clearest expression of the primitives in each computational model that are responsible for generating computational power seemingly `stronger' than that of classical computation. 
Since the major obstacles against useful quantum computation are considered likely (for a long time) to be engineering difficulties in implementation, a further incentive for such alternatives in underlying architecture is to generate ideas of how to adapt the computational model to various different `limited' sets of primitives.
This thesis investigates some particular paradigms for architectures for quantum computing devices, exploring both \emph{universal computation} and \emph{limited computation}.  The emphasis is always on understanding how a particular limitation or restriction of some aspect of the computational process can (or does not) inhibit some particular kind of operational algorithmic process.

\subsection{Computational tasks}  \label{sect:comp-tasks}

\subsubsection*{Decision languages}

Usually it is possible to examine much about the computational power of some computing paradigm by asking about the complexity classes of decision languages associated to it.  In simple terms, a decision language is just a subset of some `simple' countably infinite set (usually the positive integers or the finite-length bitstrings) that can be `decided' by some operational (or more fanciful) means.
Note that the theory of \emph{computational complexity}---dealing with the resources needed to address a computational task---differs from the theory of \emph{recursion}---dealing with whether a task would be `possible' if resources were unconstrained.  Within the former theory, we will always have in mind \emph{some} pragmatic limit on some resource, such that the possibility of a machine's never halting is of absolutely no consequence.

By way of example, we recall a few common complexity classes (\cf{} \cite{book:Papa})~:
\begin{itemize}
  \item
$\cL \in \P$ if there exists a (deterministic) machine that accepts $x$ within \textbf{polynomial time} (\ie{} $\mbox{size}(x)^{O(1)}$) whenever $x \in \cL$, but which rejects those $x$ not in $\cL$.  Informally, $\P$ is often considered to be the class of languages ``efficiently decided''.

  \item
$\cL \in \L$ if there exists a (deterministic) machine that accepts $x$ within \textbf{logarithmic space} ($O(\log(\mbox{size}(x)))$) whenever $x \in \cL$, but which rejects those $x$ not in $\cL$.  (\emph{Space} here refers to the amount of computational storage/workspace required by the machine, not the space required to submit the actual input $x$.)

  \item
$\cL \in \ParL$ if there exists a (randomized, classical) machine that accepts $x$ within \textbf{logarithmic space} ($O(\log(\mbox{size}(x)))$) on an \emph{odd} number of computational paths whenever $x \in \cL$, but which accepts $x$ on an \emph{even} number of paths when it is not in $\cL$.  (Computational paths here are required to be all of equal length for a given input string $x$.)

  \item
$\cL \in \NP$ if and only if there exists a \textbf{nondeterministic} machine that accepts $x$ with some non-zero probability, within \textbf{polynomial time}, whenever $x \in \cL$.  This notion is given an operational meaning of sorts (and unambiguouly generalised in other contexts) by observing that it is equivalent to saying that for some other language $\cL' \in \P$, the item $x$ lies in $\cL$ if and only if there is some $w$ such that the concatenation item $(x,w)$ lies in $\cL'$ (and $w$ is then called the \emph{witness} to that fact).  To ensure that reductions compose, $\mbox{size}(w)$ will need to be polynomially bounded in $\mbox{size}(x)$.  Informally, $\NP$ is often considered to be the class of languages ``efficiently verified''.

  \item
$\cL \in \PP$ if there is a \textbf{probabilistic} machine that accepts $x$ with probability strictly greater than $\frac12$, within \textbf{polynomial time}, whenever $x \in \cL$, \etc{}  This class is again \emph{syntactic} in the sense that to specify a well-formed probabilistic machine is to specify a $\PP$ decision language; but it is not \emph{operational} in the sense that there is no particular way to make use of that machine to form an actual real-world decision, because the probabilities in question might turn out to be exponentially close to the $\frac12$ threshold.

  \item
$\cL \in \BPP$ if there is a probabilistic machine that accepts $x$ with probability at least $\frac34$, within polynomial time, whenever $x \in \cL$, but rejects $x$ with probability at least $\frac34$ (assuming a polynomial time bound), whenever $x \not\in \cL$.  This class is called \emph{semantic} (as opposed to \emph{syntactic}) because its definition does not make clear exactly when an arbitrary machine might happen to display the required \textbf{probability bounds} consistently, for all $x$.  For example, the existence of even a single $x$ with an acceptance probability strictly between $\frac14$ and $\frac34$ would prevent the machine in question from issuing a $\BPP$ decision language under this definition.  But the class $\BPP$ is nonetheless \emph{operational} in flavour, because by parallel or sequential repetition of the computation, when the probabilities are promised to be bounded away from $\frac12$ as described, the threshold value of $\frac34$ can be boosted to lie exponentially close to unity, still all within \textbf{polynomial time}, at which point it becomes pragmatically \emph{beyond doubt} whether or not $x$ lies within $\cL$.  \emph{Cf.} \S\ref{sect:op}.

  \item
$\cL \in \BQP$ if there is a \textbf{quantum} machine for $\cL$ that accepts or rejects $x$ with the same $\frac14$ \emph{versus} $\frac34$ probability \textbf{bounds} as for $\BPP$, again running in \textbf{polynomial time}.
Again, this is operationally meaningful independently of the details of the ontology used to interpret the meaning of probability in a quantum context, using the same ``Chernoff bounds'' argument as for $\BPP$.  
\end{itemize}

\subsubsection*{Interactive protocols}

There are tasks more general for computation than deciding whether $x \in \cL$ for a decision language, or computing a function~: simply taking a sample from a particular probability distribution constitutes a computation of sorts.  
Such tasks can sometimes be given operational roles by embedding them within multi-party \emph{protocols}.  
The complexity of such a protocol may be measured not only in terms of the computational resources required by each party, but also by the communication resources required for signalling between the parties, and intermediate storage requirements.  
The main example used in Chapter~\ref{chap:IQP} is provided by an interactive two-party protocol, rather than by a single-party algorithm.

\subsection{Turing machines}  \label{sect:Turing-machines}

For classical computing, Turing provided a rigorous foundation by making precise definitions for the kinds of machines that might be considered.  His machines are sufficiently general as to be able to simulate many other proposed paradigms.  We next sketch some of the ideas often used when thinking about Turing machines, although not all of these ideas appear in Turing's original considerations.  Equivalence between different models depends on the notion of algorithmic \emph{reduction} (expressing one task in terms of another), which we will also come to shortly.
For our purposes, a Turing machine will be an essentially classical device, having a finite (constant) number of internal states, and access to a finite (constant) number of `tapes'.  Each tape is to be thought of as a one-dimensional array, usually of bits, with certain restrictions governing the dynamics relating the tapes and the internal state of the machine.  On each tape there is to be a pointer, and `access' to the tape is via the pointer.

The usual idea for Turing machines is that they process `eager data', that is, input data which are all present at the time the machine is activated.  There are extensions in \emph{Domain Theory} for more general concepts of data processing, but these will not be relevant to the present thesis.  Furthermore, we shall be largely glossing over the important and thorny issue of error-correction, studying instead the idealised `perfect' instantiations of computing machines.

We will take an \emph{input tape} to be a read-only tape of bits.  There are some contexts where it is more preferable to allow \emph{algorithmic input} to consist of quantum data, especially where multi-party computations are being considered and quantum communication is allowed for.  But it will suffice for every topic of this dissertation to restrict algorithmic input and other communication always to be classical.  The length of the input tape (\ie{} size of the input) is usually denoted $n$.

There should be some convention for the tape so that its input bits are all contiguous and so that some sensible mechanism is allowed for to determine where the end of the input tape is located.  The input tape pointer starts at the beginning of the input tape.   In fact, general considerations of this kind apply to all tapes of a Turing machine.

Sometimes we allow for a \emph{random tape}, for a \emph{probabilistic} Turing machine, which is a read-once-read-only tape of arbitrarily long length, whose contents are set randomly when the machine commences computation.  This models a random number generator.

There is to be a \emph{work tape}, which is blank to begin with, but may be written to and read from multiple times.  The length of the work tape is usually taken to be some polynomial in $n$, but for some `smaller' computational classes (such as $\L$) it is interesting to consider work tapes whose length is limited to being logarithmic in $n$.  There are several different ways of extending this notion into the quantum realm, and usually it will be more convenient to select a specific description for the task in hand.  We recommend \cite{lit:Watthesis} as the definitive reference for space-bounded quantum computation. 

There is to be an \emph{output tape} onto which the results of computation can be written.  This tape of bits should be write-only, and is generally taken to be of arbitrary length.  Output is `achieved' when a machine halts, and we shall be studying the complexity of halting machines only.  For some computational tasks, such as deciding operationally-defined decision languages, only a single bit of output is required.  
For example, we could adopt a convention that if ``1'' is output within the time-bound then the machine is deemed to have ``accepted'' its input, the input being otherwise deemed ``rejected''.
For reductions in general, it is necessary to consider larger outputs, so that a Turing machine can act as a pre-processor or post-processor for another machine.

Sometimes we allow for an \emph{oracle tape}.  This enables the machine to have `black-box' access to some subroutine whose complexity we deliberately wish to place out of scope of analysis.  If the machine is attached to oracle $\cO$ and has data $z$ written on its oracle tape at a time when it calls its `oracle' function, then the contents of the oracle tape are to be replaced (albeit mysteriously) by the data $\cO(z)$, in unit time.  Such tapes are not, however, to be used as proxies for \emph{work tapes}, and care has to be taken when making rigorous definitions, to avoid hiding complexity in the oracular interface.  The most famous use for oracles is to separate complexity classes of decision languages which are otherwise inseparable by known analyses.  The most famous use for oracles in context of \emph{quantum} computational complexity is probably in establishing quadratic query separation (between lower bounds for classical access to the oracle and quantum access) as per Grover's algorithm (see \cite{lit:Grover96} and \cite{lit:BBBV97}).  In the quantum case, the oracle tape should constitute qubits which can be interacted with the work tape, and the oracle action should be defined carefully as a unitary action that degenerates to the classical oracle on input computational basis states.  As with other quantum generalisations, it is best to be specific whenever implementation details can make a significant difference to computational power.

\subsection{Algorithmic reduction}  \label{sect:algorithmic-reduction}

The notion of \emph{algorithmic reduction} of problems has to do with using one computing machine as a pre-processor, or oracle, for another.  Reduction is important for relating different complexity classes~: indeed, the most oft studied complexity classes tend to be the ones with suitable closure properties under reduction.  For example, it is easy to see (by composition of polynomials) that a Turing machine fitted with an oracle that decides a given language in $\P$ will not, in polynomial time, be able to compute anything that could not be computed, in polynomial time, by some (other) ordinary Turing machine not so equipped.  We write, for example, $\P^\cL$ to denote the analogue of $\P$ defined relative to the attachment of an oracle for deciding $\cL$.  If $\cL$ is itself in $\P$ then we have just seen that $\P^\cL = \P$.  More generally, when $\mathbf{A}^\cL = \mathbf{A}$ for all $\cL \in \mathbf{B}$ then we say that $\mathbf{B}$ is \emph{low} for $\mathbf{A}$.

\subsubsection*{Completeness}

There are some languages $\cL \in \P$ which have the property that $\L^\cL = \P$.  That is to say, there exists $\cL$, a `sufficiently complex' language in $\P$, that appending to a suitably designed ordinary Turing machine the `black-box' ability to decide $\cL$ immediately, enables that Turing machine to decide, in logarithmic space, the things an ordinary Turing machine would require polynomial time (and presumably polynomial space) for.  Such a language is said to be $\P$\textbf{-complete} with respect to logarithmic-space reduction.  Strictly speaking, the determination of \emph{completeness} requires not only a specification of the complexity of the preprocessing (in this case log-space processing) but also a specification of how many times, and with what adaptive control, the oracle calls are permitted to be made.  For the sake of brevity, we usually have in mind that the preprocessing will be log-space and/or poly-time, and a polynomial number of queries to the oracle are permitted, adaptively.  

Suppose $\cL \in \NP$ denotes an $\NP$\textbf{-complete} language with respect to poly-time reductions, so that $\P^\cL \supseteq \NP$.  Because examples of this form exist, it is appropriate and common practice to denote this fact with the notation $\P^\NP \supseteq \NP$.  Note that the $\NP$ appearing in the `index' here is effectively a placeholder for any $\NP$\textbf{-complete} decision language.  
(Note also that oracle use is sometimes employed not with a decision language but with an entire function, possibly probabilistic.  In the same manner, function classes can be used in the index as placeholders for particular complete functions from those classes.)

\subsubsection*{Simulation}

By using reductions together with an encoding of machine descriptions into bitstrings, Turing was able to introduce the concept of a \emph{universal} Turing machine, a concept of \emph{simulation} now entirely fundamental---indeed intuitive---to computer science.  The idea is that we can say that $\cU$ is a universal Turing machine with respect to some encoding $\cC$ if 
\begin{eqnarray*}
  \cT  &=&  \cC( t ), \\
  \cU(~ t, x ~) &=& \cT( x )
\end{eqnarray*}
whenever $t$ is a text describing a Turing machine $\cT$, and $x$ is a putative input string for $\cT$.  The `complexity' of the encoding $\cC$ itself is not so important, because it doesn't depend on any input string $x$, and so doesn't affect the asymptotics of any language or function being computed.
Thus $\cU$ is said to be capable of simulating $\cT$, because it can have effectively the same output behaviour as $\cT$, for each input $x$.  

In the case of \emph{efficient simulation}, the equality sign in the expression above is supposed to denote the fact that not only are the extrinsic machine outputs to match---$\cU$ `accepts' $(t,x)$ iff $\cT$ `accepts' $x$---but additionally that the consumption of resources is to match (\ie{} the space/time/query/randomness requirements of $\cU$ are on the same order, as a function of the size of $x$, as the corresponding requirements of $\cT$).
Thus, in specifying a computational paradigm, we are usually concerned with establishing some kind of universal machine that is capable of efficiently simulating a class of machines, with careful accounting being made of the different resources that are implicitly required.

\subsubsection*{Universality}

A machine is said to be \emph{universal for $\BQP$} if it can be used to decide any $\BQP$ language, with bounded probability on correctness of decisions, within polynomial time, and in this sense efficiently simulate a quantum version of a Turing machine.  
The (theoretic) existence of quantum Turing machines is proven in \cite{lit:BV97}. 
This is not, of course, the most powerful form of efficient simulation that we could ask of a quantum computer.  For example, being universal for $\BQP$ does not in any way guarantee that one can `manufacture the same states' in polynomial time that are `manufactured' in polynomial time by some other quantum computing architecture.  (In fact, it is a philosophically thorny issue to determine what is meant by the concept of `same state' across two different architectures, when no natural isomorphism between state spaces need exist, and when one's ontology need not even admit the existence of quantum states as objectively real.)
Instead, we can make various definitions of \emph{universality for quantum computation} by asking for a device that can efficiently simulate any other quantum device within the context of a multiparty interactive protocol, where the interfaces on such protocols are adequately specified as part of the concept of universality.  For example, an interface might limit the exchange of data to being purely classical, or it might allow for quantum data in essentially any form, or it might require that the quantum data be encoded onto two spatially colocated bosonic modes of an optical fibre, \etc{}  As before, it is necessary that the overhead of simulation in terms of resource consumption not be too large, and so the study of simulations involves the continual audit of many aspects~: physical resources (time, space, \etc); non-operational resources (nondeterminism, oracles, \etc); encodings (how $\cU$ implements $\cT$ via $\cC$; what control signals pass from software to hardware); and interfaces (what form the inputs and outputs are to take between rounds of a protocol).

\section{Circuits}  \label{sect:circuits}

Quantum circuits are a particularly good way of putting the theory of quantum computation on a mathematically rigorous footing, sometimes preferable to quantum Turing machines, for example.  The computational complexity classes are, by and large, unaffected by which paradigm one adopts, yet it is often considered more natural to work with circuits as the basic constructs.

\subsection{Classically described quantum circuits}  \label{sect:cdqc}

It is convenient to focus our discussion on quantum circuitry on \emph{qubits} (two-level systems), though in practice, circuits can be defined on larger systems or \emph{registers}.  The standard idea \cite{book:NandC} is to regard a circuit of gates acting on qubits as a device for mapping (pure, $2^n$-dimensional) quantum states onto (pure, $2^n$-dimensional) quantum states.  Thus circuits can be composed, and likewise deconstructed into their individual gate constituents.
Their deconstruction should involve gates drawn from a finite (or simply-characterised) \emph{alphabet} of possibile gates, and there are to be specific rules for using them to construct complexity classes~: the most important rule being that a quantum circuit must be `handled' via its fully explicit classical description in terms of its deconstruction into gates.

\subsubsection*{Quantum languages from classical Turing machines} 

Ultimately, the decision languages we study still arise from particular Turing machines, even when the circuit model is used.  For example, the usual approach to defining $\BQP$ would be to take a \emph{classical} Turing machine, $\cT$, which, on receiving the unary input $1^n$, outputs the classical \emph{explicit} description of a quantum circuit, $\cC_n = \cC(\cT(1^n))$.  The machine must be bounded in the resources it uses, so that the (family of) circuits thus produced can be described as \emph{uniform}.  Throughout this dissertation, we adopt the convention that \emph{uniformity} implies a logarithmic-space bound for the pre-processing machine's operation.  

Then we can decide whether a given (classical) bitstring $x$ is in $\cL = \cL(\cT)$ by inputting the quantum state $\ket{x}\ket{0}$ (in the computational basis) into circuit $\cC_{\mbox{\small size}(x)+a}$ ---where $a$ is a prescribed polynomial function of size$(x)$ describing the circuit's \emph{ancilla} requirement---and measuring the first bit of the output in the computational basis.  Provided there is the usual semantic guarantee, as with the definition of $\BPP$, that the measurement outcome is biased one way or the other with a significant (\ie{} non-negligible) bias, then the direction of this bias is (in theory) tomographically accessible within polynomial time and space; therefore it can be said to indicate \emph{operationally} whether or not $x \in \cL$.  When this guarantee is present, we say that the Turing machine $\cT$ issues the language $\cL \in \BQP$ (see~\S\ref{sect:comp-tasks}), via a uniform family of circuits.

There are a few caveats to make clear.  First of all, the output of the classical Turing machine $\cT$ should be an explicit description of the quantum circuit to be implemented, so that no complexity is hidden in this interface, and so that since the Turing machine was limited by logarithmic space, and hence polynomial time, we can be sure that the rendering of the circuit on state $\ket{x}$ ought theoretically be possible within polynomial time and space.  Secondly, the input and output of quantum information are here described explicitly in a computational basis, again to prevent those interfaces from encoding complexity which would make the definition sensitive to changes in the details.  
That said, it is worth observing that since this definition is so close to the definition used for defining uniform classical circuits for $\BPP$ (see \eg{}~\cite{book:Papa}), we can immediately use the classical theory to see that the definition is entirely stable under many natural changes to the definitions.
Because quantum circuits can easily simulate classical ones, provided only that the gate set for the quantum circuit is capable of simulating a finite gate set that is classically universal, we can encode much of the complexity of the classical Turing machine directly into the quantum circuitry.  Therefore the definition of $\BQP$ remains stable even if we allow the classical Turing machine more space, but still with a polynomial time bound.  Likewise, if we add to the circuit a measurement of \emph{all} the qubits, this output to be processed classically by another polynomial-time-bound Turing machine, the class definition remains the same.  All of this `interface stability' is well known and documented in the literature (see \eg{} \cite{lit:Watthesis}), but is especially key to the perspective taken in Chapter~\ref{chap:FH}.

\subsubsection*{Reversibility}

One difference between the way in which classical circuits are usually constructed and the `standard' way (presented above) for handling quantum circuits lies in the detail of how \emph{space} (\ie{} memory) is managed.  Classical circuits are usually presented with ancill\ae{} being brought in as necessary and then ditched after use, whereas quantum circuits are usually presented with all ancill\ae{} `declared' up front.  
Because of this, quantum gates are usually taken to be \emph{automorphisms} (unitary transforms) on unitary spaces (finite-dimensional Hilbert spaces), rather than more general quantum operators.  Perhaps the reason for this trend has to do with a desire to avoid having to process \emph{mixed} (non-pure, entropic) quantum states within a circuit, at least for the basic complexity definitions.  A change to allow the more `dynamic' use of ancill\ae{} would not, of course, affect any of the complexity classes that we care to define, provided the rules of quantum mechanics are respected, in ensuring that only \emph{completely positive trace-preserving} maps be employed as gates.  But for our present purposes, such a change introduces unnecessary complexity, and will be avoided.

\subsection{Circuit interfaces}  \label{sect:circuit-interfaces}

\subsubsection*{Oracles in quantum circuits}

Quantum circuits are naturally associated with the unitary transforms that they induce, which are to be interfaced in a standard way when defining complexity classes of decision languages, and likewise of (Boolean) functions more generally.  In defining these classes, no use is being made of quantum data \emph{outside} of the quantum circuits.  
One way in which quantum data can be conceptually `interfaced out' of a circuit (besides classically as measurement results) is with a \emph{quantum oracle}, the analogue of the kind of black-box subroutine used in classical complexity analysis.  The most common way in which these not-necessarily-operational `devices' are used again relies on the idea of quantum processing naturally extending classical processing~: a quantum oracle in the circuit model is generally taken to be (for example) a gate that acts on $m+n$ qubits and maps computational states $\ket{\x}\ket{\y}$ to $\ket{\x}\ket{\y + f(\x)}$, where $f : \FF_2^m \rightarrow \FF_2^n$ is a Boolean function, and so the gate is unitary.  The use of oracles of this kind enables comparison of quantum and classical complexity classes by relativisation, and has a natural operational interpretation in context of algorithms such as Grover's celebrated quadratic speed-up of computational exhaustion \cite{book:NandC,lit:Grover96,lit:BBBV97}.

\subsection{Universal gate-sets}  \label{sect:universal-gatesets}

It is well known that the two-qubit C-Not gate, together with the single-qubit Pauli gates, the single qubit Hadamard gate, and the single-qubit $\pi/8$ rotation are universal for quantum computing (\cf{}~\cite{lit:DBE95}).  By this we mean that for any given $c$-qubit unitary gate $G$, for any real $\eps$, one can construct a composite approximation to $G$ (up to global phase) using $\log(1/\eps)^{O(1)}$ gates and ancill\ae{}, that is within $\eps$ of $G$ under some standard metric such as trace-distance.  As an immediate corollary, in the limit of $\eps \rightarrow 0$, the entire group of (special) unitaries (for any given fixed circuit size above a fixed lower limit) is constructible.
Indeed, there is no need to include the stated poly-logarithmic convergence rate within the definition, since the Solovay-Kitaev theorem \cite{lit:Kit97,lit:DM0505} guarantees that if the group spanned by the gates is dense in the special unitary group then for constant gate size the simulation is efficient in this sense.
This notion of approximate simulation is slightly more general than asking for \emph{exact} reproduction of arbitrary $G$ with a constant circuit size, which is theoretically interesting but perhaps not so operationally (physically) meaningful.  

An important further generalisation of the notion of \emph{universal gate set} emphasises the key roles of simulation and reduction, rather than rendition of arbitrary elements from the whole of the special unitary group.  The observation~\cite{lit:Shi0205, lit:Aha0301} that any probability distribution efficiently producible using the universal gateset quoted above is also efficiently producible using just the 3-qubit Toffoli gate $(\Lambda^2 (X))$ together with the single-qubit Hadamard gate $(H)$ serves to show that construction of a group dense in the whole of the special unitary group is unnecessary for computational purposes.  It is readily shown that the span $\span{\Lambda^2 (X),H}$ is dense in the orthogonal group, when three or more qubits are used together with a single ancilla qubit. 
Thus we may consider restricting the design of quantum circuitry to use only these gates $\Lambda^2 (X)$ and $H$, without forfeiting universality.  This has the advantage of making the state space of a machine of a constant number of qubits more amenable to combinatoric analysis.  
See \cite{lit:Aha0301} for a fuller discussion of these issues.

\section{State of the Art}  \label{sect:stateofart}

\subsection{Open conjectures}  \label{sect:openconj}

At the time of writing, there are no known proofs for any of the following conjectures.  Nonetheless, it is convenient to adopt language that assumes these conjectures tentatively, using phrases such as ``sub-universal'' as a kind of shorthand for ``almost certainly not universal, unless a major conjecture be falsified.''
\begin{itemize}
  \item
$\L \not= \P$;
  \item
$\P \not= \BPP$;
  \item
$\P \not= \NP$;
  \item
$\BPP \not= \BQP$;
  \item
$\NP \not\subseteq \BQP$;
  \item
$\BQP \not= \PP$;
\end{itemize}
(See, \eg{} Thm 6.4 in \cite{lit:ADH97} for the construction $\BQP \subseteq \PP$.)
All comments about universality for these kinds of classes should be held in tension with the fact that we cannot \emph{definitively prove} that all these complexity classes of decision languages are not in fact equal.\footnotemark{}  It is even possible (as far as is known) that some of them may be independent from the standard sets of axioms used in the formal foundations of mathematics!
\footnotetext{These conjectures are not regarded on an `equal footing' by many researchers.  For example, many more people seem to believe the truth of the third than the truth of the second in the list above.}

\subsection{This dissertation}  \label{sect:dissertation}

Our approach is to look for constructions that force `artificial' limits and restrictions on the resources allowed within a computational paradigm, in order to see what kinds of structures are necessary to enable probabilistic algorithms to operate within such constraints.  

Chapter~\ref{chap:UCLC} deals with some paradigms universal for $\BQP$ whose interfaces are sufficiently constrained that their universality is somewhat surprising.
There we use quantum cellular automata with particular symmetries in their dynamics.  We show that even one-dimensional structures with very little control can be understood as viable architectures for quantum computers (in the absence of noise), giving two particularly interesting examples.

Chapter~\ref{chap:PMC} considers mixed quantum states and probability distributions.  We explore the ``one pure qubit'' model of computation in particular, and establish new relativisation results as well as a construction for solving $\ParL$-complete problems there.  We begin that Chapter with an abstract discussion of probability distributions and the idea of post-selection, so as to develop some of the conceptual tools that will be of use in the rest of the dissertation.

Chapter~\ref{chap:FH} bridges a gap between universality for $\BQP$ and certain sub-universal structures, studying the Fourier hierarchy of quantum complexity.
There we employ the circuit model of computing, which is far more commonly used than are direct quantum analogues of Turing machines.  We provide a description of how to use a so-called `Fourier Sampling Oracle' with a classical pre-processor and post-processor to implement solvers for some well-known number-theoretic problems (our solvers have novel features, even though efficient quantum solutions to these problems are by no means original to this work).  

Chapter~\ref{chap:IQP} bridges the gap between $\ParL$ and $\BQP$ in a different way, emphasising the role of \emph{inherent temporal structure} in a quantum process, using the Clifford-Diagonal hierarchy.  This gives rise to a novel quantum procedure---having apparently no classically efficient analogue---for performing the role of `Prover' in a certain two-party interactive proof game.  It is our hope that this algorithm can be appreciated as being `the simplest genuinely quantum algorithm'.  To that end, we provide some analysis of attempts to approximate it classically.

\cleardoublepage

\chapter{Universal Computing with Limited Control}
\label{chap:UCLC}

In this chapter, we describe some frameworks for universal quantum computation, \ie{} paradigms that allow for simulation of $\BQP$ in polynomial time, but where the `quantum memory' is laid out in a one-dimensional array of small quantum systems (qudits), and \emph{control} over those qudits is limited in some important regard.  These constraints take the form of certain spatial and temporal symmetries in the dynamics.  In \S\ref{sect:QCA} we consider \emph{quantum cellular automata}, where the constraints require that the `same processing' happens at every site, at every time-step.  In \S\ref{sect:spinChains} we consider an architecture based on \emph{spin chains}, where the constraints prevent almost all of the computer from interacting with the outside world.
In each case, a novel construction is provided.

The main purpose in each case is to show how highly symmetric systems that lack the possibility of local addressing can nonetheless perform powerful computation, if implemented without errors.  This principle has long been established for classical systems, and the more recent theory of quantum computing systems has shown this also to carry over to the quantum case; so the main technical contribution of this work is to show that it remains valid even for \emph{one-dimensional} quantum designs, and with the particular constraints that we consider.  The presumption motivating this purpose is that by enforcing various symmetries, both spatial and temporal, we potentially broaden the range of physical architectures on which one might consider implementing a computing paradigm, and by limiting the design dimension to one, we aspire to maximise flexibility for potential implementations.

\section{Quantum Cellular Automata}  \label{sect:QCA}

In this section\footnotemark{} we discuss the role of classical control in the context of reversible quantum cellular automata, giving a one-dimensional universal construction with single cell dimension 12.
\footnotetext{This section is largely taken from my 2006 publication with Torsten Franz and Reinhard Werner \cite{me:UPQCA}.}

\subsection{Overview}  \label{sect:QCA-overview}

Cellular automata are, broadly speaking, a way of doing computation whereby data are distributed across a computer that has much translational symmetry in its dynamics, so that every `site' of the computer is doing effectively the same kind of processing.  Perhaps the most famous example of a classical cellular automaton is Conway's Game of Life, whereby each site (cell) holds a single bit, and each such bit is modified based on the settings of the neighbouring bits.  To generalise this kind of idea to a quantum setting, one asks that the update rule for changing the `state' of a cell should have a unitary behaviour.  Now when we consider an infinite lattice of cells, it is hard to conceive of a single update rule as being an actual mapping from states to states, so it is more convenient to think of the rule in the so-called Heisenberg picture, whereby its action is understood on the algebra of \emph{quantum observables} rather than on Hilbert space vectors.  The structure theorem given in \cite{lit:SW04} shows that it is always possible to regard a single transition rule as being comprised of two time-slices of applications of a finite unitary map repeated in parallel, as shown in Fig.~\ref{fig:ccQCA}.

We consider within this section two different computational models, both of which are quantum cellular automata (QCAs), \ie, distributed systems of lattice cells with a spatially homogeneous discrete time dynamical evolution of strictly finite propagation speed.  These differ from the abstract notions of one-dimensional QCAs given in \cite{lit:Wat95}, which correspond with dynamics which may be unphysical or code for arbitrary complexity at the physical layer, not having any constructive local Hamiltonian representation.  The models we use have an explicit decomposition, being physical according to the definitions given in \cite{lit:SW04}.  (See also \cite{lit:vanDamthesis} for background and the general theory of quantum cellular automata, and \cite{lit:Wie0808} for a recent survey of QCAs.)

Our two models differ from each other in the way the program operates, or more precisely how the quantum part of the computer interacts with a classical controller, being somewhat analogous to the gate model and the Turing machine model respectively.
In the gate model, the classical controller has to be comparatively powerful~: on receiving the input, it will compile the program in a version adapted to the size of the input, and actually build a quantum circuit to run it. The flexibility of this model hence largely resides in the classical controller, and the quantum computer hardware is, so to speak, scrapped after each instance. 
In contrast, a classical universal Turing machine takes its flexibility from the possibility of writing both the program and the input data on its tape for initialization.  We can apply these ideas to running a quantum cellular automaton as a computer~: on the one hand, we can use a classical compiler to select a classically described sequence of operations each of which is a QCA time step in its own right.  Such a machine will be called a \emph{classically controlled QCA} (ccQCA).  On the other hand, we can insist that program and data are written into the system by the initial preparation, after which the machine runs autonomously for a certain number of steps, and with a fixed transition rule independent of the problem. The only role left for the classical controller is then final measurement to read out the result.
It is entirely possible that the absence of classical signalling from this second model (except at initialisation and readout), coupled with temporal translational symmetry, may prove to have the pragmatic value that an implementation can be more readily isolated from decoherence channels while it is `running' its program, thereby enabling lengthy computation without explicit error-correction.

We show constructively that these two ways of programming a QCA (see \S\ref{sect:QCA-general}) are computationally equivalent.  In the proof we use a structure theorem for cellular automata obtained in \cite{lit:SW04}. This theorem holds in any lattice dimension, and so do the ideas of our construction, but we stick to the one-dimensional case as it is sufficient for bounded error quantum probabilistic computation.  We then use this equivalence to build a universal autonomous QCA, with an explicitly given transition rule, where ``universal'' means that it simulates the gate model up to polynomial overhead.

The universality of a one-dimensional QCA may be seen as surprising, since recent research (\cite{lit:YS06}) has shown that a one-dimensional cluster-state computer is always simulable classically in polynomial time, and two dimensions are therefore necessary for that computing model to transcend $\BPP$.
More importantly, the practical importance of using just one dimension in the QCA lattice has been suggested by certain authors (\cite{lit:BB0401}) not for philosophical physical reasons but for practical engineering concerns, it being much easier in many cases to design equipment to interface with a low-dimensional structure.
Besides showing universality in one lattice dimension, our construction also employs a significantly smaller cell size than that of other similar machines discussed in the literature, \cite{lit:Rau05, lit:Vlasov}.

Here is a summary of the properties of the \emph{main} construction that we present~:
\begin{itemize}
  \item 
\textbf{Universal;}  a (physically reasonable) paradigm capable of simulating quantum circuits, with polynomial overhead in most reasonable measures.  
 \item
\textbf{Discrete space, infinite or unbounded;}  there is an infinite line (lattice) of cells (qudits), $c_i$, each cell associated with the algebra of a $d$-dimensional Hilbert space, where $d$ is some constant.
  \item 
\textbf{Discrete time;}  unital homomorphisms representing discrete steps, as opposed to a Hamiltonian description.  (Care must be taken with unitaries, now that the underlying Hilbert space is potentially infinite-dimensional.)
  \item 
\textbf{Non-adaptive;}  all `software' is encoded at the input stage, thereafter dynamics are completely fixed for the universal device.
  \item 
\textbf{Reversible;}  update rule for observables is given by a unital homomorphism, ${\mathcal T}$, on the (quasi-local) operator algebra, \ie{} an algebra homomorphism that transforms rank-1 Hermitian projectors to rank-1 Hermitian projectors, conformally preserving orthogonality.  (Physically, this can be understood as generating no entropy.)  
 \item
\textbf{One-dimensional;}  the rank of the lattice is 1, so that each cell has just two neighbouring cells, and the cell indices are integers.
  \item 
\textbf{Spatially symmetric;}  ${\mathcal T}$ commutes with all lattice translations $(c_i \mapsto c_{i+1})$.
  \item 
\textbf{Temporally symmetric;}  apart from initialisation and readout, the only dynamic is ${\mathcal T}$, repeated over and over.
  \item 
\textbf{Nearest-neighbour locality;}  if $H$ is an operator supported on cells belonging to $S$, then ${\mathcal T}(H)$ is supported on cells of $S$ and their nearest neighbours in the lattice.
\end{itemize}

Recent work \cite{lit:NW08} has shown that there is a one-dimensional QCA design in the \emph{continuous time} model that requires only ten levels per cell, rather than the 12 that we use.  It is still an open problem to establish tight bounds in either case.  Other aspects of the complexity of one-dimensional continuous systems, including the \emph{local Hamiltonian problem}, are discussed further in \cite{lit:AGIK07}.

\subsection{General construction techniques} \label{sect:QCA-general}

The description of a QCA is most readily given using the Heisenberg picture, which is to say that we describe evolutions by how they transform the $C^*$-algebra of local quantum observables for the system \cite{lit:SW04}.  This transformation must always have some spatial symmetry if it is to be called a QCA.

\subsubsection*{From circuit model to ccQCA}

\begin{definition}
  A \emph{classically controlled QCA (ccQCA)} is modelled as a list of unital homomorphisms of the observable algebra associated to an infinite line (lattice) of qudits.  The symmetry requirement is that there be some full-rank group of lattice translations, each element of which commutes with each homomorphism on the list.  Quantum data stored in the lattice is processed by the sequential application of these homomorphisms. 
\end{definition}

For universality with regard to $\L$-reductions, we look for there to be a log-space Turing machine that converts the description of an arbitrary quantum circuit (given in some standard explicit form) to a description of such a list of homomorphisms, so that the effect of the quantum circuit is emulated within the cells of the ccQCA as it works through applying the homomorphisms on the list.

Consider the circuit model of quantum computation wherein qubits are present in a one-dimensional lattice (also called a `band'), and any gate may act unitarily on just two neighbouring qubits.  Such models are seen to be $\BQP$-universal, when a sufficiently complex gate-set is admitted, \emph{e.g.} as exemplified in \cite{lit:Aha0301}. 
Then there are various direct ways of implementing such circuits as classically controlled QCAs.  For example, one could envisage increasing the cell size by a constant factor so that it can effectively represent two parallel bands, one (called the `data band') for encoding the qubits of a circuit, and one (called the `pointer band') for encoding a pointer, much like the `read/write head' of a Turing machine.  The transformations of the ccQCA could manipulate the location of the pointer and then use that pointer to break the spatial symmetry of the dynamics so that individual specific neighbouring data qubit pairs (encoded in the `data band') may be addressed, as required.  The data band and pointer band can of course be regarded as one single band, by interleaving their qubits; at the expense perhaps of having fewer of the translations of the lattice commute with the homomorphisms of the ccQCA.

\subsubsection*{From ccQCA to QCA}

\begin{definition}[\emph{Cf.} \cite{lit:SW04} Def 1]
  A QCA is modelled by a unital homomorphism $\cT$ of the algebra of observables on a lattice of qudit cells (Heisenberg picture).  $\cT$ must commute with all lattice translations.
\end{definition}

For a QCA to emulate a ccQCA, we look for there to be a log-space Turing machine that converts the list of homomorphisms associated to the ccQCA into a list of bits that can be interpreted as a `program' to be loaded into the cells of the QCA, alongside the data, at time $t=0$, so that after some polynomial number of applications of $\cT$ the `program' will have interacted with the `data' so as to emulate the desired transformation.  We allow for the possibility that the physical location of the `data' in the cells after the applications of $\cT$ may be different from its starting location, but naturally there ought to be no complexity hidden in this translation. 

The main conceptual tool for understanding the decomposition is the QCA structure theorem (Theorem 6 in \cite{lit:SW04}).  This theorem guarantees the existence of a \emph{Margolus decomposition~:} two finite unitary operations ($U_i$ and $V_i$) for each of the ccQCA transition rules which implement the time evolution by sequential application to non-overlapping neighbourhoods (as indicated in \mbox{Fig.~\ref{fig:ccQCA}}).  
This saves having to reason purely in terms of unital homomorphisms.  (Note that a \emph{single} finite unitary map will not generally suffice for a QCA homomorphism in any discrete model because it will have fixed eigenvalues---with algebraic multiplicities matching geometric multiplicities---and therefore be close to a unitary map having finite order, independent of the size of the computation.) 
The structure theorem applies to nearest-neighbour ccQCAs; so given an arbitrary ccQCA, one first needs to convert it into a ccQCA with nearest-neighbour interaction, which is always possible in a trivial fashion by merging cells and enlarging the dimension of the qudits that form the lattice.

\begin{SCfigure}
  \centering
     \includegraphics[width=75mm]{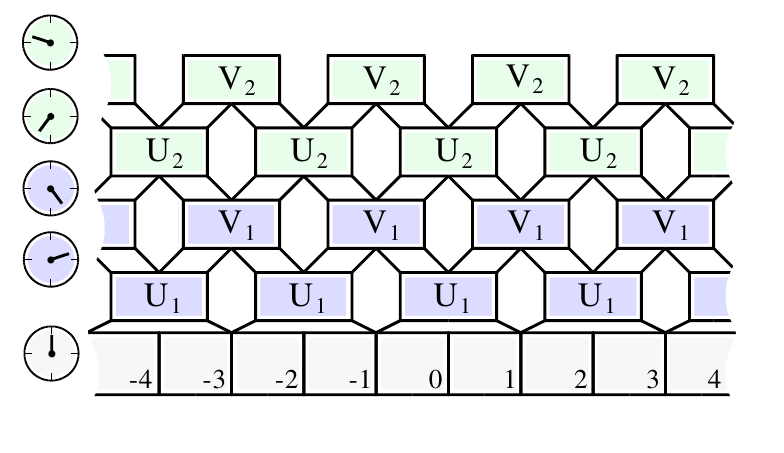}
     \caption{\label{fig:ccQCA} Two time-steps of the ccQCA. The clocks indicate the the time before and after application of the unitaries for comparison to Fig.~\ref{fig:QCA}.}
\end{SCfigure}

Consider an autonomous QCA that consists of a data band representing the one-dimensional lattice of the ccQCA being simulated and a program band containing information about the sequence of transition functions that the ccQCA would apply.  Let the ccQCA have access to $k$ different homomorphisms. Then the cell size of the program band is chosen to be $2k + 1$, enough to distinguish the $2k$ unitaries, allowing for an extra symbol representing the identity map.  Let the time evolution of the QCA be the product of a `shift step' shifting the program band two cells past the data band, followed by a `calculation step' performing the required unitary maps on pairs of cells of the data band, each controlled by the neighbouring contents of the program band.
As the program band moves past the data band, each data cell (qudit) undergoes the time evolution of the ccQCA being simulated, yet it should be noted that different time-steps in the ccQCA evolution are present at one time-step of the autonomous QCA (see \mbox{Fig.~\ref{fig:QCA}}).  In accordance with the definition of QCA, it is important that there arises no possibility of non-commuting unitaries operating on the same cell at any time; note that the unitaries $U_i$ and $V_i$ obtained from the QCA structure theorem work on different combinations of odd and even cells (see \mbox{Fig.~\ref{fig:ccQCA}}).  Therefore, to circumvent this possibility, one can design the autonomous QCA such that the localisation regions of the unitaries are separated by one \emph{idle} cell, as depicted in our example (Fig.~\ref{fig:QCA}).

This construction gives an autonomous QCA which, since its dynamics must by definition be (spatially) translationally symmetric, has cells composed of three cells of the original ccQCA plus one cell from the program band; and it turns out to have only nearest-neighbour interactions.  This general construction scheme can be optimized in an explicit situation to reduce the large cell-size.  Next we give such an explicit construction by starting from a universal ccQCA with homomorphisms that already have sequential structure, so the Margolus decomposition can be omitted.

\begin{SCfigure}
  \centering
     \includegraphics[width=75mm]{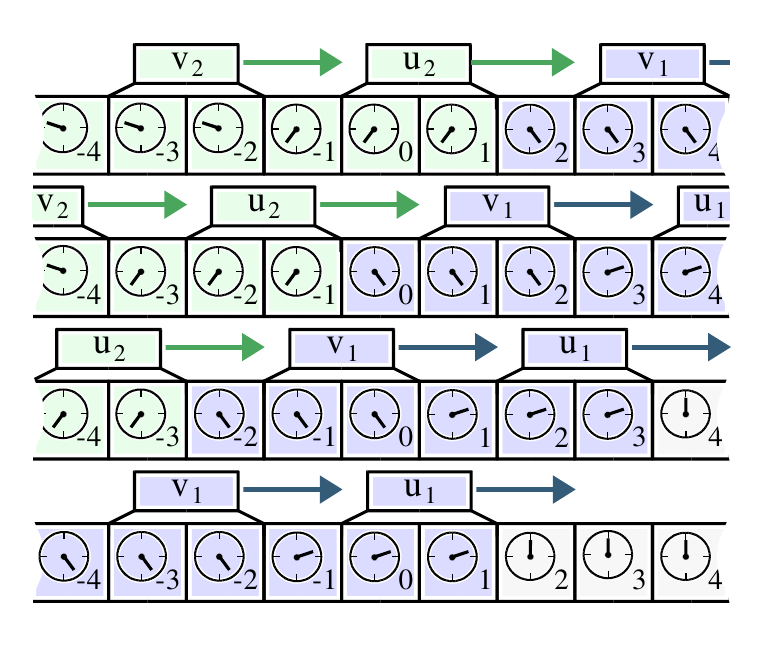}
     \caption{\label{fig:QCA} Four time-steps of the autonomous QCA. As the cells are updated sequentially, the clocks indicate the \textit{local time} of the cells corresponding to the time of the ccQCA.}
\end{SCfigure}

\subsection{Explicit construction} \label{sect:QCA-explicit}

In this subsection, we lay out a series of emulations, so as to make clear an explicit construction.

\subsubsection*{Circuits of ``controlled partial-$Y$''}

There is a two-qubit gate which, if not constrained to act always on neighbouring qubits but allowed to act on qubits within arbitrary range, serves as universal for computation within the standard gate model.  For example, we use the gate defined as 
\begin{equation} \label{def:G}
  G ~=~ \left( \begin{array}{cccc}
     1&0&0&0\\
     0&1&0&0\\
     0&0&\frac{1}{\sqrt2}&\frac{-1}{\sqrt2}\\
     0&0&\frac{1}{\sqrt2}&\frac{1}{\sqrt2} 
   \end{array} \right)
\end{equation}
in the computational basis, $\{\,\ket{00}, \ket{01}, \ket{10}, \ket{11}\,\}$. 
It performs a $\pi/4$ rotation on one qubit conditioned on the setting of another, and may equivalently be written $G_{ab} = \Lambda_a( \sqrt{-iY_b} )$.
To see that this gate is universal, it suffices to show that by use of simple (computational-basis) ancill\ae, one can simulate both the Hadamard and the Toffoli gates (\cf{}~\cite{book:NandC}), according to the well-known results of~\cite{lit:Aha0301}.  

\medskip
\begin{proposition}
  Gates $G = \Lambda(\sqrt{-iY})$ with no nearest-neighbour restriction, together with ancill\ae{} $\ket0$ and $\ket1$, can emulate ancill\ae{} $\ket+$ and $\ket-$ and all gates in the set $\{X, Y, Z, H, G^{-1}, \Lambda^2(\pm iY), \Lambda^3(\pm iY), \Lambda^2(Z), \Lambda^2(X) \}$.
\end{proposition} 

\begin{proof}
$\ket+_a$ and $\ket-_a$ are emulated respectively by $G_{ba}\ket1_b\ket0_a$ and $G_{ba}\ket1_b\ket1_a$, and $Y_a$ is emulated by $G^2_{ba}\ket1_b$, since global phase is unphysical.  $Z_a$ is emulated by $G_{ab}^4\ket\psi_b$ for any $\psi$.  $X$ is emulated (without ancill\ae{}) by $Y \cdot Z$, and $H_a$ by $G_{ba} \cdot Z_a \ket1_b$. 
The inverse $G^{-1}$ is equal to $G^7$ because its eigenvalues are all eighth roots of unity.

To emulate $\Lambda^2_{ab}(\pm iY_c)$, we use an ancilla $\ket0_p$ which will temporarily hold the `parity' of qubits $a$ and $b$.  Thus we first need the subroutine $\Lambda_{a \oplus b}(\pm iY_c)$ emulated by the sequence $G^{-2}_{ap} \cdot G^{-2}_{bp} \cdot G^{\mp 1}_{pc} \cdot G^2_{ap} \cdot G^2_{bp} \cdot \ket0_p$.  Then we have $\Lambda^2_{ab}(\pm iY_c) = G^{\mp1}_{ac} \cdot G^{\mp1}_{bc} \cdot \Lambda_{a \oplus b}(\pm iY_c)$.

The emulation for $\Lambda^3_{abc}(\pm iY_d)$ is rather similar, \emph{e.g.} it suffices to use the sequence $\Lambda^2_{ab}(iY_e) \cdot \Lambda^2_{ec}(\pm iY_d) \cdot \Lambda^2_{ab}(-iY_e) \cdot \ket0_e$.

Then $\Lambda^2_{ab}(Z_c)$ is emulated by $\Lambda^3_{abc}(iY_e)^2 \ket0_e$, and $\Lambda^2_{ab}(X_c) = H_c \cdot \Lambda^2_{ab}(Z_c) \cdot H_c$.
\end{proof}

(We offer no guarantee that these are the simplest emulations possible.  See \S\ref{sect:universal-gatesets} for more on universal circuit gate-sets.)

\subsubsection*{Construct qubit ccQCA}

Consider a ccQCA on a one-dimensional qubit lattice that allows the use of four different kinds of QCA-homomorphisms as described below, called $A$, $B$, $C$, and $D$.
These homomorphisms will be constructed from infinite tensor products of $G$ unitaries.   To prevent subscripts from becoming unreadable in what follows, we will also write $G(x,y)$ for $G$ acting on qubit $y$ controlled on qubit $x$, which was formerly denoted $G_{xy}$.

To show how this ccQCA can be used to simulate an arbitrary gate model circuit whose gates are all of the kind $G$ (line~(\ref{def:G})), we will think of the ccQCA's qubits as belonging to three interleaved one-dimensional lattices, and label them accordingly as $d_i, a_i, h_i$ with $i \in \ZZ$, as illustrated in Fig.~\ref{fig:three_bands}.
We will load the $d$-band with input corresponding to the input of the circuit being simulated, initialising its unused qubits to $\ket0$.
The $a$-band will be used as `ancilla space' and should be initialised to $\ket0$ everywhere.
The $h$-band is used to break spatial symmetry of the dynamics, and should contain a single `pointer' $\ket1$, with the rest of its qubits containing $\ket0$.
The four families of homomorphisms we consider here are given explicitly as (tensor) products of unitaries
\begin{eqnarray}  \label{eqn:homomorphisms1}
  A_i=\prod_{x \in \mathbbm{Z}} G(h_x,a_{x+i}),&&
  B_i=\prod_{x \in \mathbbm{Z}} G(h_x,d_{x+i}), \nonumber\\\\
  C_i=\prod_{x \in \mathbbm{Z}} G(d_x,a_{x+i}),&&
  D_i=\prod_{x \in \mathbbm{Z}} G(a_x,d_{x+i}).\nonumber
\end{eqnarray}

\begin{proposition}  \label{propos:ABCDG}
  For each $i,j \in \ZZ$, $i \not=j$, there exists a sequence of homomorphisms drawn from those of line~(\ref{eqn:homomorphisms1}) which, when applied to a tri-band lattice of cells initialised as described above, emulates the unitary gate $G(d_i, d_j)$ on the $d$-band and restores both the $a$-band and the $h$-band to their initial (separate) configurations.  The complexity of the sequence is constant, though its description complexity grows logarithmically with $i$ and $j$.
\end{proposition}

\begin{proof}
First note that each of $A_i, \ldots, D_i$ has order 8, since that is the order of the unitary $G$.
Suppose without loss of generality that $h_0$ is the present location of the pointer.
Consider the sequence $T := C_{-i}^{-1} \cdot B_i^2 \cdot C_{-i}$.  The only place it has net effect (because of the pointer) is between qubits $d_i$ and $a_0$, where it emulates $T' := G^{-1}(d_i,a_0) \cdot Y_{d_i} \cdot G(d_i,a_0)$.
The sequence required by the Proposition is then taken to be 
\begin{equation}
A_0^7\cdot C_{-i}^7 \cdot B_i^6 \cdot C_{-i} \cdot D_j \cdot C_{-i}^7 \cdot B_i^2 \cdot C_{-i} \cdot A_0,
\end{equation}
which we can reparse as
\begin{equation*}
A_0^{-1}\cdot T^{-1} \cdot D_j \cdot S \cdot A_0,
\end{equation*}
and which---given the promised initial conditions---emulates
\begin{equation*}
\sqrt{Y_{a_0}} \cdot {T'}^{-1} \cdot G(a_0,d_j) \cdot T' \cdot \sqrt{Y_{a_0}},
\end{equation*}
restoring all other qubits.  Since the $a$-band starts out entirely zero, this last line can be shown (by direct computation of 8-by-8 matrices) to emulate $G(d_i,d_j)$, restoring $a_0$ also.
(To simplify this final computation, it helps to notice that the product $T' \cdot \sqrt{Y_a}$ is given by a 4-by-4 \emph{integer} matrix, whose action can be perceived using `classical intuition'.)
\end{proof}

Of course, a clever compiler would find ways of simulating a given circuit that are more efficient than repeated application of this technique.

\begin{SCfigure}
  \centering
     \includegraphics[width=50mm]{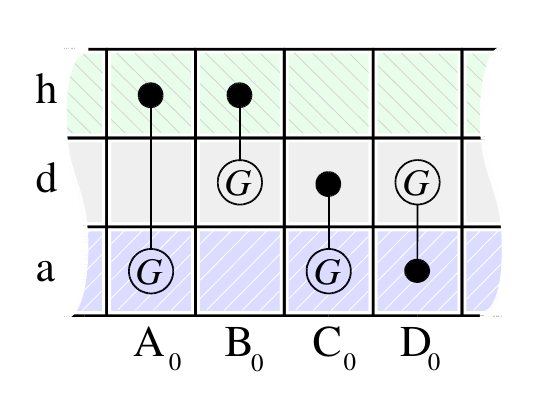}
     \caption{\label{fig:three_bands} Structure of the `three band' QCA with pointer, data and ancilla band. Localisations of the homomorphisms of line~(\ref{eqn:homomorphisms1}) are indicated.}
\end{SCfigure}

\subsubsection*{Construct nearest neighbour qubit ccQCA}

The ccQCA described above employs operations with arbitrarily large neighbourhood. With additional neighbour-swap operations we can first move two bands to the required interband-distance $i$, then apply a cellwise $G$-operation ($A_0,B_0,C_0$, or $D_0$) and finally shift back, in order to implement all operations of line~(\ref{eqn:homomorphisms1}) with a nearest neighbour ccQCA. Moreover, the inter-band $G$-operations are quite similar, which suggests interleaving the three bands into a single qubit band labelled $q_i$, with $i \in \ZZ$,
\begin{equation} \label{eqn:relabel_adh}
   (\ldots,q_{-2},q_{-1},q_0,q_1,q_2,\ldots) =
   (\ldots,a_{-1},h_{-1},d_0,a_0,h_0,\ldots).
\end{equation}
By Proposition~\ref{propos:ABCDEF}, a sufficient set of operations is then given by
\begin{eqnarray} \label{eqn:homomorphisms2}
  E_j & = & \prod_{x \in \ZZ} \mbox{Swap}(q_{3x+j},q_{3x+j+1}), \nonumber \\
  F_j & = & \prod_{x \in \ZZ} G(q_{3x+j},q_{3x+j+1}),
\end{eqnarray}
for $j \in \{\,0,1,2\,\}$. 

\medskip
\begin{proposition}  \label{propos:ABCDEF}
  With the relabelling of line~(\ref{eqn:relabel_adh}), for each $i \in \ZZ$, for each homomorphism $A_i, B_i, C_i, D_i$, there exists a sequence of `short-range' homomorphisms drawn from those of line~(\ref{eqn:homomorphisms2}) which, when applied to a single-band of qubits, emulates the required homomorphism.  The complexity of the sequence is linear in $i$, and so its description complexity also grows as $O(i)$.
\end{proposition}

\begin{proof}
Note that it suffices to move two bands relative to each other, since the homomorphisms of line~(\ref{eqn:homomorphisms1}) act non-trivially on only two bands at a time.
Since all the cases are basically the same, we will illustrate emulation of $A_3$ only~:
\begin{eqnarray*}
  A_3  &=&  ( E_2 \cdot E_0 \cdot E_1 \cdot E_2 \cdot E_0 ) \cdot F_1 \cdot 
            ( E_0 \cdot E_2 \cdot E_1 \cdot E_0 \cdot E_2 )  
\end{eqnarray*}
\end{proof}

\subsubsection*{Construct universal nearest-neighbour QCA}

The homomorphisms of the ccQCA described above already work on non-overlapping neighbourhoods, and so there is no need for further Margolus decomposition here. For our main design of an autonomous nearest-neighbour QCA, we introduce a `program band', and focus on minimising the dimension of the individual cells.  

Take a one-dimensional lattice of qudits labelled $c_i$ (for \mbox{$i \in \mathbbm{Z}$}) of single cell dimension $d=12$, and regard these as incorporating one qutrit cell of a program band $t_i$ with two qubit cells $q_{2i}$ and $q_{2i+1}$ from the data band of the ccQCA.  The cell $c_i$ we define explicitly as the tensor product
\begin{equation}
  c_i = t_i \otimes q_{2i} \otimes q_{2i+1}.
\end{equation}
(Identification of data cells is indicated in Fig.~\ref{fig:interleaved}.)  As before, it is not necessary to have any `fine control' over the relative motion of the two sub-bands $t$ and $q$; rather we simply allow one to pass by the other with an invariant velocity.  This is achieved by decomposing the QCA transformation step into two parts, a unitary and a shift~:
\begin{eqnarray} \label{def:finalTransform}
  U\colon~ \mathbbm{C}^{12} \rightarrow \mathbbm{C}^{12} \mbox{ acting on every cell simultaneously,} \nonumber \\
  S\colon~ \left( \begin{array}{rcl}t_i&\mapsto&t_{i+1}\\
  q_i&\mapsto&q_{i-1} \end{array} \right) \mbox{ sliding the bands relatively.~~~~}
\end{eqnarray}

To simulate the nearest-neighbour ccQCA, we will interpret the data band $q_i$ exactly as before, but the program band $t_i$ must be initialised so as to execute the appropriate transformations on the data band as the two bands slide past one another. At initialisation, the cells $i>0$ will be used to hold the non-zero content of the data band in their qubits, while the cells $i<0$ will be used to hold the program band in their qutrits. We will initialise the $t_i$ in the computational basis, and the $U$ operation will be defined to leave these qutrits invariant.
Specifically, $t_i=\ket0$ will cause no transformation, $t_i=\ket1$ will cause a swap of data between $q_{2i}$ and $q_{2i+1}$, and $t_i=\ket2$ will cause the transformation $G(q_{2i}, q_{2i+1})$ described at line~(\ref{def:G}).

\begin{SCfigure}
  \centering
     \includegraphics[width=70mm]{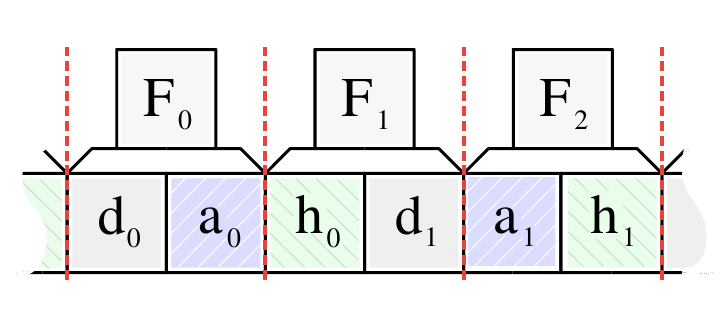}
     \caption{\label{fig:interleaved} Structure of the autonomous QCA.
     Indicated are three cells with one program qutrit and two data qubits each;
     identification of the data qubits corresponding to \mbox{Fig.~\ref{fig:three_bands}}.}
\end{SCfigure}

\medskip
\begin{proposition}
  There is a nearest-neighbour QCA on a 1-dimensional lattice that efficiently emulates each of the six homomorphisms of line~(\ref{eqn:homomorphisms2}), in each case using $O(1)$ cells to store the instruction.  
  The description complexity of the QCA program is therefore linear in the number of ccQCA homomorphisms it emulates, and the run-time of the QCA is linear in the sum of the program length and data length.
\end{proposition}
Note that in our construction (below), the `instructions' or `program-segments' are given by triples of qutrits which remain in computational basis states throughout.

\begin{proof}
Take $U$ to be the 12-by-12 unitary described above, $S$ to be the shift operator described above that slides program qutrits ($t_i$) past data qubits ($q_i$), and $\cT = S \cdot U$ to be the nearest-neighbour unital homomorphism of the QCA.
Each of the six homomorphisms of line~(\ref{eqn:homomorphisms2}) is emulated on all the data qubits ($q_i$) by a specific pattern of three neighbouring qutrits of program completely sliding past all of the data qubits.  
That this $\cT$ describes a nearest-neigbour homomorphism is immediate from Fig.~\ref{fig:interleaved}.

The program-segment $\ket{100}$ on $t_{3i},t_{3i+1},t_{3i+2}$---as it moves rightwards---will simulate the homomorphism $E_0$.  This is because the $\ket0$ initially on $t_{3i+2}$ will hit every pair $q_{3j+1},q_{3j+2}$ having no effect, then the $\ket0$ initially on $t_{3i+1}$ will hit every pair $q_{3j+2},q_{3j+3}$ having no effect, then the $\ket1$ initially on $t_{3i}$ will hit every pair $q_{3j},q_{3j+1}$ thereby implementing $E_0$.
Similarly, the program-segments $\ket{010}$ and $\ket{001}$ will simulate the homomorphisms $E_2$ and $E_1$ respectively.
Likewise, the program-segments $\ket{200}$, $\ket{020}$, $\ket{002}$, will simulate the homomorphisms $F_0$, $F_2$, $F_1$, respectively. 
This is in accordance with the general construction idea outlined in \S\ref{sect:QCA-general}.
The cells with negative index should be initialised with program-segments of these kinds in order to induce the desired transformations on the data.  The cells with non-negative index should be loaded with the relevant data.  The computation output may be read (in the computational basis) any time after the content of the program band has moved past the content of the data band.
\end{proof}

To show that the composite simulation is efficient, one needs to estimate the necessary resources. Consider a quantum circuit (QC) consisting of \textit{Space}$_{QC}$ qubit-wires and \textit{Time}$_{QC}$ $G$-gates (assuming no exploitation of parallelism). 
In the first simulation step, the resources of the ccQCA depend linearly on the corresponding QC resources (Propos.~\ref{propos:ABCDG}). The use of swap gates in the next step increases the time (Propos.~\ref{propos:ABCDEF}); the encoding of the program into the program band increases the space, so one ends
up with an estimate for the autonomous QCA of
\begin{eqnarray}
\mbox{\textit{Time}}_{QCA} = \mathcal{O}(\mbox{\textit{Time}}_{QC}\cdot \mbox{\textit{Space}}_{QC}), \nonumber \\
\mbox{\textit{Space}}_{QCA} = \mathcal{O}(\mbox{\textit{Time}}_{QC}\cdot \mbox{\textit{Space}}_{QC}).
\end{eqnarray}
The resources depend polynomially on the given QC, so the simulation is deemed efficient for universal $\BQP$ simulation.

\section{Discrete-time Spin Chains}  \label{sect:spinChains}

In this section,\footnote{Previously unpublished work, the ideas in this section were presented during a talk given at Bristol in 2007.} we highlight another design for a novel QCA-based paradigm universal for $\BQP$.
This is based on observations of so-called `quantum wires' or `spin chains', \cf{} 
\cite{lit:CDEL04, lit:BB0401, lit:R0501, lit:KS0501, lit:FT0601, lit:FXBJ06}.

The main technical contribution of this section is to construct a design of discrete-time spin-chain processor which, under classical control, is universal for $\BQP$, but which has the special feature that all signal addressing (classical control) goes not to the whole machine but only to a tiny part of it (the `control window').  Accordingly it also seems that our encoding of logical qubits within physical spins is novel and marginally more efficient (2/3 density) than the more common methods of `barrier qubits' \cite{lit:BB0401}. The context of our design is similar to that of \cite{lit:R0501}, but again there signalling is passed to all qubits of the machine rather than only to a small part, even though translational invariance of dynamics is enforced.
It appears that our design construction is, in some sense, `simplest' amongst the discrete-time spin-chain models, and an analogous continuous-time universal construction is presently lacking.  For example, \cite{lit:YBB06} presents a nice continuous-time `processor core' model, but it nonetheless engages control signals to the bulk of the qubits, rather than to a small `window'. 

It can be argued that continuous-time and discrete-time models for dynamics on a lattice of quantum cells are not directly comparable, since the homomorphism for a discrete-time QCA is generally given by a Margolus decomposition into alternating unitaries (\cf~\S\ref{sect:QCA}), whereas for a continuous-time QCA it is given by a local Hamiltonian.  But a general local Hamiltonian on an arbitrarily lattice, when executed for any fixed length of time $\delta t$, is liable to induce a unitary that is \emph{not} completely local but rather allows a small (albeit negligible) amount of information to propagate arbitrarily far.  
Conversely, if the alternating unitaries of a Margolus decomposition are encoded directly within a Hamiltonian, then that Hamiltonian must oscillate in time and not be constant, so that it can represent each of two different unitaries in turn.
Thus there seems to be no obvious way to transfer results from one context directly to the other, and so this work is perhaps not directly comparable with those studies of spin-chains in context of continuous time.

Here is a quick summary of the properties of the main construction, the autonomous QCA, that we present in this section~:
\begin{itemize}
  \item 
\textbf{Universal;}  a (physically reasonable) paradigm capable of simulating quantum circuits, with polynomial overhead in most reasonable measures.  
 \item
\textbf{Discrete space, finite, one-dimensional;}  a `chain' of $N$ physical qubits `attached' to a single qutrit `window'.
  \item 
\textbf{Discrete time, nearest-neighbour locality;}  a `clock' homomorphism, composed of two local unitaries interleaved, which causes finite speed of data propagation, and is highly symmetric in space and time.  
  \item
\textbf{Limited dynamics;} apart from the clock, all other operations must affect only the qutrit window; \ie{} all control signals are addressed to $O(1)$ of the storage space.
\end{itemize}

\subsection{Addressing control in a discrete-time spin-chain} 

Our approach here differs from the one of \S\ref{sect:QCA}, and from other similar considerations in the literature, in that now we make no assumption about being able to address all of the computer to read out and load in data and program, but we do allow local time-dependent control of a very small part of the computer.  This small part is effectively to be considered as the only `window' that the device has onto the outside world, the rest of the machine being isolated from control and environment.  
Whereas one might expect it to be necessary to possess a large degree of localised control during initialisation and output---the very places where decoherence has the `benevolence' of enabling non-reversible `entropic' effects to take place, such as resetting and measuring---relegating initialisation and output to the first and final time-phases respectively of the overall computation; yet, in the present design, we instead constrain `entropic' effects not to certain temporal phases but to a particular spatial location~: the terminus of a `quantum wire'.  The design remains $\BQP$-universal despite requiring only a constant number of different kinds of operation (including addressing), in the same spirit as the designs given recently in~\cite{lit:Sev0601}.

\subsubsection*{Physical terminology}

It is convenient to borrow some language from the physical architectures proposed for implementing `quantum wires'.  Thus we refer to these structures as \emph{spin chains}, since the individual `low-level qubits' constituting a `quantum wire' are invariably imagined to be (or indeed convincingly implemented as) nuclear spins in context of an Ising model, or similar.  The idea is to have a chain of sites, indexed by $\{0,1,2,\ldots,N-1\}$, where a two-dimensional Hilbert space is associated to each site, to describe a \emph{physical} qubit there.  Such low-level qubits are then termed \emph{spins}, to emphasise two important properties~: firstly the idea that these qubits, unlike the \emph{logical} qubits to which we shall be coming shortly, are not abstracted very far away from the underlying physical architecture, and are most likely implemented as the quantum spin of a spin-$\frac12$ particle; secondly the idea that there need be no difference in energy levels between $\ket0$ and $\ket1$, no explicit method of addressing these qubits arbitrarily, and no preferred basis for (unwanted) decoherence.  The exception to this rule applies at the terminal site, having index $0$.  This `window' site is instead associated with a three-dimensional space (hence a qutrit), since  the use of a larger site for the `window' onto the device will be seen to simplify much of the rest of the design of the computing paradigm.   (Whether a qubit window would in fact suffice here is not presently known.)

\subsection{Clocks with graph-symmetry}  \label{sect:clocks}

The notion of \emph{spin chain} can be generalised to that of a \emph{spin network}, according to an undirected graph.  Although we shan't need graphs more complex than linear arrays, it is appropriate to describe the \emph{clock dynamics} in the more general case.

Let $\cG(\cV,\cE)$ denote an undirected graph each of whose nodes is associated to a distinct physical qubit.  Using $X$ and $Z$ to denote canonical Pauli operators and subscripts to denote qubit indices, we can define the symmetric discrete-time clock dynamics for $\cG$ according to the following formul\ae{} (\cf{}~\cite{lit:R0501} and \S\ref{sect:notations})~:
\begin{eqnarray} \label{def:graphclock}
  H_j            &:=&  e^{i\pi(~ X_j + Z_j - \sqrt2 ~)/\sqrt8}  ~~=~~
                       \frac{X_j+Z_j}{\sqrt2},  \nonumber \\
  V_\cG          &:=&  \prod_{j \in \cV} H_j, \nonumber \\
  \Lambda_j( Z_{k} ) &:=&  e^{i\pi(~1-Z_j-Z_k+Z_jZ_k~)/4} ~~=~~ 
                       \frac{1+Z_j+Z_k-Z_jZ_k}{2},  \\ 
  E_\cG          &:=&  \prod_{(j,k) \in \cE} \Lambda_j (Z_{k}), \nonumber \\
  G_\cG              &:=&  E_\cG \cdot V_\cG.  \nonumber 
\end{eqnarray}

This discrete-time picture is in many ways simpler than the corresponding continuous-time dynamic, fitting more naturally with a discrete-space model and with standard notions of computation.  For example, there is no need to tune the individual interaction strengths in order to obtain a uniform flow of data, \cf{}~\cite{lit:CDEL04}.

These formul\ae{} are reminiscent of the operations used in Graph State computing, where $\cG$ would be a two-dimensional lattice and $G_\cG$ would map the all-zero state $\ket{\mathbf0}$ into a so-called \emph{cluster state} for measurement-based quantum computing (\cf{}~\cite{lit:Raus03,lit:RB01}), which is again a discrete-time universal computing paradigm.
By contrast, that model uses $G_\cG$ only once, and only to establish initial entanglement, not to distribute control signal nor indeed any other data.

We encode \emph{logical} qubits using the clock $G_\cG$ directly.\footnote{I have a program that allows one to draw an arbitrary graph, colour its nodes with Pauli operators, and then evolve the operators at various speeds with the clock $G_\cG$.  This not only makes for a novel screen-saver, but also helps make more intuitive the Gottesman-Knill theorem.} 
Specifically focussing on any spin at the end of a spin chain (a `leaf' of $\cG$), it is easily shown to be necessary to wait for precisely three clock-ticks before all of the data from that spin has been transported away, assuming that $\cG$ has the local topology of a simple spin chain in the immediate vicinity of the terminal spin in question.  Reason as follows~:
\begin{eqnarray} \label{eqn:conjugatePaulibyG}
  G_\cG \cdot X_j \cdot G_\cG^\dag &=~~ Z_j, 
  ~~~~~~~~~~~~~~
  G_\cG \cdot Z_j \cdot G_\cG^\dag ~~=& X_j \prod_{(j,k) \in \cE} Z_k;  \nonumber \\
  G_\cG^3 \cdot X_0 \cdot G_\cG^{\dag3} &=~~ X_1 Z_2, 
  ~~~~~~~~
  G_\cG^3 \cdot Z_0 \cdot G_\cG^{\dag3} ~~=& X_2 Z_3.
\end{eqnarray}
(Note that each of $X_1Z_2$ and $X_2Z_3$ commutes with each of $X_0$ and $Z_0$, and so can be taken to represent a qubit distinct from the one represented by $X_0$ and $Z_0$.)

Since it is necessary for logical qubits to be properly distinct from one another, this naturally suggests taking our logical qubits to be revived in sequence at a given terminus after every three clock-ticks.
Accordingly, we can define the logical qubits by specifying pairs of anticommuting operators to serve as their `Pauli basis'~: 
\begin{eqnarray} \label{eqn:logicalqubits}
  \cX_j &:=& G_\cG^{3j} \cdot X_0 \cdot G_\cG^{-3j}, \nonumber \\
  \cZ_j &:=& G_\cG^{3j} \cdot Z_0 \cdot G_\cG^{-3j}.
\end{eqnarray}
Calligraphic script is being used to denote operators that define logical qubits, while ordinary script is being used to denote operators that define the physical spins.

\medskip
\begin{proposition}
  Let $\cG$ be the graph that is a simple line on $N=3n+2$ vertices, and let $G_\cG$ be a clock homomorphism on that graph as defined at line~(\ref{def:graphclock}).  
  Then $N+1=3n+3$ clock-ticks reverses the data on the physical spins (vertices) of the graph, and $6n+6$ clock-ticks therefore revives the initial state perfectly.
\end{proposition}

\begin{proof}
The \emph{projective} Pauli group on $N$ spins, obtained by quotienting away global phase, is Abelian, and therefore isomorphic to the (\emph{additive} group of the) vector space $\FF_2^{2N}$.  The operator $G_\cG$ is in the Clifford group (that is, its conjugative action stabilizes the Pauli group), and so its conjugative action on the projective Pauli group must be a linear endomorphism.  Thus it must have a representation via a $2N$-by-$2N$ matrix over $\FF_2$.  
This is called the \emph{stabiliser formalism} in \cite{lit:R0501}.
We can choose to list a basis for the projective Pauli group in the order $\{X_0, X_1, \ldots, X_{N-1}, Z_0, \ldots, Z_{N-1} \}$, and then a matrix for $G_\cG$ is given in line~(\ref{eqn:graphmat}) in $N$-by-$N$ block form, where $A_\cG$ is an adjacency matrix for $\cG$~:
\begin{equation}  \label{eqn:graphmat}
  M_\cG ~~:=~~ \left( \begin{array}{cc} 0 & 1 \\ 1 & A_\cG \end{array} \right).
\end{equation}

Then generate the recurrence $S_{-2} := 1,~ S_{-1} := 0,~ S_{0} := 1, ~ S_{k} := S_{k-2} + A_\cG \cdot S_{k-1}$, and observe inductively that 
\begin{equation}  \label{eqn:graphmat1}
  M_\cG^k ~~=~~ \left( \begin{array}{cc} S_{k-2} & S_{k-1} \\ S_{k-1} & S_k \end{array} \right).
\end{equation}

Using this formulation, it is straightforward to check the properties of various spin networks (graphs) in context of the (discrete-time) dynamics of $G_\cG$ for data flow.
For the present Proposition, it suffices to consider the case where $\cG$ is a line on $N=3n+2$ vertices.  In that case, it only remains to show that $S_{N}=0$ and that both $S_{N\pm1}$ are equal to the `reversal' matrix~: the permutation matrix that reverses the order of the vertices.  Then $M_\cG^{N+1}$ will have the effect of reversing the data on the vertices, as required. 

But $S_k$, as a polynomial in $A_\cG$ over $\FF_2$, actually \emph{is} the characteristic polynomial of the adjacency matrix of the line $\cG(k)$ on $k$ vertices, because each is given by the formula $Det( \lambda I + A_{\cG(k)} )$ over $\FF_2$.  
\begin{eqnarray*}
  Det( \lambda I + A_{\cG(k)} )  
    &=&  Det( \lambda I + A_{\cG(k-1)} ) + \lambda \cdot Det( \lambda I + A_{\cG(k-1)} ).
\end{eqnarray*}
So for $k=N=3n+2$ we see that $S_k=0$ automatically for the linear graph.  To see then that $S_{N \pm 1}$ must be reversal matrices, note that the group of symmetries of the line is of cardinality 2, so $S_{N \pm 1}$ can only be either a reversal or the identity.  That it is in fact a reversal can be seen directly from Fig.~\ref{fig:spinchainthing}, where the case $N=8$ is fully illustrated.  A full analysis of this algebra is given in the appendix of \cite{lit:R0501}.
\end{proof}

\medskip
\begin{proposition}
  A linear spin chain having $N=3n+2$ nodes will encode exactly $w=2n+2$ logical qubits by the rule at line~(\ref{eqn:logicalqubits}).  
\end{proposition}

\begin{proof}
One logical qubit is encoded every three clock-ticks, under the encoding rule suggested.  It takes $6n+6$ clock-ticks to revive the original state (previous Proposition).  Thus in one cycle of $6n+6$ clock ticks there is scope for $2n+2$ logical qubits to be encoded.  That these logical qubits are independent can be seen by more matrix algebra or directly intuited from Fig.~\ref{fig:spinchainthing}, where the case $n=2$ is fully illustrated.
\end{proof}

This entails a natural encoding density of $\frac{w}{N} \sim \frac23$, (logical to physical ratio) \cf{}~\cite{lit:BB0401}.  Unlike the technique of block-coding discussed in \cite{lit:FXBJ06}, our method keeps the logical qubits from dissociating over a wide area, so that a reasonably standard local error model could be utilised.  That is to say, the spontaneous depolarisation of a spin will damage at most two logical qubits at any given time.  Another advantage of retaining a good degree of locality in the encoding is that it makes tomography more straightforward in the case where the implementation is such that one does not know \emph{a priori} how many spins are in the chain.  (Having said this, our main concern is with structural simplicity, and not the adaption of design for error correction capability.)

\begin{SCfigure}
  \centering
     \includegraphics[width=70mm, height=100mm]{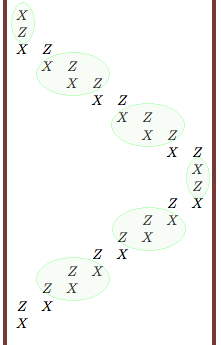}
     \caption{\label{fig:spinchainthing} The encoding relating $3n+2$ spins to $2n+2$ logical qubits on a linear graph, in the case $n=2$. (Time vertical, space horizontal.)  Each column represents a physical spin, 8 in this picture.  Each row represents a single Pauli operator on the $3n+2$ spins, and the row below it represents the same operator conjugated by a clock-tick, $G_\cG$.  The shaded rings pick out pairs of anticommuting Pauli operators that serve to define the $2n+2$ logical qubits, 6 in this picture.  Operators from different pairs commute, and so the logical qubits are distinct.}
\end{SCfigure}%

As can be seen from Fig.~\ref{fig:spinchainthing}, at most two of the logical qubits will be revived on local spins after any given clock-tick, these spins being the ones at either end of the chain.
(Indeed, it is not hard to show that for \emph{any} undirected graph $\cG$, if a logical qubit is identified with one of the vertices of $\cG$, then there can be at most one other vertex where that logical qubit is capable of fully reviving under the repeated action of the clock $G_\cG$ alone.)

\subsection{Window qutrit}

Recall that we intend for the site at one terminus of our spin chain to house a qutrit rather than a qubit.  (Here it is most definitely appropriate to speak of ``three energy levels'', because we expressly intend to address control signals to this physical qutrit directly.)   We identify the two lowest energy levels of the qutrit with a logical qubit, so that the operators $X_0$ and $Z_0$ remain well-defined\footnotemark{} for the sake of the clock $G_\cG$ (line~(\ref{def:graphclock})) and the encoding of the other logical qubits (line~(\ref{eqn:logicalqubits})).  
\footnotetext{\eg{} $X_0 := \ketbra01 + \ketbra10 + \ketbra22$; ~$Z_0 := \ketbra00 - \ketbra11 + \ketbra22$.}
The third energy level is reserved as a `storage space' to enable fine control over the evolution of the system as a whole.

Now, the computing paradigm is defined by assuming that $G_\cG$ evolutions happen regularly, and between any two $G_\cG$ evolutions we are free to apply any physically plausible operation we please to the qutrit at site $0$. 

\medskip
\begin{definition}
  A program for the spin-chain computer is defined to be a list of valid qutrit operations (\eg{} 3-dimensional unitaries, measurements in the computational basis, \&c.), to be applied on the `window qutrit', to be interleaved with clock homomorphisms $G_\cG$. Quantum data stored in the spin-chain is processed by working through the list.
\end{definition}

The qutrit operations we shall use to obtain universality are listed below~:
\begin{itemize}
  \item
\textbf{Reset;}  replace the qutrit with the pure state $\ket0$.
  \item
\textbf{Measure;}  obtain a classical trit, collapsing the state of the system according to the usual Born rule by projecting onto an energy level.
  \item
\textbf{Unitary;}  apply a 3-dimensional unitary gate to the qutrit.
\end{itemize}

For universality with regard to $\L$-reductions, we look for there to be a log-space Turing machine that converts the description of an arbitrary quantum circuit (given in some standard explicit form) to a description of such a list of homomorphisms, so that the effect of the quantum circuit is emulated within the spin chain as it works through applying the operations on the list, interleaved with $G_\cG$ evolutions.

  Operations are permitted to be adaptive in general, so that a unitary on the list might be a function of the result of a measurement previously listed. 
  However, just as a quantum circuit usually begins by resetting all of its qubits to zero and then delays all other measurements to the end, so the spin chain computer could, in general, be expected to emulate such quantum circuits by resetting all of its logical qubits to zero at the beginning, and delaying all of its measurements to the end also.  These are the kinds of emulations that we are interested in, and so we will proceed by showing that \emph{after resetting to zero and before final measurement}, the list of operations considered need only contain non-adaptive unitaries, provided that the circuit being emulated is likewise constructed.

\subsection{Emulating quantum circuits}

To initialise logical qubits of the spin chain to zero, simply apply clock-ticks until the qubit in question is revived at site $0$, then reset it.  Final measurements are rendered in likewise fashion.  

Similarly, it is trivial to emulate a single-qubit unitary on the spin-chain computer, because each logical qubit periodically revives to site $0$ where we can access it directly at the physical layer.

To complete the emulation of a general circuit, we need only show how to implement some non-trivial two-qubit unitary in a general position on the logical space.  Unfortunately this will require (in the worst case) about $18n$ clock-ticks for nearest-neighbour gates, \ie{} three cycles of the data, rather than just one, and even more clock-ticks for non-nearest neighbour gates.  This is because the clock-ticks move the \emph{logical} data in just one direction, whereas a non-trivial two-qubit unitary implicitly requires a bidirectional flow of data.  (There is possibly some scope for making use of the unused $n$-qubits-worth of space in the spin system to circumvent this slowdown, but that is immaterial if we merely wish to show polynomial efficiency of simulation, and would presumably require a more complex encoding, and perhaps the use of partial measurements, \etc)

\subsubsection{Logical nearest-neighbour interactions}

In our first example of emulating a non-trivial two-(logical-)qubit gate, we shall use $G_\cG$ up to about $12n$ time-steps, and also make use of the third energy level at site $0$.  Define $U_0$ on the qutrit to be the unitary operator that exchanges the top two energy levels at site $0$, \viz{} $\ketbra00 + \ketbra12 + \ketbra21$.  Because of the way that $E_\cG$ has been defined at line~(\ref{def:graphclock}), a single application of $U_0$ effectively switches off the `natural' interaction between sites $0$ and $1$ during a clock-tick~$G_\cG$.  Thus, the following identities are immediate~:
\begin{eqnarray}  \label{eqn:logicalcnot}
  E_\cG \cdot ( U_0 \cdot E_\cG \cdot U_0 )  
    &=&  \Lambda_0 (Z_{1}) \\
  =~~ G_\cG \cdot H_0 U_0 H_0 \cdot G_\cG^\dag \cdot U_0  
    &=&  \Lambda_1 (Z_0)~;  \nonumber \\
  H_0 \cdot G_\cG \cdot H_0 U_0 H_0 \cdot G_\cG^\dag \cdot U_0 H_0 
    &=&  \frac{ 1 + X_0 + Z_1 - X_0 Z_1 }2~;  \nonumber \\
  G_\cG \cdot \left(~ H_0 \cdot G_\cG \cdot H_0 U_0 H_0 \cdot G_\cG^\dag \cdot U_0 H_0 ~\right) \cdot G_\cG^\dag 
    &=&  \frac{ 1 + Z_0 + Z_0 X_1 Z_2 - X_1 Z_2 }2 \nonumber \\
  =~~ G_\cG \cdot H_0 \cdot G_\cG \cdot H_0 U_0 H_0 \cdot G_\cG^{6n+5} \cdot U_0 H_0 \cdot G_\cG^{6n+5} 
    &=&  \frac{ 1 + \cZ_0 + \cZ_0\cX_1 - \cX_1 }2. \nonumber 
\end{eqnarray}

This latter gate (on logical qubits $0$ and $1$) is locally equivalent to a logical nearest-neighbour C-Not gate.
A C-Not gate can be used three times, with appropriate intervening single-logical-qubit unitaries, to emulate a \emph{swap} gate (on the same logical qubits), and thence logical qubits can be swapped about as necessary to emulate non-nearest-neighbour two-logical-qubit unitaries.  

And so we are able to render efficiently a universal set of operations on the logical qubits of the system, using only $G_\cG$ as a clock, together with access between clock-ticks to a qutrit at one end of a spin chain.
Of course, blind substitution according to the description given here will likely lead to programs in this paradigm that could be otherwise compiled in a more optimal fashion, \eg{} by taking advantage of opportunities to pack more than one simulated gate into each cycle of the data.

\subsubsection{Efficiency of emulation}

We end the section by considering the efficiency of the emulation described.  To improve the simplicity of the reduction, we take $\Lambda(Z)$ to be the principal two-qubit gate used in quantum circuit design, rather than the more usual choice of $\Lambda(X)$ (C-Not).

\medskip
\begin{proposition}
  Let $C$ be a quantum circuit on a line of $w=2n+2$ qubits, composed of nearest-neighbour $\Lambda(Z)$ gates and single-qubit unitaries arranged into time-slices.  Let $d$ be the total number of time-slices in $C$, that is, the \emph{depth} of $C$.  The emulation of $C$ on the spin-chain computer as described requires $O(w)$ space and $O(w \cdot d)$ time.
\end{proposition}

\begin{proof}
Set $N=3n+2$ and work with a spin-chain processor of that size.
In accordance with reasoning very similar to that used at line~(\ref{eqn:logicalcnot}), make the abbreviations
\begin{eqnarray*}
  J_{[0..N-1]} &:=&  Z_0 \cdot G_\cG \cdot H_0 \cdot G_\cG \cdot H_0 U_0 H_0 \\
  K_{[0..N-1]} &:=&  U_0 H_0 \cdot G_\cG^{2} \cdot H_0;
\end{eqnarray*}
and then a logical $\Lambda(Z)$ between qubits $j-1$ and $j$ may be rendered as
\begin{eqnarray*}
  G_\cG^{3j} \cdot H_0 \cdot G_\cG^{6n+3} \cdot J \cdot G_\cG^{6n+5} \cdot K \cdot G_\cG^{6n+6-3j},
\end{eqnarray*}
which effectively involves three `cycles' of the data structure.  While these cycles are taking place for the emulation of some $\Lambda(Z)$ in $C$, any number of single qubit unitaries from the same time-slice of $C$ can be inserted at the appropriate point, and so contribute nothing to the overall cost of the emulation, as measured in clock-ticks.  
Moreover, other $\Lambda(Z)$ gates from the same time-slice of $C$ may also be inserted, for no additional cost~: for example, to emulate $\Lambda(Z)$ between qubits $j-1$ and $j$ \emph{and also} $\Lambda(Z)$ between qubits $k-1$ and $k$, when $k-j\ge2$, we use
\begin{eqnarray*}
  G_\cG^{3j} \cdot H_0 \cdot G_\cG^{3(k-j)} \cdot H_0 \cdot G_\cG^{6n+3-3(k-j)} \cdot J \cdot G_\cG^{3(k-j)-2} \cdot \phantom{XXX} \\  \cdot J \cdot G_\cG^{6n+5-3(k-j)} \cdot K \cdot G_\cG^{3(k-j)-2} \cdot K \cdot G_\cG^{6n+6-3k},
\end{eqnarray*}
which still involves only three cycles.
(In fact, a closer inspection reveals that this method is perfectly valid for implementing \emph{overlapping} nearest-neighbour $\Lambda(Z)$ gates, \eg{} $k-j=1$, so it is possible to implement more of these kinds of gates in three cycles than can be implemented in one time-slice within the standard quantum circuit model.) 
Therefore, since $w=2n+2$, a single time-slice can be emulated on the spin-chain computer using $\sim 3w/2$ physical spins and $\sim 3w$ clock-ticks, and all $d$ time-slices are emulated in $\sim 3w \cdot d$ clock-ticks.
\end{proof} 

Since this upper bound is polynomial, we declare the emulation to be efficient.  Moreover, it is essentially optimal, due to the following lower bound.

\medskip
\begin{proposition}
  Any emulation that seeks to encode a nearest-neighbour circuit $C$ of width $w$ and depth $d$ obliviously into a list of operations to be performed on a constant-sized `window' in some architecture must require the list in question to be at least $\Omega(w\cdot d)$ long in the worst case.
\end{proposition}

\begin{proof}
An oblivious encoding must encode each gate of the circuit into the list.  We assume as before that each gate of $C$ is either a $\Lambda(Z)$ between neighbouring qubits or else a single qubit unitary drawn from a constant alphabet of single qubit unitaries.  Then it takes $\Omega(w)$ bits of information to describe each time-slice of $C$ in the worst case, \ie{} when $C$ is densely packed with gates.  Assuming for a moment that the elements on the list are to be drawn from a constant alphabet, so that at most $O(1)$ data can be fed `through the window' each clock-tick, it will require $\Omega(w)$ of them to represent the time-slice, and hence $d \cdot \Omega(w)$ to represent the entire circuit.

If, however, the elements on the list are \emph{not} drawn from a constant alphabet, but instead the size of the alphabet grows with $w$ or $d$ even though the size of the data structure at the `window' remains constant, then the elements of the alphabet will tend to come arbitrarily close to one another, because the space of bounded operators on a finite dimensional Hilbert space is compact.  This means that different patterns of elements cannot be obliviously simulating different circuits after all, so that the emulation strategy must break down after some finite point.  Therefore this case need not be further analysed.
\end{proof}

Suppose now we drop the nearest-neighbour conditions.  How do the upper and lower bounds change?

\medskip
\begin{proposition}
  Let $C$ be a quantum circuit on a line of $w$ qubits, composed of $\Lambda(Z)$ gates (not necessarily nearest-neighbour) and single-qubit unitaries arranged into time-slices.  Let $d$ be the depth of $C$.  The emulation of $C$ on the spin-chain computer now requires $O(w^3 \cdot d)$ time.
\end{proposition}

\begin{proof}
If we begin by transferring each $\Lambda(Z)$ gate into its own time-slice, this adds a factor of $O(w)$ to the time cost in the worst case, \ie{} when time-slices tend to start with $O(w)$ two-qubit gates in them.  Then each $\Lambda(Z)$ can be unpacked in the usual fashion into a product of nearest-neighbour $\Lambda(Z)$ gates interwoven with appropriate single qubit $H$ gates.  This unpacking increases depth by another factor of $O(w)$, and now the previous Proposition applies.
\end{proof}

\medskip
\begin{proposition}
  Any emulation that seeks to encode a circuit $C$ of width $w$ and depth $d$ obliviously into a list of operations to be performed on a constant-sized `window' in some architecture must require the list in question to be at least $\Omega(w\cdot d \cdot \log(w))$ long in the worst case.
\end{proposition}

\begin{proof}
With the assumptions of before regarding use of constant alphabets and oblivious encodings, the amount of information contained in a time-slice of $C$ must be $\Omega( w \cdot \log(w) )$, because the first qubit could be involved in a gate with any of $w-1$ later qubits, then the next qubit with any of $w-3$ later qubits, and so on.
Since there are $d$ time-slices, the total amount of information that needs to be passed `through the window' is $d \cdot \Omega(w \cdot \log(w))$ in the worst case, and an oblivious emulation---by definition---knows no way of improving upon this.  Thus the list in the emulation must involve $\Omega(w \cdot \log(w) \cdot d)$ operations in the worst case.    
\end{proof} 

This gap between the upper and lower bounds in the latter case suggests that our upper bounds strategy may be \naive, and not asymptotically optimal.

\cleardoublepage

\chapter{Probabilistic and Mixed Computing}
\label{chap:PMC}

Non-determinism in a general sense refers to the idea that there might be more than one `path' that a physical process can/does/might/could take~: no unique path need be determined.
In this chapter, we consider the role of non-determinism in computation, focussing on classical probability distributions, computations involving mixed states, and the non-operational concept of post-selection.  Our main technical contribution (\S\ref{sect:DQC1}) is to show that  in log-space ($\L$) one can produce a quantum circuit that uses only one pure qubit and solves a $\ParL$-complete problem, thereby generalising work of \cite{lit:ASV06}.  But we begin with a more abstract discussion of probability in computing to motivate the definitions used in \S\ref{sect:DQC1}, and also take the opportunity to introduce a new way of thinking about post-selection (\S\ref{sect:postselection}) that will have some relevance in Chapter~\ref{chap:IQP}.

\section{Operational Approach to Probabilistic Computing}  \label{sect:op}

This section just recalls some standard definitions and lemmata relevant to computation with probability distributions, extending some of the discussion of Chapter~\ref{chap:approach}.  
Definition~\ref{def:boundprob}---for Bounded Probability decision languages with arbitrary post-processing---may be seen as an abstract generalisation of standard definitions for classes such as $\BPP$, $\BQP$, \etc

\subsection{Elementary definitions}  \label{sect:eldef}

\subsubsection*{Probability Distributions}

We use discrete probability distributions to model the classical data output of physical processes that are designed for computation.  A discrete probability distribution may be construed as a function, $P$, having countable domain, mapping to the interval $[0,1]$.  It is \emph{stochastic}, which simply means that the sum over the whole domain must converge to 1.  For our purposes, it will be appropriate generally to take the domain to be the set of all finite binary strings, which is denoted $\{0,1\}^*$.
We write $P(\cL)$ as shorthand for $\sum_{x \in \cL \cap dom(P)} P(x)$.

\medskip
The direct product of distributions corresponds to the physical notion of running experiments independently in parallel and considering their combined output.
\begin{definition}
If $P$ and $Q$ are two distributions, then 
\begin{eqnarray*}
  P \otimes Q  &:&  x,y  ~\mapsto~  P(x) \cdot Q(y).
\end{eqnarray*}
Also, write $P^{\otimes k}$ to denote the direct product of $k$ copies of $P$.
\end{definition}

\medskip
The standard way of describing the distance between two probability distributions is to use $l_p$ additive gaps~:
\begin{definition}  \label{def:lpaddgap}
  For $p \in [1,\infty]$, the \emph{$l_p$ additive gap} between distributions $P$ and $Q$ is given by
\begin{eqnarray*}
  || P - Q ||_p  &:=&  \left(~ \sum_x | P(x) - Q(x)|^p ~\right)^{\frac1p},
\end{eqnarray*}
where the sum is taken over the union of the two domains.
In the case $p=\infty$, a limit is taken.
\end{definition}
The case $p=1$ is called the \emph{statistical distance} (or \emph{total variation distance}, up to scaling).  It has a special interpretation that makes it useful for defining the \emph{Bounded Probability} decision classes.

\subsubsection*{Operational nature of the statistical distance}

Here are some basic comments regarding the statistical distance~:
\begin{proposition}  \label{propos:oneadd}
Let $D$ be the union of domains of $P$ and $Q$, and let $k$ be a positive integer.
\begin{eqnarray*}
  \frac12 \cdot || P - Q ||_1  &=&  \max_{\cL \subseteq D} ~ P(\cL) - Q(\cL). \\
  || P^{\otimes k} - Q^{\otimes k} ||_1  &\le&  k \cdot || P - Q ||_1.
\end{eqnarray*}
\end{proposition}

\begin{proof}
The proof of the first line is elementary from the definition.  
The second line follows from elementary induction together with the basic inequality
\begin{eqnarray*}
  |ac-bd|  ~~=~~  |ac - bc + bc - bd|  &\le&  |ac-bc| + |bc-bd|  ~~\le~~ |a-b| + |c-d|
\end{eqnarray*}
whenever $a,b,c,d \in [0,1]$.
\end{proof}

\medskip
In the theory of computation, we usually wish to post-process samples from a probability distribution, in order to make a decision and complete a computation.  To avoid encoding complexity in the post-processing phase, it is appropriate to use some simple structure, such as some decision language $\cL$ in some `simple' class (\eg{} $\L$), to compress a probability distribution down onto just two points~:
\begin{definition}  \label{def:twopoint}
  Let $P$ be a probability distribution with domain $D \subseteq \{0,1\}^*$, and let $\cL \subseteq \{0,1\}^*$ be some fixed decision language.  Define the \emph{fully post-processed} two-outcome distribution $\hat{P}_\cL : \{\top,\bot\} \rightarrow [0,1]$ as follows.
\begin{eqnarray*}
  \hat{P}_\cL(\top)  &:=&  P( D \cap \cL ),  \\
  \hat{P}_\cL(\bot)  &:=&  P( D \backslash \cL ).
\end{eqnarray*}
\end{definition}

\medskip
\begin{proposition}
Let $Coin$ be an independent random coin.  For any two-outcome distribution $\hat{P}$,
\begin{eqnarray*}
  || \hat{P} - Coin ||_1  &=&  | \hat{P}(\top) - \hat{P}(\bot) |.
\end{eqnarray*}
The value $(\hat{P}(\top) - \hat{P}(\bot))$ is called the \emph{bias} of $\hat{P}$; and so we see the magnitude of the bias of $\hat{P}$ is given by its statistical distance from a random coin.
\end{proposition}

\begin{proof}
Simply consider $|\hat{P}(\top)-\frac12| + |\hat{P}(\bot)-\frac12|$.
\end{proof}

Putting these two ideas together, we immediately see that a non-negligible bias in $\hat{P}$ is \emph{necessary} if we are to use a reasonable number of copies of $\hat{P}$ to magnify that bias to something substantial.  (That is, if $\hat{P}$ is very close to a random coin, then $\hat{P}^{\otimes k}$ will also be close to random.)  
It is a simple corollary of the Hoeffding inequality that a non-negligible bias is also \emph{sufficient} for bias amplification, as shown in the following Lemma~:
\begin{lemma}[Chernoff/Hoeffding]  \label{lem:amplify}
  If $\hat{P}$ is a two-outcome distribution with bias $b$, then one can straightforwardly post-process $\hat{P}^{\otimes k}$---using a majority vote, for odd $k$---to obtain a new distribution whose bias has the same sign as $b$ and magnitude at least $1 - 2e^{-k|b|^2/2}$.
\end{lemma}

\begin{proof}
The proof follows directly from Hoeffding's standard inequality \cite{lit:Hoeffding} applied to Bernoulli trials.
\end{proof}

\subsubsection*{Families of distributions lead to decision languages}

When we come to consider not a single physical experiment but a whole family of them, we start to think about families of probability distributions also. 

The following definition gives a useful way of creating semantic decision languages directly from families of distributions, using the same `operational' idea.  Decision languages can thus be derived from a family of probability distributions directly, without reference back to the underlying machine or process that takes samples from the distributions. 

\begin{definition}  \label{def:boundprob}
  Let $\cP = \{P_i ~:~ i \in I\}$ be a family of probability distributions, indexed by some totally ordered indexing set $I$.  Let $\cL$ be some fixed decision language.  Let $\bar{c}$ be a constant in $(0,\frac12)$.  
For every index $i \in I$, the value $P_i(\cL) = \hat{P_i}_\cL( \top )$ lies in one of the three partitions $[0,\bar{c}]$, $(\bar{c},1-\bar{c})$, or $[1-\bar{c},1]$; and we can tri-partition the set of all indices accordingly.  If the middle partition turns out to be empty, then we define the semantic decision language $\Bb{\cP}{\cL}$ to be the third partition~:
\begin{eqnarray*}
  \Bb{\cP}{\cL}  &:=&  \{~ i \in I ~:~ P_i( \cL ) \ge 1-\bar{c} ~\},  
\end{eqnarray*}
which is in fact independent of $\bar{c}$ whenever defined.
\end{definition}
(Note that this language is a subset of $I$, therefore, it is appropriate in some circumstances to take $I$ to be $\{0,1\}^*$.  But it is also often convenient to have it be~$\NN$.)

This definition can be used as an alternate way of constructing classes such as $\BPP$ and $\BQP$.  For example, a generic $\BPP$ decision language can be defined in the form $\Bb{\cP}{\cL}$ by fixing some polynomial-time randomized Turing machine $\cM$ and taking $P_i$ to be the distribution of the output string of $\cM$ on input the string $i$; while $\cL$ could simply be the set of all strings that begin with a 1.  The computational power of polynomial-time Turing machines is sufficiently great that one need encode no `complexity' into $\cL$ in order to have this $\Bb{\cP}{\cL}$ be a `powerful' class.  The language $\cL$ serves as a kind of post-processor for the probability distributions, and for those families of distributions that are significantly weaker than $\BPP$ ones, allowing some additional post-processing of a comparatively complex nature can perhaps provide a significant boost to the complexity of the ensuing language.

\subsection{Philosophy of simulation}

This short section sets up some background context for the rest of the dissertation, which is about computing paradigms that are not necessarily universal for $\BQP$.

Let us look again at the constructions of the previous section.
A family $\cP$ of probability distributions is more likely to be of general interest if there is (at least conceptually) some programme of physical experiments whereby ``the $i$th experiment in the programme draws a sample from $P_i$''.  And it will be of even greater interest if the \emph{resources} required to implement the $i$th experiment scale efficiently in the complexity of (the size of) $i$.  
Designs for quantum computers that are severely limited---ones for which there is no apparent oblivious strategy for simulating arbitrary quantum circuits---can be analysed by modelling their output as a family of probability distributions.

This perspective leads to a nice way of thinking about simulation.  Imagine two programmes of physical experiments, one ($\cP$) whereby the $i$th experiment draws samples from $P_i$, and one ($\cQ$) whereby the $i$th experiment draws samples from $Q_i$.  (If you care to, you might also suppose that our lab technicians assure us that neither of these programmes uses very much more equipment or time than the other.)  Under what circumstances can we reasonably say that the two programmes simulate one another?  If $P_i = Q_i$ for all $i$, then the simulation is \emph{exact}.  When simulation is not exact, we need to quantify ``how $\cP$ is unlike $\cQ$''.  
To quantify the difference between $P_i$ and $Q_i$, one could use one of the $l_p$ measures of additive gap (Def.~\ref{def:lpaddgap}), and we have already seen that the statistical distance ($l_1$) is the most operationally relevant.
Then one must choose whether to be concerned with the \emph{worst case} for $i$, conisdering $\max_i || P_i - Q_i ||$;  or some kind of \emph{asymptotic case}, $\limsup_{i \rightarrow \infty} || P_i - Q_i ||$ perhaps;  or else some kind of \emph{average case} measure. 

The following Proposition about asymptotic similarity relates the language definition back to the notion of statistical distance.
\begin{proposition}  \label{propos:thing}
  If $\cL$ is a decision language, and $\cP$ and $\cQ$ are two families of probability distributions for which both $\Bb{\cP}{\cL}$ and $\Bb{\cQ}{\cL}$ are defined, then the following implication is valid.
\begin{eqnarray*}
  \lim_{i \rightarrow \infty} ||P_i - Q_i||_1 = 0 
    &\Rightarrow&  \Bb{\cP}{\cL} \sim \Bb{\cQ}{\cL},
\end{eqnarray*}
where the relation $\sim$ denotes set equality up to finite difference.
\end{proposition}

\begin{proof}
Let $\bar{c}_p$ and $\bar{c}_q$ be the two constants in $(0,\frac12)$ used in the definitions of $\Bb{\cP}{\cL}$ and $\Bb{\cQ}{\cL}$ respectively.  
Then 
\begin{eqnarray*}
  \lim_{i \rightarrow \infty} ||P_i - Q_i||_1 = 0  
  &\Rightarrow&  \max_S P_i(S) - Q_i(S) \rightarrow 0 \\
  &\Rightarrow&  | P_i(\cL) - Q_i(\cL) | ~=:~ r_i \rightarrow 0.
\end{eqnarray*}
Now if $i \in \Bb{\cQ}{\cL} \backslash \Bb{\cP}{\cL}$ then $P_i(\cL) \le \bar{c}_p$ and $Q_i(\cL) \ge 1-\bar{c}_q$.  This means that $r_i \ge 1 - \bar{c}_q - \bar{c}_p > 0$.  But if $r_i$ is to tend to zero, then it can only take values above the positive constant $1 - \bar{c}_q - \bar{c}_p$ finitely often, and so the symmetric difference of the two languages must be finite.
\end{proof}

\section{Post-selection}  \label{sect:postselection}

This section discusses the idea of \emph{post-selection}, which is a non-operational concept.  It is the idea that when making experiments of a probabilistic nature, one might focus on those instances whose outcomes satisfy a certain condition, and then analyse the outcomes as though those were the only instances.  Of course, the post-selection condition may itself be exceptionally rare, which is why the corresponding decision languages tend to be very large and `non-operational' in nature.  

We consider that post-selection is a useful conceptual tool to have in mind when looking at paradigms for quantum computation, and also potentially for proving results about classical compexity classes (\cf{} \cite{lit:Aa04}).  Notation is introduced here, mirroring \S\ref{sect:op} as much as possible, but no use is made of these ideas until Chapter~\ref{chap:IQP}.

Following what we did in \S\ref{sect:op}, we begin with a definition analogous to the \emph{statistical distance}, but for post-selective concepts.  Then we give a definition that compresses probability distributions down to two-point distributions, this time using post-selective post-processing rather than ordinary post-processing.  We proceed by considering again families of distributions, and we discuss what these ideas might mean for simulation.
Definition~\ref{def:postprob}---for Post-selected decision languages with arbitrary post-processing---may be seen as an abstract generalisation of standard definitions for classes such as $\BPP_{path}$, $\mathbf{PostBQP}$, \etc

\subsubsection*{Non-standard distance measures}

The following non-standard measure\footnotemark{} for gaps between probability distributions is offered as a candidate for the analogue of the statistical distance in a post-selective context, justified by its use in Proposition~\ref{propos:thing2}.
\footnotetext{It is defined similarly to the Renyi Information Divergence, but with important differences.}
\begin{definition}[\emph{Cf.} Definition~\ref{def:lpaddgap}]
  For $p \in [1,\infty]$, the \emph{$l_p$ multiplicative gap} between distributions $P$ and $Q$ is infinite if $P$ and $Q$ have different support; otherwise it is given by
\begin{eqnarray*}
  || P / Q ||_p  &:=&  \left(~ \sum_x | \log P(x) - \log Q(x)|^p ~\right)^{\frac1p},
\end{eqnarray*}
where the sum is taken over the (mutual) support.
In the case $p=\infty$, a limit is taken.
\end{definition}
Like the additive gap measures defined earlier, these multiplicative gap measures are symmetric, in that $||P/Q||_p = ||Q/P||_p$.

\subsubsection*{The post-selective distance}

Here is an elementary remark on the case $p=\infty$, which we henceforth dub the \emph{post-selective distance}~:
\begin{proposition}[\emph{Cf.} Propos.~\ref{propos:oneadd}]  \label{propos:infmult}
  Let $P$ and $Q$ be two distributions with the same support $D$, and let $k$ be a positive integer.
\begin{eqnarray*}
  || P/Q ||_\infty  &=&  \max_{\cL \subseteq D} \left| \log \frac{P(\cL)}{Q(\cL)} \right|.  \\
  || P^{\otimes k}/Q^{\otimes k} ||_\infty  &=&  k \cdot || P/Q ||_\infty.
\end{eqnarray*}
\end{proposition}

\begin{proof}
  In the limit $p \rightarrow \infty$, the definition immediately tells us that $|| P/Q ||_\infty = \max_x | \log \frac{P(x)}{Q(x)} |$.  But $\left| \log \frac{P(\cL)}{Q(\cL)} \right|$ is maximal when $\cL$ contains only the singleton $x$ that maximises this expression.  The second line also follows from the same observation.    
\end{proof}

\subsubsection*{Post-selection}

As before, we wish to compress a probability distribution down onto two points, so as to imply a decision.  But this time, we condition on some specific type of outcome before taking that decision.  So it is necessary to use a \emph{nested pair} of decision languages $\cL_S \subseteq \cL_C$ to compress down onto two points, as follows~:
\begin{definition}[\emph{Cf.} Definition~\ref{def:twopoint}]
  Let $P$ be a probability distribution with support $D \subseteq \{0,1\}^*$, and let $\cL_S \subseteq \cL_C \subseteq \{0,1\}^*$ be some fixed nested pair of decision languages.  If $P(\cL_C ) \not= 0$ then define the \emph{fully post-selected renormalised} two-outcome probability distribution $\tilde{P}_{\cL_S \subseteq \cL_C} : \{ \top, \bot \} \rightarrow [0,1]$ as follows.
\begin{eqnarray*}
  \tilde{P}_{\cL_S \subseteq \cL_C}(\top)  &:=&  \frac{P(\cL_S)}{P(\cL_C)}, \\
  \tilde{P}_{\cL_S \subseteq \cL_C}(\bot)  &:=&  \frac{P(\cL_C \backslash \cL_S)}{P(\cL_C)}.
\end{eqnarray*}
\end{definition}
This interpretation of ratios of probabilities as conditionals is due to Bayes's Theorem.

As before, the \emph{bias} of $\tilde{P}$ can be defined as $\tilde{P}(\top) - \tilde{P}(\bot)$.

Note that it is possible to amplify the bias of $\tilde{P}_{\cL_S \subseteq \cL_C}$ by taking its $k$-fold product and appealing to the majority-vote method of Lemma~\ref{lem:amplify}.  In fact, we can do slightly better in the post-selection case, to obtain a better amplification rate.

\medskip
\begin{lemma}[\emph{Cf.} Lemma~\ref{lem:amplify}]
  If $P$ is some probability distribution, and $\cL_S \subseteq \cL_C$ are languages, and $\tilde{P}_{\cL_S \subseteq \cL_C}$ has bias $b$, then for any positive integer $k$ of our choosing, we could take $Q = P^{\otimes k}$ and take modified languages $\cL'_S \subseteq \cL'_C$ such that the bias of $\tilde{Q}_{\cL'_S \subseteq \cL'_C}$ has the same sign as $b$ and has magnitude at least $1 - e^{-k|b|}$.
\end{lemma}

\begin{proof}
To prove this, take $\cL'_S = (\cL_S)^{\otimes k}$ and $\cL'_C = (\cL_C \backslash \cL_S)^{\otimes k} \cup (\cL_S)^{\otimes k}$.
Plugging these into the definition, the new bias is seen to be
\begin{eqnarray}
  b'  &=&  \frac{\tilde{P}(\top)^k - \tilde{P}(\bot)^k}{\tilde{P}(\top)^k + \tilde{P}(\bot)^k} 
     ~~=~~ \frac{ (1+b)^k - (1-b)^k }{ (1+b)^k + (1-b)^k }.
\end{eqnarray} 
Now for $0 \le b < 1$, we need only show
\begin{eqnarray}  \label{eqn:dansineq}
  2 e^{kb} &\le&  1 + \left(\frac{1+b}{1-b}\right)^k;
\end{eqnarray} 
the case $-1 < b \le 0$ is symmetrically the same.

Line~(\ref{eqn:dansineq}) can be established analytically.  It clearly holds at $b=0$.  Then the first derivative of the left side is $2k e^{bk}$, while the first derivative of the right side is $\frac{2k}{1-b^2} (\frac{1+b}{1-b})^k$, so it suffices to show (for positive $b$ up to 1)
\begin{eqnarray}
  2ke^{kb} &\le&  \frac{2k}{1-b^2} \cdot \left(\frac{1+b}{1-b}\right)^k.
\end{eqnarray}
This is easily seen for $k=0$, and for other (positive) values of $k$ it suffices if
\begin{eqnarray}
  e^b  &\le&  \frac{1+b}{1-b}.
\end{eqnarray}
This last line follows immediately (term by term) from the power series expansions.
\end{proof}

\subsubsection*{Families of distributions, post-selected}

The following definition is for post-selective classes of decision languages.
\begin{definition}[\emph{Cf.} Definition~\ref{def:boundprob}]  \label{def:postprob}
  Let $\cP = \{P_i : i \in I \}$ be a family of probability distributions, indexed by some totally ordered indexing set $I$.  Let $\cL_S \subseteq \cL_C$ be a nested pair of decision languages.  Let $\bar{c}$ be a constant in $(0,\frac12)$.  For every index $i$ for which $P_i( \cL_C ) \not= 0$, the ratio $P_i( \cL_S )/P_i( \cL_C )$ lies in one of the three partitions $[0,\bar{c}]$, $(\bar{c},1-\bar{c})$, or $[1-\bar{c}, 1]$; and we can tri-partition the set of all such indices accordingly.  If this works for \emph{all} $i$, and the middle partition turns out to be empty, then we define the semantic decision class $\Postb{\cP}{\cL_S \subseteq \cL_C}$ to be the third partition~:
\begin{eqnarray*}
  \Postb{\cP}{\cL_S \subseteq \cL_C}  &:=&
    \{~ i \in I ~:~ P_i( \cL_S ) \ge (1-\bar{c}) \cdot P_i( \cL_C ) ~\},
\end{eqnarray*}
which is in fact independent of $\bar{c}$ whenever defined.
\end{definition}

Although not operationally relevant, this definition is just as general as the earlier one for \emph{bounded probability}, again making no reference to the origin or complexity of the distributions in question.

This definition can be used as an alternate way of constructing classes such as $\BPP_{path}$ or $\mathbf{PostBQP}$.  For example, a generic $\BPP_{path}$ decision language can be defined in the form $\Postb{\cP}{\cL_S \subseteq \cL_C}$ by fixing some polynomial-time randomized Turing machine $\cM$ and taking $P_i$ to be the distribution of the output string of $\cM$ on input the string $i$; while $\cL_S$ and $\cL_C$ could simply be the sets of all strings beginning ``11\ldots'' and ``1\ldots'' respectively.

\subsubsection*{Non-operational simulation}

It will not have escaped the reader's notice that we have tried to make our post-selective constructions and discussions of section~\ref{sect:postselection} follow a parallel course to the constructions and discussions of section~\ref{sect:op}.
And so it remains to prove one more analogous Proposition.
 
\begin{proposition}[\emph{Cf.} Propos.~\ref{propos:thing}]  \label{propos:thing2}
  If $\cL_S \subseteq \cL_C$ are decision languages, and $\cP$ and $\cQ$ are two families of probability distribution for which both $\Postb{\cP}{\cL_S \subseteq \cL_C}$ and $\Postb{\cQ}{\cL_S \subseteq \cL_C}$ are defined, then the following implication is valid.
\begin{eqnarray*}
  \lim_{i \rightarrow \infty} || P_i/Q_i ||_\infty = 0
    &\Rightarrow&  \Postb{\cP}{\cL_S \subseteq \cL_C} \sim \Postb{\cQ}{\cL_S \subseteq \cL_C},
\end{eqnarray*}
where the relation $\sim$ denotes set equality up to finite difference.
\end{proposition}

\begin{proof}
  Let $\bar{c}_p$ and $\bar{c}_q$ be the two constants in $(0,\frac12)$ used in the definitions of the post-selective languages $\Postb{\cP}{}$ and $\Postb{\cQ}{}$ respectively.
  Since $\max_\cL \left| \log \frac{P_i(\cL)}{Q_i(\cL)} \right| \rightarrow 0$ by Proposition~\ref{propos:infmult}, it follows that there must be some real sequence $r_i$, tending to 1 from above, such that the values $\frac{P_i( \cL_S )}{Q_i( \cL_S )}$ and $\frac{P_i( \cL_C )}{Q_i( \cL_C )}$ both lie in $[1/r_i, r_i]$.  Therefore, every time $i \in \Postb{\cQ}{} \backslash \Postb{\cP}{}$, it follows that 
\begin{eqnarray*}
  r_i^2  &\ge&  \frac{Q_i(\cL_S)}{Q_i(\cL_C)} \cdot \frac{P_i(\cL_C)}{P_i(\cL_S)}
        ~~\ge~~  \frac{1-\bar{c}_q}{\bar{c}_p}.
\end{eqnarray*}
This means that $r_i \ge \sqrt{(1-\bar{c}_q)/\bar{c}_p} > 1$, which can occur only finitely often since $r_i \rightarrow 1$.  Likewise for the case $i \in \Postb{\cP}{} \backslash \Postb{\cQ}{}$, and so the symmetric difference of the two languages must be finite.
\end{proof}

\subsubsection*{Synopsis}

The key idea for our understanding of post-selection concepts is that whenever an operational paradigm for quantum computing is proposed, it may be a hard or impossible task to show that it is universal for $\BQP$, but it may be much easier to show that its natural post-selective variant (with suitable limitations on $\cL_S$ and $\cL_C$) is universal for $\PP$ (Aaronson \cite{lit:Aa04} has shown this to be equal with $\mathbf{PostBQP}$).
Since there is no known way to establish the equivalence of $\BPP_{path}$ with $\PP$, any such demonstration of post-selective $\PP$-universality is tantamount to a proof that there will be no oblivious classical simulation strategy for the operational version of the paradigm. 
This means that one can argue for `genuinely quantum computational effects' or `intractibility of simulation' without needing a full-blown $\BQP$-powerful architecture.
We will put this into practice in Chapter~\ref{chap:IQP}.

\section{Computing with Mixed States}  \label{sect:DQC1}

In this section, we consider formalising the notion of ``computing with just one pure qubit''; a paradigm that was introduced in \cite{lit:KL9812}, and further investigated in \cite{lit:ASV06}.
In this paradigm, arbitrary quantum circuits are allowed, but the inupt to the quantum circuit must have very limited purity.  Moreover, the input is fixed, and so cannot be used to index the elements of a decision language.
We shall explore what can be done with uniform families of quantum circuits in this paradigm, introduce a particular model\footnotemark{} and class of decision languages, and argue for why the model is aptly described using the terminology of the present Chapter.
\footnotetext{I posted research notes on this subject on the arXiv in 2006, but didn't pursue publication of them at the time.}

\subsection{Overview}

In \cite{lit:KL9812}, Knill and Laflamme considered an extreme limitation on state purity by asking for computations that have only one pure qubit, the other qubits being fully depolarised.  This paradigm they called ``DQC1''.\footnotemark{}
\footnotetext{The `D' here stands for `deterministic', which is being used to mean what I have elsewhere termed `operational'.}
They asked about what can be computed if one allows arbitrary quantum circuits of polynomial length, taking input as described above, and measuring a single qubit to obtain an output of the computation.

It is convenient to denote mixed states using the \emph{density operator} formalism.  A density operator is essentially\footnotemark{} the quantum generalisation of a probability distribution.  
This operator encodes all the information about the propensity for how a state will behave in relation to all possible measurements.  (With respect to a computational basis, the diagonals of a matrix representation of such an operator correspond to an actual probability distribution.)
\footnotetext{Under the Everettian interpretation, density operators are taken to describe objective states rather than subjective states of knowledge, but just as the nature of \emph{probability} is philosophically ambiguous, the same can be said for density operators.}  
Accordingly, using a discrete TIME model and a finite-dimensional unitary space, one can take state space for mixed quantum computation as the convex hull (in the space of linear functions) of rank-1 Hermitian projectors on the unitary space, rather than the unitary space itself. 

After outlining prior work on this subject, we shall introduce a more formal way of expressing operationally relevant decision languages that naturally belong to this paradigm, and show that the ``one pure qubit'' analogue of $\BQP$ contains the class $\ParL$.  We shall also provide oracle separations, both ways (one of which is new), between this class and $\P$.

\subsection{Prior work}  \label{sect:DQCPrior}

The physical motivation for the DQC1 paradigm comes from architectures based on Nuclear Magnetic Resonance (NMR), where purity of quantum state is hard come by.
For NMR computing, the mixed state that one is forced to work with has its mixedness spread across all qubits, so that they are initialised in a `hot' state of the form
\begin{eqnarray}
  \rho &=& \left(  \frac{1+\eps}{2}\ketbra00 + \frac{1-\eps}{2}\ketbra11  \right)^{\otimes n}.
\end{eqnarray}
But using an analogy from thermodynamics, in \cite{lit:SV99} it is shown how to build an \emph{efficient unitary circuit} that `distils' out purity with high probability, leaving a state that is close to $\ketbra00$ on the first $n - n \cdot H(\eps) -o(1)$ qubits.  (Here $H$ measures the entropy of the state, and so the limit is effectively tight.)
And so (\cf{} \cite{lit:ASV06}), provided such transformations are reasonable within one's computational model, it is no loss of generality to restrict one's attention to the more `digital' perspective whereby the initialisation state (as a density operator in an algebra of $w$ qubits) is taken to be
\begin{eqnarray}  \label{eqn:rhostart}
  \rhoS(k,w)  &=&  \frac{ \left(1+Z\right)^{\otimes k} }{ 2^w },
\end{eqnarray}
where $w$ counts the total number of qubits and $k \le w$ counts the number which are pure, with the standard Pauli operators being used to describe density operators. 
(For subscript notation throughout this section, we take qubits $[1..k]$ to be the ones initially pure, and $[k+1..w]$ to be the ones initially depolarised.)

If one begins with state $\rhoS(1,w)$, applies an arbitrary unitray map $W_{[1..w]}$ across all qubits, and then measures the first qubit in the computational basis, one obtains $\ket0$ with \emph{bias}
\begin{eqnarray}  \label{eqn:DQC1bias}
  Tr[~ \left( W_{[1..w]} \cdot \rhoS(1,w) \cdot W_{[1..w]}^\dag \right) \cdot Z_1 ~]  
  \phantom{XXXX}\\
  =~~  2^{-w} \cdot Tr[~ W_{[1..w]} \cdot Z_1 \cdot W_{[1..w]}^\dag \cdot Z_1 ~].
  \nonumber
\end{eqnarray}

Knill and Laflamme showed that this paradigm---even with just the one pure qubit---can be used to estimate the trace of a unitary operator, as outlined below.  Moreover, there is a sense in which the problem of trace estimation is complete for the paradigm \cite{lit:SJ08}.  Given a circuit for an arbitrary unitary $V$ on $w-1$ qubits, the unitary used to estimate the (real part of the) trace of $V$ is taken to be $W_{[1..w]} = H_1 \cdot \Lambda_1( V_{[2..w]} ) \cdot H_1$, because then the measurement bias is
\begin{eqnarray}  \label{eqn:tracest}
  2^{-w} \cdot Tr[~ W_{[1..w]} \cdot Z_1 \cdot W_{[1..w]}^\dag \cdot Z_1 ~]  &=&
  2^{1-w} \cdot Tr[~ Re[V] ~].
\end{eqnarray}
Such biases can be amplified in the usual fashion, by parallel instantiation and majority vote (\cf{} Lemma~\ref{lem:amplify}).  

Ambainis, Schulman, and Vazirani \cite{lit:ASV06} showed that the non-uniform version of the complexity class $\NC^1$ (classical, polynomial time, logarithmic circuit depth) is computable within a non-uniform model of the ``one pure qubit'' DQC1 paradigm, and also showed that there is no obvious efficient way (\ie{} no \emph{oblivious} technique) to simulate a circuit with $k$ pure qubits using fewer pure qubits, except at the cost of exponentially decaying efficiency.

Shor and Jordan \cite{lit:SJ08} discussed the differences between considering quantum circuits supplied by a polynomial-time classical computer and a more restricted classical computer computing only $\NC^1$.  They also showed that, even in the weaker model, having logarithmically many pure qubits is no better than having just one, provided it is understood that one is free to make polynomially many runs of DQC1-type experiments, with majority-vote post-processing, in order to make any specific decision.

\subsection{Decision languages for mixed states}  \label{sect:egDQC1}

What makes models within the DQC1 paradigm a little different from the usual notion of a computation model?
\begin{itemize}
  \item
One is not permitted to make intermediate measurements (or other non-unitary gates) during the execution of a circuit, since otherwise such operations could be used to introduce new purity into the system, effectively boosting its power back to that of universal $\BQP$ computing (\cf{}~\S\ref{sect:adaption}).
  \item
One cannot define decision languages in terms of the input into a unitary circuit in this model, because the quantum input is always constrained to be the one given at line~(\ref{eqn:rhostart}).  We shall see that this means that classical input must be interfaced via classical control of the circuit elements.
  \item
The computational power of the model is potentially affected by how many bits can be interfaced out of the computation at measurement time~: \eg{} $1$, $k$, or $w$?
\end{itemize}

The three points above must be addressed properly if one is to use the paradigm to define formally a class of decision languages.  But before we attempt such a definition (Def.~\ref{def:DQCk}), let us first give an algorithm for a $\ParL$-complete decision language~: the task of evaluating  one output bit of a polynomial-sized classical circuit composed entirely of C-Not gates.
(The decision language corresponding to this class is not believed to be in $\NC^1$, despite the fact that matrix multiplication over the field $\FF_2$ can be computed in logarithmic circuit depth \cite{lit:Damm90,lit:BDHM92}.)

\medskip
Let $\{C_i\}$ be a family of classical circuits composed entirely of C-Not gates, such that the number of bits input to $C_i$ is equal to $i$.  Let $\cL$ be the language of strings $x$ which, when input to the appropriate $C_i$, cause the first output bit to be 1~:
\begin{eqnarray}  \label{eqn:ParLlang}
  \cL  &=&  \{~ x \in \{0,1\}^* ~:~ i = len(x), \mbox{ first bit of } C_{i}(x) = 1 ~\}.
\end{eqnarray}

Now let $x$ be some particular input string, and let $i = len(x)$.
Let the $i$ bits of the string $x$ be denoted $x_2,x_3,\ldots,x_{i+1}$.
Let $s_2,s_3,\ldots,s_{i+1}$ be the output bits of $C_i(x)$, so that $s_2$ is the bit whose setting decides whether $x \in \cL$.
Suppose we wish to determine whether or not $x \in \cL$, using some DQC1-style computation.  To do this, we must specify a quantum circuit, denoted $W(x)$, designed in some appropriately uniform manner (relative to the uniformity of the family $\{C_i\}$, see below) that will be used to compute the value $s_2$.

Let $w=i+1$ measure the total number of qubits on which our circuit will act, so that it makes sense to apply our circuit $W(x)$ to the state $\rhoS(1,w)$ in accordance with line~(\ref{eqn:DQC1bias}).

Next, let $V_i(x)$ be the circuit on qubits $[1..w]$ that consists of one C-Not gate from qubit $1$ to qubit $j$ each time that bit $x_j$ is set.  This we notate
\begin{eqnarray}  \label{eqn:classiccontrol}
  V_i(x)_{[1..w]}  &:=&  \prod_{j=2}^w ~\Lambda_1(X_j)^{x_j}.
\end{eqnarray}
Informally, we say that this circuit will be used to interface the information contained within the string $x$ to the `DQC1 algorithm' that we are designing.  More formally, the incorporation of $V_i(x)$ as a `subroutine' within $W(x)$ is to be the \emph{only} way in which $W(x)$ depends on $x$, so that a sensible notion of uniformity applies to the family $\{W(x)\}$.

Let $Q_B$ denote the Hadamard gate being applied to \emph{every} qubit (see~\S\ref{sect:QFTdefs} for an explanation of this notation).
Let $C_i$ be recast as a \emph{quantum} circuit to be applied on qubits $[2..w]$.
Finally, let $U(x) := Q_B \cdot C_i \cdot V_i(x) \cdot Q_B$, and define $W(x)$ to be the overall circuit given by 
\begin{eqnarray}  \label{eqn:DQC1W}
  W(x)_{[1..w]}  &:=&  U(x)^\dag_{[1..w]} \cdot X_2 \cdot U(x)_{[1..w]}.
\end{eqnarray}

\medskip
\begin{lemma}  \label{lem:DQC1doesParL}
  When the circuit $W(x)$ defined above is applied to the state $\rhoS(1,w)$, the bias as specified at line~(\ref{eqn:DQC1bias}) is always either 1 or -1, and is -1 exactly when $x \in \cL$ as defined at line~(\ref{eqn:ParLlang}).  
\end{lemma}

\begin{proof}
We claim that the effect of $U(x)=U(x)_{[1..w]}$ on $\rhoS(1,w)$ is to map it to
\begin{eqnarray} \label{eqn:DQC1actU}
  U(x)_{[1..w]} \cdot \frac{1+Z_1}{2^w} \cdot U(x)^\dag_{[1..w]}  &=&
    \frac{1 + Z_1 \cdot Z_2^{s_2} \cdot Z_3^{s_3} \cdots Z_w^{s_w}}{2^w}.
\end{eqnarray}
Perhaps the easiest way to see why this claim holds is to regard $\rhoS(1,w)$ as being the proper uniform mix of pure states $\ket{0}_1\ket{r}_{[2..w]}$, where $r$ ranges over all $i$-bit strings.  
Write $R$ for $C_i(r)$.
Then the effect of $U(x)$ on $\ket0\ket{r}$ is readily seen to follow from
\begin{eqnarray}
  V_i(x) \cdot Q_B \ket0\ket{r}  
    &=&  2^{-w/2}\sum_y (-1)^{r \cdot y} \Bigl( \ket0 \ket{y} + \ket1 \ket{y \oplus x} \Bigr), \nonumber \\
  C_i \cdot V_i(x) \cdot Q_B \ket0\ket{r}  
    &=&  2^{-w/2}\sum_y (-1)^{r \cdot y} \Bigl( \ket0 \ket{C_i(y)} + \ket1 \ket{C_i(y \oplus x)} \Bigr), \nonumber \\ 
    &=&  2^{-w/2}\sum_y (-1)^{C_i(r) \cdot y} \Bigl( \ket0 \ket{y} + \ket1 \ket{y \oplus C_i(x)} \Bigr),  \\
  U(x) \ket0\ket{r}  
    &\propto& \Bigl( (1+(-1)^{R \cdot C_i(x)})\ket0 + (1-(-1)^{R \cdot C_i(x)})\ket1  \Bigr)\ket{R}. \nonumber 
\end{eqnarray}
The proper mix (over $r$) of these states must conform to the claim of line~(\ref{eqn:DQC1actU}), because application of $X_1$ causes each to map to something orthogonal, as does application of $X_j$ exactly when $s_j$ is set; whereas application of $Z_1$ or $Z_j$ causes no physical change.

Then we see (from line~(\ref{eqn:DQC1W})) that the action of $W(x)$ involves `computing' the $s_j$ bits in the sense of line~(\ref{eqn:DQC1actU}) by applying $U(x)$, then `kicking' the value of $s_2$ into the internal phase of the state by applying $X_2$, then finally `uncomputing' the $s_j$ bits by applying $U(x)^\dag$, so that we are left with 
\begin{eqnarray}
  W(x) \cdot \frac{1+Z_1}{2^w} \cdot W(x)^\dag  &=&  \frac{1 + (-1)^{s_2}Z_1}{2^w}.
\end{eqnarray}
The bias for this state (\cf{} line~(\ref{eqn:DQC1bias})) is $(-1)^{s_2}$, as required.
\end{proof}

The processing involved in the construction of Lemma~\ref{lem:DQC1doesParL} is achieved efficiently and deterministically (the final measurement returning a classical deterministic bit), using just one pure qubit, but the circuit $W(x)$ that provided the processing of data within the quantum memory required to incorporate two copies of $V_i(x)$.  That is, the algorithm required the ability to `read' the input bit-string twice, each time reading its bits in arbitrary order.

\subsection{$\DQC{k}$ and parity-control}  \label{sect:paritycontrol}

Here we offer a definition for a class of decision languages, based on the ideas used within the construction of Lemma~\ref{lem:DQC1doesParL}, but generalised to allow for computations that are not deterministic.

Besides the parameters $k$ and $w$ for determining the initial quantum state, we also need a parameter $i$ to determine the length of the classical string that will be used to control some of the gates within the circuit, which is the same string that the circuit is effectively `deciding' on.  And we need another parameter $b$ that describes the magnitude of the bias that the circuit must produce for \emph{all} valid inputs, since very tiny biases are not to be considered operationally significant (\cf{} \S\ref{sect:op}).
Parameters $k,w,b$ will all be taken to be functions of the argument $i$.

Finally, we need a sensible mechanism for describing \emph{how} the bits of the classical input string $x$ will control the circuit's gates.  It seems appropriate to adopt \emph{parity-control}, which means that if a gate is subject to classical control (\eg{} just as the gates at line~(\ref{eqn:classiccontrol}) depend on classical bits from $x$), it will be controlled by an $\FF_2$-linear function of the input $x$.   
\begin{definition}
  A gate $U$ is said to be under \emph{parity-control} from the \emph{input string} $x$ according to the \emph{control specification string} $c$ if the gate is applied (in its turn) when the circuit is executed if and only if the derived bit $c \cdot x$ should be set.  This parity-controlled gate is denoted $U^{c \cdot x}$.
\end{definition}
That is, a gate from a quantum circuit may have included within its description an arbitrary but explicit control specification string $c$, to describe how the string $x$ should affect whether or not the gate is to be applied.
The length of the control specification string $c$ should obviously match the length of the input string $x$, which is $i$. 
This device is a generalisation of the classical notion of a \emph{sequential branching program} studied in \cite{lit:Bar87} for example, although less directly related to the \emph{graphical $\Omega$-branching programs} discussed in \cite{lit:Damm90}. 

Note that when the number of pure qubits is not limited, there is a trivial reduction from `ordinary' quantum circuits (with classical input directly made quantum in the computational basis) to parity-controlled circuits with `null' quantum input $\ket{\mathbf0}$~: \viz, the first thing the parity-controlled circuit would do to simulate the ordinary circuit is to implement $X$ on qubit $i$, controlled by the parity of the single input bit $x_i$ (assuming of course that $X$ is amongst the allowable quantum gates).  Having done this, the rest of the simulating circuit would just proceed with the simulated circuit `uncontrolled' by classical input bits.  And so parity-controlled circuits can usefully be standardised in paradigms other than DQC1, particularly appropriate whenever one has no need for the concepts of \emph{circuit composition} and \emph{quantum communication}, or no notion of \emph{preprocessing} classical data before forming quantum data from it (\ie{} quantum input/output).

Here then is a definition for a DQC1-style complexity class, informed by the discussion above and by Definition~\ref{def:boundprob} of \S\ref{sect:op}.
\begin{definition}  \label{def:DQCk}
  Consider a uniform family of quantum circuits $\{W(i)\}_{i=1}^{\infty}$, some of whose gates may be under \emph{parity-control}.  
  Let $k=k(i) \le w=w(i)$ be a pair of polynomially bounded complexity functions, with $w(i)$ counting the width of $W(i)$.  Let $0<b=b(i)=\Omega(1/poly(i)) \le 1$ be another function.
  Then partition up the set of all $x \in \{0,1\}^*$ each according to which of the three sets
\begin{eqnarray*}
  [-1,-b], ~~(-b,b), ~~[b,1]
\end{eqnarray*}
contains the bias
\begin{eqnarray*}
  Tr\bigl[~ W(i)_{[1..w]} \cdot \rhoS(k,w) \cdot W(i)_{[1..w]}^\dag \cdot Z_1 ~\bigr],
\end{eqnarray*}
where the argument $i=len(x)$ is used throughout.
If the middle partition turns out to be empty (no string $x$ causes a negligible bias), then we define the semantic decision language $\cL_{k,W,w,b}$ to be the third partition~: 
\begin{eqnarray*}
  \cL_{k,W,w,b}  &:=&  \{~ x \in \{0,1\}^* ~:~ i = len(x),  ~Tr[ W(i) \cdot \rhoS(k,w) \cdot W(i)^\dag \cdot Z_1 ] \ge b ~\}.
\end{eqnarray*}
The class $\DQC{k}$ contains all such $\cL_{k,W,w,b}$ for that value of $k$.  (The union of \emph{all} these classes is clearly $\BQP$.)
\end{definition}

\medskip
This definition is based on Definition~\ref{def:boundprob}, but an important difference is that since the one-pure-qubit model has no apparent way to amplify bias \emph{within} the quantum part of computation, we instead allow for polynomially small bias rather than constant bias.

Here is the main result of this section~:

\begin{corollary}  \label{cor:DQC1main}
  $\ParL \subseteq \DQC{1}$.
\end{corollary}

\begin{proof}
Definition \ref{def:DQCk} clearly allows scope for our algorithm of Lemma~\ref{lem:DQC1doesParL} to ensure that a $\ParL$-complete language is contained within $\DQC{1}$.

Moreover, the result of Shor and Jordan \cite{lit:SJ08} about the utility of logarithmically many qubits not exceeding that of a single qubit likewise holds under this definition, with essentially no modification to their proof, so that $\DQC{log} = \DQC{1}$.

But it is trivial that $\L \subseteq \DQC{log}$, completing the argument.
\end{proof}

One may think of the structure $\{\DQC{k}\}_k$ as forming a hierarchy that reaches from the simplest model of the paradigm (one pure qubit) up to full $\BQP$ universality ($\DQC{poly}$).  In \cite{lit:ASV06}, it is shown that an `oblivious' simulation of a program in this hierarchy by a program much lower in the hierarchy is impossible; but now we see that a formal unconditional proof of this hierarchy's not collapsing would constitute an unconditional separation between $\DQC{1}$ and $\BQP$, and thence also imply an unconditional separation between (say) $\ParL$ and $\PP$ by Corollary~\ref{cor:DQC1main} (\cf{}~\cite{lit:Aa04}, and also~\S\ref{sect:openconj}).

\subsection{Oracle separations for $\DQC{1}$}

One can use the notion of an \emph{oracle}~(\S\ref{sect:circuit-interfaces}) to make a formal \emph{relativised} separation between complexity classes.  In this context, an oracle would take the form of a (non-uniform) family of permutations on the set $\FF_2^w$, supplied as so-called ``black-box unitaries'' or classically as ``black-box functions''.  

\medskip
\begin{proposition}  \label{prop:bb}
  There is a ``black-box'' oracle $\cO$ for which $\P^\cO \not\subseteq \DQC{1}^\cO$.
\end{proposition}

\begin{proof}
An example is given in \cite{lit:KL9812}, showing implicitly why certain `classically easy' facts about an oracle cannot be learned using only DQC1 methodology.  The same proof works for this Proposition, with only very minor changes.
\end{proof}

Simon's algorithm~\cite{lit:Si97} provides an oracle for establishing a separation of the form $\BQP^\cO \not\subseteq \BPP^\cO$.  With a small change, the same kind of oracle establishes the converse to Proposition~\ref{prop:bb}, as follows.
\medskip
\begin{proposition}
  There is a ``black-box'' oracle $\cO$ for which $\DQC{1}^\cO \not\subseteq \P^\cO$.
\end{proposition}

\begin{proof}
The oracle in Simon's algorithm \cite{lit:Si97} is based on randomly selected `hidden shift' functions, $f_n : \FF_2^n \rightarrow \FF_2^m$, with $f(\x)=f(\z) \Leftrightarrow \x+\z \in \{\0, \s\}$, for some random non-zero vector $\s$.
Instead, generalise this to have $f_n$ be a random function that is constant on cosets of some large subspace $S_n \le \FF_2^n$.  Further generalise by taking a similar function $g_n$ as a random function that is constant on cosets of some other large subspace $T_n$.  We shall also need that these functions are distinct on distinct cosets.

The oracle is considered to provide a family of such functions in the usual fashion, parameterised by $n$.  We can let $m$ be any polynomial function of $n$, so we'll pick $m=m(n)=2n$ for a concrete example.

Using these random functions, define the following permutation-unitaries on $w=w(n)=n+2m(n)$ qubits (where $\ket{\a,\b,\c} = \ket{\a}_{[1..n]}\ket{\b}_{[n+1..n+m]}\ket{\c}_{[n+m+1..n+2m]}$)~:
\begin{eqnarray}
  U_{[1..w]}  &:&  \ket{\a,\b,\c} ~~\mapsto~~ \ket{\a,\b+f_n(\a),\c}, \nonumber \\
  U'_{[1..w]} &:&  \ket{\a,\b,\c} ~~\mapsto~~ \ket{\a,\b,\c+g_n(\a)}, \nonumber \\
  V_{[1..w]}  &:=& \Bigl( H^{\otimes n}_{[1..n]} \cdot U_{[1..w]} \cdot H^{\otimes n}_{[1..n]} \cdot U'_{[1..w]} \Bigr)^2.  
\end{eqnarray}
Now let's evaluate the trace of $V$~:
\begin{eqnarray}
  2^{-w} \cdot Tr[V]  
  &=&  2^{-w} \sum_{\a\b\c} \bra{\a,\b,\c} H^{\otimes n} U H^{\otimes n} U' H^{\otimes n} U H^{\otimes n} U' \ket{\a,\b,\c} \\
  &=&  2^{-w-2n} \sum_{\a\b\c\x\y\z} \begin{array}{l}
                   \bra{\b + f_n(\x),\c + g_n(\y)} \\
                   \phantom{XX} 
                   (-1)^{(\a+\y) \cdot (\x+\z)} \ket{\b+f_n(\z),\c+g_n(\a)}. 
                                     \end{array} \nonumber
\end{eqnarray}
The only terms here that won't vanish are those whereby $(\x+\z) \in S_n$ and $(\a+\y) \in T_n$, using the fact that functions $f_n$ and $g_n$ are distinct on different cosets of $S_n$ and $T_n$ respectively, but otherwise constant.  So we make a change of variables, $\s = \x+\z$ and $\t = \a+\y$.
Then 
\begin{eqnarray}
  2^{-w} \cdot Tr[V]  
  &=&  2^{-3n-2m} \sum_{\b\c\x\y} \sum_{\s \in S_n, ~\t \in T_n} (-1)^{\s \cdot \t}  
  \nonumber \\
  &=&  2^{-n} \sum_{\s \in S_n, ~\t \in T_n} (-1)^{\s \cdot \t}.
\end{eqnarray}

If we are careful to ensure that the dimension of $S_n$ matches the codimension of $T_n$, so that $|S_n| \cdot |T_n| = 2^n$, then this expression further simplifies to 
\begin{eqnarray}
  2^{-w} \cdot Tr[V]  
  &=&  \left\{ \begin{array}{ccl} 1 & \mbox{if} & S_n ~\bot~ T_n \\
                                  0 & \mbox{if} & S_n \not{\!\bot}~ T_n 
                                  \end{array} \right..
\end{eqnarray}
One can use the trace-estimation algorithm (\S\ref{sect:DQCPrior}) to distinguish these two cases.  Since the trace-estimation algorithm requires implementing $\Lambda(V)$ once, and since each implementation of $V$ makes use of four oracle calls, it follows that four oracle calls are sufficient for distinguishing between the two cases of ``orthogonal cosets'' \emph{versus} ``non-orthogonal cosets''.

This quantum black box algorithm therefore solves a certain promise-problem, but can the same problem be solved classically efficiently?
No, because in the worst case, exponentially many samples of $f_n$ and $g_n$ are needed.  If the dimension and codimension of each of $S_n$ and $T_n$ is $\frac{n}{2}$, then the domains of each of $f_n$ and $g_n$ partition into $2^{n/2}$ different cosets, on which different values are taken.  There need be no other structure in $f_n$ and $g_n$, and so there is no efficient way even to find an element of $S_n$ or $T_n$.  We formalise this idea next by showing that if a classical algorithm were to sample each of $f_n$ and $g_n$ at any $2^{n/4}$ points each, then it would be possible that no two samples of $f_n$ were found to be the same and neither were two samples of $g_n$ the same, and moreover there would exist a consistent choice of $S_n$ and $T_n$ with $S_n \bot T_n$ as well as a different consistent choice with $S_n{\not\!\!\bot}T_n$.  Therefore the algorithm would fail; which establishes a classical (deterministic worst case) lower bound of $2^{n/4}$ queries required.

Suppose $2^{n/4}$ queries are made of $f_n$.  That amounts to $2^{n/4-1}( 2^{n/4} - 1 )$ \emph{pairs} of (unequal) points sampled, and the two samples of any pair being different is the same thing as the (non-zero) sum of those two points lying outside $S_n$.  
Now the number of non-zero points in $\FF_2^n$ is plainly $2^n-1$, and the number of non-zero points in any candidate subspace $S_n$ of dimension $\frac{n}{2}$ is $2^{n/2}-1$.  Therefore any point being declared to lie outside of $S_n$ denies a proportion $(2^{n/2}+1)^{-1}$ of the possibilities for $S_n$.  (Think of a bipartite graph between non-zero points of $\FF_2^n$ and subspaces of dimension $n/2$.)  Therefore our samples---if they do all turn out to be distinct---must certainly preclude fewer than \emph{half} of all candidate $S_n$ subspaces, since $2^{n/4-1}( 2^{n/4} - 1 ) \cdot (2^{n/2}+1)^{-1} \le \frac12$.  The same reasoning applies to $T_n$.

To each candidate $S_n$ there is precisely one $T_n$ (namely its dual) for which $S_n \bot T_n$ (and plenty of other $T_n$ for which $S_n{\not\!\!\bot}T_n$).  Since more than half of all possible $S_n$ and $T_n$ remain as candidates, it must be possible to find a pair such that $S_n \bot T_n$, as well as a pair for which $S_n{\not\!\!\bot}T_n$.  Since both possibilities are available, no deterministic algorithm having made $2^{n/4}$ queries can possibly solve the problem in the worst case. 
\end{proof}

\subsection{Probabilistic quantum polytime, $\PDQC{k}$}

For completeness, we can also define syntactic classes $\PDQC{k}$ in an analogous fashion, by dropping the requirement that the bias be non-negligible.

\begin{definition}
  In the terminology of Def.~\ref{def:DQCk},
\begin{eqnarray*}
  \cL'_{k,W,w}  &:=&  \{~ x \in \{0,1\}^* ~:~ i = len(x),  ~Tr[ W(i) \cdot \rhoS(k,w) \cdot W(i)^\dag \cdot Z_1 ] \ge 0 ~\}.
\end{eqnarray*}
The \emph{syntactic} class $\PDQC{k}$ contains all such $\cL'_{k,W,w}$ for that value of $k$.
\end{definition}

Relaxing the probability bounds in this manner results in far greater computational power.

\begin{proposition}
  The classes $\PDQC{k}$ are all equal to $\PP$, for all polynomially bounded $k \ge 1$.
\end{proposition}

\begin{proof}
$\PDQC{1} \subseteq \PDQC{k} \subseteq \PP$ follows directly from standard results (\cf{}~\cite{lit:ADH97}), so it suffices to show that $\PP \subseteq \PDQC{1}$.
To see this, we simply apply the trace estimation algorithm of Knill and Laflamme \cite{lit:KL9812} to the unitary that defines an arbitrary efficiently computable Boolean function.

Let $f : \FF_2^n \rightarrow \FF_2$ be a function computable in classical polynomial time, let $V_{[2..w]} = \sum_x (-1)^{f(x)}\ketbra{x}{x}$, and let $W_{[1..w]} = H_1 \cdot \Lambda_1(V_{[2..w]}) \cdot H_1$.
Then apply $W$ to the state $\rhoS(1,w)$ where $w$ is the width of the circuit that implements $W$. 
When the first qubit is measured in the computational basis, it will be $\ket1$ with probability $2^{-n} \cdot \#\{~ x ~:~ f(x) = 1 ~\}$, as required for $\PP$.
\end{proof}

\cleardoublepage

\chapter{The Fourier Hierarchy}
\label{chap:FH}

Classical computation permutes a discrete set of states (\cf{}~\S\ref{sect:classical-comp}), whereas quantum computation (despite the name) allows for a more continuous notion of state evolution.  Therefore perhaps one can make quantum computation `seem' like more of a natural extension of its classical counterpart by limiting to gates of a discrete group.  
This Chapter is concerned with the study of groups of transformations that fix the computational basis, \eg{} the group generated by gates from the set $\{X, \Lambda(X), \Lambda^2(X)\}$.
There are several different ways in which one can think of combining reversible circuits built from basis-preserving gates of this kind.  
For example, one might take the output of one such circuit, rotate each qubit in some prescribed fashion, and input this to the next circuit for further processing.  (We call this \emph{quantum adaption}, because the data being passed from one circuit to the next---determining the next phase of computation---is entirely quantum.)  This idea leads to the Fourier hierarchy of quantum complexity classes, introduced by Shi in \cite{lit:Shi0312}.  It provides us with a measure of quantum computation complexity that has to do with the branching and recombination of computational paths from the perspective of a canonical computational basis, and therefore allows (loosely speaking) for a kind of comparison with classical complexity that appeals to a classical-centric way of thinking.
By interleaving `classical' circuits with quantum basis-changes, resource requirements for a quantum computer (running with, say, polynomial spatial and temporal resources) can be quantified with more granularity~: by asking about both the complexity of the `classical' (non-branching) parts and also by counting the number of basis-changes employed.  

A more limited way of interfacing such circuits together would be to measure the output of one circuit in some pre-specified basis, and then use the resulting \emph{classical} data as classical control on the gates of the next circuit, whose quantum input should be `trivial' in some appropriate sense.  (We call this \emph{classical adaption}, because the data being passed from one circuit to the next is entirely classical.)  This idea leads to the definition of a \emph{Fourier Sampling Oracle}, as discussed in \cite{lit:BV97}.
Such a computing paradigm acquires its power from the fact that the quantum states input to a circuit---as well as the basis in which output measurements are taken---can be different from the computational basis.

Kitaev showed \cite{lit:Kit9511} that the `core part' of Shor's algorithm \cite{lit:Shor95} need not be expressed in terms of some Fourier transform directly related to the group being studied; rather, he developed the technique of \emph{eigenvalue estimation} to solve the \emph{Abelian Stabiliser Problem}, which generalises many of the problems that can be solved using Fourier techniques.  This means that the family of problems that seem to depend on Fourier techniques for their efficient solution (such as integer factorisation, computation of discrete logarithms, the abelian hidden subgroup problem, solving Pell's equation, and so on \cite{lit:Joz98,lit:Halesthesis,lit:Hal0205}), can be rendered efficiently without recourse to `complicated' \emph{Quantum Fourier Transforms}.  
We integrate Kitaev's algorithm with the approach taken here, and modify the control of the algorithm slightly in order to simplify the classical post-processing.
While this, on its own, does not seem to lead to a practical speed-up for solving problems, it does go some way to `demystifying' such algorithms, hopefully making them more accessible to further investigation and development.  In other words, by requiring all of the `work' of computation to be performed within `classical' circuits---encoding essentially no complexity within unitaries that are \emph{not} simply permutations of the computational basis---it is hoped that it could be easier to understand which parts of an algorithm might be easier to optimise, parallelise, or otherwise simplify, especially when adapting an algorithm to target a marginally different problem.
The `naturalness' of restricting to classical gates and Hadamard gates for analysing aspects of complexity has been noted by many authors (see especially \cite{lit:BvDR08} for recent work on algebraic circuits).  In particular, in~\cite{lit:DHHMNO} it is shown that simpler proofs exist for $\BQP \subseteq \PP$ when this approach is taken.
The ideas of this Chapter motivate a similar analysis in Chapter~\ref{chap:IQP} of a different discrete group.

We begin with some basic definitions and observations, discussing the role of \emph{adaption} in defining the Fourier hierarchy classes $\FH_k$, $\FH'_k$, and $\BPP^{\cFS[k]}$, considering the various ways in which quantum circuits implementing classical logic might be interfaced.  We show how these classes are related, and where they are likely to differ.  
Then we go on to consider Kitaev's algorithm for eigenvalue estimation, which belongs naturally in $\FH_2$, and consider the \emph{control schedule} for that algorithm in some detail.  We use this to show the new result that at least one cryptanalytically significant problem also belongs in $\BPP^{\cFS[1]}$ (Theorem~\ref{thm:Dlog}), which is tantamount to saying that it can be rendered without the use of any ancilla workspace.  

In \S\ref{sect:other}, we discuss extensions to these ideas, showing that other related problems might not be solvable without ancill\ae{}.  We briefly consider the trade-off between use of ancill\ae{} and circuit depth, and end by showing that continuous-group problems such as the solution of Pell's equation can also be rendered using Kitaev's scheme in $\FH_2$.
It is hoped that this understanding and analysis of the Fourier hierarchy will help with the future classification and development of quantum algorithms and subroutines.

\section{Definitions}

Throughout, global phases are ignored.  This means that wherever it is well-defined to do so, we shall conflate a matrix group with its projective equivalent (quotienting by $\CC^*$).

\subsection{Basic definitions}  \label{sect:QFTdefs}

\subsubsection*{Definitions of `classical' gates}

The perspective taken in this chapter is to regard quantum circuitry as a natural extension of classical (reversible) circuitry.  For this reason, it is convenient to fix a computational basis as usual, and then label certain quantum gates as `classical' because they fix that particular basis.  This expression ``classical'' is not to be understood as saying anything about an incapacity for such gates to create or modify superposition or entanglement, rather it is a basis-dependent property that describes how such gates \emph{collectively} stabilise the computational basis. 

Our first definition covers all permutations of the computational basis of an $n$-qubit machine, generated by `generalised Toffoli' gates.
\begin{definition}  \label{def:permgroup}
The \emph{Permutation Group} associated to a system of $n$ qubits is generated by the set of generalised Toffoli gates~:
\begin{eqnarray}  \label{eqn:def:permgroup}
  \mbox{Permutation Group}  &:=&  
  \span{ \Lambda^j(X) ~:~ j \in [0..n-1] } 
  \nonumber \\
  &\cong&  Sym(~ 2^n ~).
\end{eqnarray}
\end{definition}
This group is represented by the permutation matrices, constructed over $\CC$ in general.
The cardinality of the group is $2^n!$.
(If we were instead to limit to $\Lambda^2(X)=$ Toffoli gates, $\Lambda^1(X)=$ C-Not gates, and $\Lambda^0(X)=X$ gates, then only the \emph{alternating subgroup} would be generated, having cardinality $2^n!/2$~: so appending a separate ancilla qubit $\ket{0}$ would be a way to restore the entire permutation group without resorting to `large' gates.)

A more general definition, which still avoids the introduction of complex phases for the superposition phenomenon, is represented by the group of all signed permutation matrices, and is the semidirect product of real orthogonal diagonal matrices with permutation matrices.
\begin{definition}  \label{def:classicalGroup}
The \emph{Classical Group} associated to a system of $n$ qubits is generated by the Permutation Group together with generalised controlled-$Z$ gates~:
\begin{eqnarray}  \label{eqn:def:classicalGroup}
  \mbox{Classical Group}  &:=&  
  \span{ \Lambda^j(X), ~\Lambda^j(Z) ~:~ j \in [0..n-1] }
  \nonumber \\
  &\cong&  (~ \ZZ/2\ZZ ~)^{2^n} \rtimes Sym(~ 2^n ~).
\end{eqnarray}
\end{definition}
Again, $n$ counts all qubits in a circuit, the full circuit width.
The size of the group is $2^{2^n} \cdot 2^n!$, if we count global phase.
As with the permutation group, an alternative construction for simulating this group makes use of a small ancilla while limiting individual gates to three qubits.
(Because it is abelian, we write the group $\ZZ/2\ZZ$ additively, rather than multiplicatively as $Cyc(2)$ or $Sym(2)$.)

Any element $U$ of the classical group can be factored \emph{uniquely} into a permutation $f \in Sym( 2^n )$ \emph{followed by} a `diagonal' operator $\sigma \in ( \ZZ/2\ZZ )^{2^n}$, because of the structure as a semidirect product, and so we can sensibly write $U=(\sigma, f)$ to abbreviate line~(\ref{eqn:groupfactor}) below.  
\begin{eqnarray}  \label{eqn:groupfactor}
  U  &:&  \ket{x} ~\mapsto~ (-1)^{\sigma(f(x))} \ket{f(x)}
\end{eqnarray}
Note that the map $U = (\sigma, f) \mapsto f$ is a group homomorphism, and so if a circuit is given for $U$, then the subset of gates of the circuit that implement the $f$ part form a well-defined subset~: indeed they are just those gates from the permutation group.
But the map $U=(\sigma, f) \mapsto \sigma$ is \emph{not} a group homomorphism (the classical group is \emph{not} a direct product), and so the `complexity' apparent in the $\sigma$ part can be owing to the gates that implement $f$ as much as to any other part of the circuit.

The broad motivation for these definitions comes not from physical considerations pertinent to the task of fabricating a quantum information processor, but from the desire to analyse a fairly natural-looking measure of circuit complexity that is not apparent within the standard model, \viz{} the number of global Hadamard transformations ($Q_B$, defined below) needed, when quantum circuitry is seen as \emph{directly} extending classical circuitry.

\subsubsection*{Definitions of basis-change}

We consider the Binary Quantum Fourier Transform, denoted $Q_B$, also called the (global) Hadamard transform.  Because we sometimes wish to think of it as a \emph{passive} transform, acting not as a gate but rather by conjugating subsequent gates or measurements, we consider that it is to act on \emph{every} qubit in a computing system.
\begin{definition}  \label{def:QB}
The Binary QFT is given by
\begin{eqnarray}
  Q_B  &:=&  H^{\otimes n},
\end{eqnarray}
where $n$ counts \emph{all} the qubits in a circuit.
\end{definition}
As a gate, it acts on a unitary space of dimension $2^n$, and is defined by its action on the computational basis as follows, interpreting labels $\s$ and $\t$ as vectors in $\FF_2^n$~:
\begin{equation}  \label{eqn:QB}
  Q_B ~:~ \ket{\t} ~\mapsto~ 2^{-n/2} \sum_\s (-1)^{\s \cdot \t} \ket{\s};
\end{equation}
it has order 2, and hence is an involution.

As a passive action conjugating gates or measurements, it preserves locality of the operator algebra, effectively just exchanging Pauli $X$ operators with Pauli $Z$ operators.  (Other Fourier transforms, such as the \emph{Integer Fourier Transform} associated to the ring $\ZZ/2^n\ZZ$, do not share this property of preserving locality, and ought presumably be regarded as essentially more complex for that reason.)

\subsubsection*{Simulating single-qubit Hadamards}

\begin{proposition}  \label{propos:simH}
An Hadamard gate can be emulated from a gate-set containing all small gates from the permutation group  (Def.~\ref{def:permgroup}), together with the conjugates of those gates by $Q_B$ (computational basis input is assumed); and hence such a gate-set is universal for $\BQP$.
\end{proposition}

\begin{proof}
We can employ a simple technique from the idea of \emph{spin chains} (\cf{}~\S\ref{sect:spinChains}) to render a local $H$ operation on each of two qubits $a$ and $b$, while simultaneously swapping over their data, simply by using `classical' gates and two applications of~$Q_B$~:
\begin{eqnarray}  \label{eqn:hadamardsim}
  \Lambda_a(Z_b) \cdot Q_B \cdot \Lambda_a(Z_b) \cdot Q_B \cdot \Lambda_a(Z_b)
  &=& H_a \cdot H_b \cdot \mbox{Swap}_{ab}.
\end{eqnarray}

With the incorporation of two ancill\ae, $\ket{1}\ket{-}$, it is easy to render the same operation using only \emph{permutation gates} and two applications of $Q_B$.  The following construction emulates the previous one, preserving the ancill\ae{}~:
\begin{eqnarray}  \label{eqn:hadamardsim2}
  \Lambda^2_{ab}(X_\nu) \cdot Q_B \cdot \Lambda^2_{ab}(X_\mu) \cdot Q_B \cdot \Lambda^2_{ab}(X_\nu)  ~\ket{1}_\mu\ket{-}_\nu
  &\Rightarrow&   H_a \cdot H_b \cdot \mbox{Swap}_{ab}.
\end{eqnarray}
Of course, the Swap gate itself is also an element of the permutation group.

We can even drop the requirement for there to be provided an Hadamard-basis ancilla, because one can be constructed directly.  The gate $Q_B \cdot \Lambda^2_{ab}(X_c) \cdot Q_B$ may also be written $1-2\ket{--1}_{abc}\bra{--1}_{abc}$, and so applying it to $\ket{001}_{abc}$ one obtains a non-trivial superposition state $(\ket{00}-\ket{--})\ket1$.  Form two copies of such a state, and together these must be related by some permutation of the computational basis to a state that contains separately a $\ket+$ state amongst its qubits~:
\begin{eqnarray}
  \frac14 \left( \begin{array}{r} 1\\1\\1\\-1 \end{array} \right)^{\otimes 2} 
    &\stackrel{perm}{\mapsto}&
  \frac1{\sqrt8} \left( \begin{array}{r} 1\\1\\1\\1\\1\\-1\\-1\\-1 \end{array} \right)
                 \otimes \ket+.
\end{eqnarray}
Apply such a permutation, ignore the remaining qubits besides the $\ket+$, and apply $Z = Q_B \cdot X \cdot Q_B$ to it in order to obtain the state $\ket-$, for subsequent use as an ancilla. 

Because it is well-known that Hadamard plus Toffoli suffice for universality, so it follows that permutation gates together with their conjugates by $Q_B$ are sufficient for implementing a universal gate set for $\BQP$.  
\end{proof}

\subsection{Definition of Fourier hierarchy}  \label{sect:defFH}

Following \cite{lit:Shi0312}, the Fourier hierarchy is defined in terms of the number of time-slices within which Hadamard gates are used within a computation that is otherwise `classical'.  

\begin{definition}
A language $\cL$ belongs to $\FH_k$ if it is decided with bounded probability by a uniform family of circuits that have Hadamard gates within at most $k$ time-slices, and computational basis-preserving gates otherwise, and computational-basis ancill\ae{}.
\end{definition}

This definition should be understood as meaning that the way in which one decides whether some string $x$ of length $i$ is in $\cL$ is by applying a circuit $C_i$ from a uniform family to the computational-basis state $\ket{x}\ket{0}$, where the size of the ancilla register $\ket0$ would depend only on $i$ and be bound by some polynomial.  Moreover, the \emph{decision} would rest on the value of a single qubit (allowing for bounded probability), measured in the computational basis.  It is generally understood that the allowed `classical' gates are those from the permutation group (Def.~\ref{def:permgroup}) that affect at most a constant number of qubits, \eg{} three.

We also consider a slightly different version of the Fourier hierarchy, obtained by disallowing individual Hadamard gates, instead allowing only the $Q_B$ operation that spans the entire computer.  We let $\FH'$ denote the so-called \emph{strict} Fourier hierarchy.

\begin{definition}
A language $\cL$ belongs to $\FH'_k$ if it is decided with bounded probability by a uniform family of circuits that use only gates up to three qubits wide from the classical group (Def.~\ref{def:classicalGroup}), together with at most $k$ uses of the $Q_B$ operation, allowing also for computational-basis ancill\ae{}.
\end{definition}

By Proposition~\ref{propos:simH}, this latter hierarchy is not so very different, since individual $H$ operations can be simulated by appropriate use of $Q_B$ and ancill\ae{}.
\begin{corollary}
  ~~$\FH_k \subseteq \FH'_{2k}$.
\end{corollary}
\begin{proof}
  Given a circuit that uses Hadamard gates in $k$ time-slices, we can readily construct an $\FH'_{2k}$ equivalent circuit by using the substitution indicated at line~(\ref{eqn:hadamardsim}).
\end{proof}

Note that $\FH'_1 = \FH'_0 = \FH_0 = \P$  (\cf~\cite{lit:Shi0312}).  The reason for the first equality is that there is no utility in having \emph{exactly one} use of $Q_B$, because given the constraints on input and output, nothing can be computed in that case.  Formally, the magnitude of $\bra{\phi} C_1 \cdot Q_B \cdot C_2 \ket{\psi}$ is completely independent of $C_1$ and $C_2$ if they are both unitaries in the classical group and $\ket\phi$ and $\ket\psi$ are both in the computational basis.

\subsection{Definition of Fourier sampling oracle}

Another way of understanding this kind of extension of classical computing to quantum computing uses the idea of \emph{oracular access} to $\FH'_k$.  It turns out that in this way of thinking, $\FH'$ is not quite the simplest non-classical computing model that we can consider.  Accordingly, we define the \emph{Fourier Sampling Oracle}.

Let $U=(\sigma, f)$ be a unitary map on $w$ qubits, given by a circuit using gates from the classical group, as at line~(\ref{eqn:groupfactor}),
where $f : \FF_2^w \rightarrow \FF_2^w$ is a permutation, and $\sigma : \FF_2^w \rightarrow \FF_2$ is a boolean function.
Let $P_U$ be the probability distribution on domain $\FF_2^w$ that ascribes weights as follows~:
\begin{eqnarray}  \label{eqn:FSProb}
  P_U( \y )  &:=&  \bra\y~ Q_B \cdot U \cdot Q_B ~\ket\0^2   \nonumber \\
             &=&   2^{-2w} \sum_{\d \in \FF_2^w} (-1)^{\y \cdot \d} 
                   \sum_{\a \in \FF_2^w} (-1)^{ \sigma(\a) + \sigma(\a+\d) }.
\end{eqnarray}
\begin{definition}  \label{def:FSO}
In the notation given above, the \emph{Fourier Sampling Oracle}, denoted $\cFS$, is a device which, on input a classical description of a circuit for some such $U$, returns a single sample (from $\FF_2^w$) from the corresponding distribution $P_U$. 
\end{definition}
Note that the distribution $P_U$ depends on the function $\sigma$ but is independent of the function $f$.  One might as well therefore consider taking $f$ to be the identity, by composing $U$ with $(0,f^{-1})$ to get $(\sigma,1)$.  A circuit for $(0,f^{-1})$ can readily be identified from a circuit for $U$ (just pick out the permutation gates and reverse their direction), but the composition to make a circuit for $(\sigma,1)$ will still contain permutation gates, which cannot generally be rearranged to `cancel' one another without increasing the number of `diagonal' gates exponentially.

Parallel calls to $\cFS$ can be emulated by a single call to $\cFS$, because of the identity
\begin{eqnarray}  \label{eqn:locality}
  Q_B \cdot \left( U \otimes V \right) \cdot Q_B  &=&
    \left( Q_B \cdot U \cdot Q_B \right) \otimes \left( Q_B \cdot V \cdot Q_B \right).
\end{eqnarray}

$\BPP^\cFS$ is an interesting object of study if only because it can be defined quite independently of quantum mechanics, and therefore without reference to the actual creation or simulation of quantum states, processes, or circuits.
Note that the success probability of an algorithm in $\BPP^\cFS$ can be boosted toward unity \emph{without} increasing the number of calls to the $\cFS$ oracle, simply by parallel instantiation during the $\cFS$ part, followed by majority voting afterwards.   (This would not be true if $Q_B$ were replaced by some more general Fourier transform not satisfying the locality condition of line~(\ref{eqn:locality}), which is another reason for preferring to use the simpler $Q_B$ in these definitions.) 

It is clear from the definitions that $\FH'_2$ works out to be the class of problems that can be solved by $\BPP$ with a single call to $\cFS$, with deterministic pre-processing and subsequent post-processing. 
More generally, we write $\BPP^{\cFS[k]}$ to denote classical computation with up to $k$ adaptive calls to $\cFS$, with \emph{randomised} pre- and post-processing permitted, and this too can be regarded as forming a hierarchy of complexity classes. 
However, it is by no means apparent that even a single call to an oracle for $\FH'_3$, say, might be simulable by polynomially many calls to $\cFS$.  Therefore, the `hierarchy' of what can be computed in $\BPP$ with increasingly many calls to $\cFS$ is quite plausibly strictly contained within $\BQP$.  That is to say, it would be surprising if $\BQP \subseteq \BPP^{\cFS[poly]}$.  Yet to prove this separation rigourously would of course involve separating $\BQP$ from $\BPP$.

\medskip
\begin{proposition}
  ~~$\FH'_2 \subseteq \BPP^{\cFS[1]} \subseteq \FH_2$
\end{proposition} 
\begin{proof}
Immediate from the definitions.
\end{proof}

In \S\ref{sect:ShorRevisit} we consider problems in $\FH_2$.  Our main result of the section is to see why some of these problems are also in $\BPP^{\cFS[1]}$.

\subsection{Adaption}  \label{sect:adaption}

Informally speaking, one might say about algorithmics within the strict hierarchy $\FH'$ that it is a way of `gluing together classical subroutines quantumly', by interleaving `classical' circuits with the $Q_B$ operator.   We call this \emph{quantum adaption}, because quantum data is being passed from one `classical' circuit to the next to drive the computation.

Similarly, one might say of $\BPP^{\cFS[k]}$ that it uses \emph{classical adaption}, because there is no transfer of quantum data between oracle calls.  (The oracle calls themselves do involve `classical' circuits in some sense, and also do involve quantum computing, but there is no actual flow of quantum data between distinct classical circuits.)

A third kind of adaption that we mention, for completeness, is one whereby both quantum and classical data are explicitly passed from one `classical' circuit to the next, in a serial manner.  This we call \emph{mixed adaption}, and in this case the $Q_B$ operation is no longer needed.  Proposition~\ref{propos:Hfromfeedforward} below illustrates another way of conceptualising $\FH$, in terms of mixed adaption.

\begin{proposition}  \label{propos:Hfromfeedforward}
Limiting circuits to use only classically-controlled gates (\S\ref{sect:paritycontrol}) from the permutation group, inputs from the Hadamard basis, intermediate single-qubit measurements in the Hadamard basis, and feed-forward of classical measurement data to subsequent classical-control, one can efficiently emulate Hadamard transforms (and hence ultimately $\BQP$-universality).
\end{proposition}

\begin{figure}[h]
  \begin{center}
     \includegraphics[width=90mm]{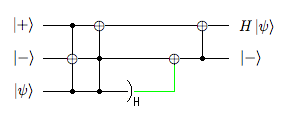}
     \vspace{-5mm}
     \caption{\label{fig:hadamard1} Using permutations and mixed adaption to simulate an Hadamard gate.  The green wire denotes a classical control signal to an $X$ gate, following a measurement in the Hadamard basis.  A formula for this figure is given in line~(\ref{eqn:hadamardsim3}).}
  \end{center}
\end{figure}

\begin{proof}
The formula is straightforward, and can be seen directly in Fig.~\ref{fig:hadamard1}.
Write $X_b^{\ketbra--_c}$ to denote applying gate $X$ to qubit $b$ conditional on qubit $c$ having been measured to be in the image of the projector $\ketbra--$.  Then the formula for emulating an Hadamard gate is given symbolically as
\begin{eqnarray}  \label{eqn:hadamardsim3}
  \Lambda_b(X_a) \cdot X_b^{\ketbra--_c} \cdot \Lambda^2_{bc}(X_a) \cdot 
    \Lambda^2_{ac}(X_b) ~\ket+_a \ket-_b \ket\psi_c
    &\Rightarrow&
    H_a \ket\psi_a \ket-_b.
\end{eqnarray}
To see why this works, take $\ket\psi = \alpha\ket0 + \beta\ket1$.  Then the starting state on the three qubits may be written as a vector of amplitudes as $(\alpha,\beta,-\alpha,-\beta,\alpha,\beta,-\alpha,-\beta)$.  After the first Toffoli gate, this becomes $(\alpha,\beta,-\alpha,-\beta,\alpha,-\beta,-\alpha,\beta)$.  After the second Toffoli gate it becomes $(\alpha,\beta,-\alpha,\beta,\alpha,-\beta,-\alpha,-\beta)$.  Measure the third qubit in the Hadamard basis and it becomes either $(\alpha+\beta,-\alpha+\beta,\alpha-\beta,-\alpha-\beta)\ket+$ or $(\alpha-\beta,-\alpha-\beta,\alpha+\beta,-\alpha+\beta)\ket-$.  Apply $X_b$ controlled on the measurement result, and this gives $(\alpha+\beta,-\alpha+\beta,\alpha-\beta,-\alpha-\beta)$ in either case, up to global phase.  Apply the final C-Not gate to obtain $(\alpha+\beta,-\alpha-\beta,\alpha-\beta,-\alpha+\beta)$, which is equivalent with $H_a\ket\psi_a \ket-_b$.  Since it works for all pure $\ket\psi$, by linearity it must work for all quantum data.
\end{proof}

If one were ever to discover a paradigm for quantum computing within which gates from the permutation group were fast to implement, but where other gates were not feasible, and where measurements and feed-forward of classical data were slow, then perhaps the Fourier hierarchy $\FH$ would be an ideal way of measuring algorithmic complexity within such a paradigm.

Note also that one could easily adapt the proof of Proposition~\ref{propos:Hfromfeedforward} to use \emph{post-selection} in place of measurement and feed-forward, to prove that classical computing with a single call to $\cFS$ and \emph{post-selective} post-processing is universal for $\mathbf{PostBQP}$, which is $\PP$ (\cf{} \S\ref{sect:postselection}).  
\begin{proposition}  \label{propos:post-FS}
  $\BPP^{\cFS[1]}$ with post-selection gives rise to $\mathbf{PostBQP}$.
\end{proposition}
\begin{proof}
Without loss of generality, consider starting with a $\BQP$ circuit composed of Hadamard gates and Toffoli gates and $X$ gates, beginning with $\ket{0}$ input and then a $Q_B$ operation, and ending with another $Q_B$ operation before measurement (and post-selection) in the computational basis.  

Every Hadamard gate in this circuit (besides the ones comprising the initial and final $Q_B$ operations) should then be replaced by the gadget of line~(\ref{eqn:hadamardsim3}); but the measurement involved within that gadget should be delayed until the end of the computation, and correspondingly the classically controlled gate in the gadget should be omitted.  This leaves us with a circuit having Hadamard-basis input, Hadamard-basis measurements at the end, and otherwise all gates from the permutation group, and such a circuit is of the correct form for an $\cFS$ oracle, as per line~(\ref{eqn:FSProb}).  At the end of the computation, when the $\cFS$ oracle returns a string, the bits that \emph{ought} to have been used for feed-forward classical control (which are otherwise no longer used) should be post-selected to have been qubits in state $\ket+$ (so that it was correct to have dropped the controlled gates), which always happens with non-zero amplitude by the proof of Proposition~\ref{propos:Hfromfeedforward}.
\end{proof}

\section{Kitaev's Algorithm Revisited} \label{sect:ShorRevisit}

Kitaev's algorithm for the \emph{Abelian Stabilizer Problem} \cite{lit:Kit9511} belongs naturally within the category of $\FH_2$ computing, and certain applications (most notably in cryptography) lead to the solution of problems in $\BPP^{\cFS[1]}$.  

\begin{theorem}[Kitaev]
  The decision variants of Integer Factorisation and the Discrete Logarithm problem are in $\FH_2$.
\end{theorem}

In this section, we recall the core of Kitaev's algorithm---\emph{Eigenvalue Estimation}---describing it in terms of the $\cFS$ oracle, and offer a slightly different `control schedule' for simplifying the follow-on post-processing.  Using this, our main result of the section is to establish that (the decision variants of) the cryptographic problem ``Discrete Log over Finite Fields of Characteristic 2'' is in $\BPP^{\cFS[1]}$. 
We also discuss other cryptographic problems in $\FH_2$ whose arithmetic is sufficiently complex that it is not apparent whether or not they are also in $\BPP^{\cFS}$.
Unlike $\FH_2$, the oracle $\cFS$ admits no space for ancill\ae{}, so all the `computational work' it performs is done `in place'~: a very limited form of computing (\cf{}~\S\ref{sect:DQC1}).
We also consider the depth of circuits required within the $\FH_2$ and $\BPP^{\cFS[1]}$ frameworks.

\subsection{Eigenvalue estimation}  \label{sect:eigest}

Let $f$ be a permutation on the set of strings of length $n$, and suppose that for any `control integer' $c$ we can construct a circuit of width $n$ and size $O( poly(n) \cdot \log c )$ for implementing $f^c$, using gates from the permutation group.  The goal of Eigenvalue Estimation is to find the length of one of the larger cycles of $f$.  This is called Eigenvalue Estimation because the cycle-structure of $f$ is naturally encoded within the spectrum of the unitary map corresponding to the circuit that implements $f$~: that is, to each cycle of length $q$ there corresponds a subspace of dimension $q$ spanned by the computational basis elements associated to the elements of that cycle, and in a different basis this space can be expressed as the product of one-dimensional eigenspaces having eigenvalues $\exp( 2\pi i \kappa / q )$ for each $\kappa \in [0..q-1]$. 
In other words, for $x$ some point on a $q$-cycle of $f$, the following two bases span the same space~:
\begin{eqnarray}  \label{eqn:eigenvec}
  && \left\{~ \ket{x}, \ket{f(x)}, \ket{f^2(x)}, \ldots, \ket{f^{q-2}(x)}, \ket{f^{q-1}(x)} ~\right\}, \nonumber \\
  && \left\{~ \ket{\lambda_x(\kappa)} ~:=~ \frac{1}{\sqrt{q}}\sum_{j=0}^{q-1} \exp{\left( \frac{-2\pi i j \kappa}{q} \right)}\ket{f^j(x)} ~\right\}_{\kappa=0}^{q-1}.
\end{eqnarray}
Eigenvalue Estimation is about finding both a $q$ and a $\kappa$ for some suitable permutation $f$ in context of some eigenvector $\ket{\lambda_x(\kappa)}$, or possibly for two different commuting permutations in context of the same mutual eigenvector.

\subsubsection*{Using the $\cFS$ oracle}

To find one such $q$ using a single call to $\cFS$, we need to construct a unitary $U$ to submit to the oracle.  Kitaev's idea is to implement $\Lambda(f^c)$, onto some essentially arbitrary target, many times for many different values of $c$, using different control qubits but the same target qubits.  (One could say that this technique has the effect of \emph{measuring} the target state in the eigenvalue basis.)  So we implement this idea by fixing some schedule of `control integers'---a list $\{ c_1, c_2, \ldots, c_m \}$---and consider a circuit on $w=m+n$ qubits of the form
\begin{eqnarray}  \label{eqn:U0}
  U(\0)_{[1..w]}  &:=&  \Lambda_1( f^{c_1}_{[m+1.m+n]} ) \cdot \Lambda_2( f^{c_2}_{[m+1.m+n]} ) \cdots \Lambda_m( f^{c_m}_{[m+1.m+n]} ). 
\end{eqnarray} 

Now this $U(\0)$ applied to a state of the form $\ket+^{\otimes m}_{[1..m]} \ket{\lambda_x(\kappa)}_{[m+1..m+n]}$ would leave the $\ket{\lambda_x(\kappa)}$ register unchanged, and would transform the $j$th qubit (for $j \in [1..m]$) independently by rotating it around the equator of the Bloch sphere through an angle of $2\pi c_j \kappa/q$.  (This is sometimes called `phase kickback'.)  This transfers some information about the eigenstate $\ket{\lambda_x(\kappa)}$ into the $j$th qubit in a way that can be measured.  Moreover, by using different values of $c_j$ for different qubits, we obtain different data about the eigenstate.  We call the first $m$ qubits \emph{control qubits} or \emph{the control register} and we call the last $n$ qubits the \emph{target register}.

So $U(\0)$ is almost the unitary we want for submitting to the $\cFS$ oracle, except that the oracle would apply $U(\0)$ to the state $\ket+^{\otimes w}$, which is not so useful~: the target register state $\ket+^{\otimes n}$ is a linear combination of eigenvectors of the form $\ket{\lambda_x(0)}$, so that for all of these the rotation angle is 0 (independent of $c_j$ and $q$) and hence the phase kicked back is 1.  We therefore adjust $U(\0)$ by composing it with a random pattern of $Z$ gates to make the unitary $U(\r)$~:
\begin{eqnarray}  \label{eqn:EE}
  U(\r)_{[1..w]}  &:=&  U(\0)_{[1..w]} \cdot \prod_{t=1}^{n} Z_{m+t}^{r_t},
\end{eqnarray}
where $\r$ is a random $n$-long string of bits.  
Now if $U(\r)$ is submitted to the oracle, the effect would be that of applying $U(\0)$ to a state whose control register is correctly set ($\ket+^{\otimes m}$) and whose target register is effectively \emph{fully depolarised}.  Thus the final measurement results will be no different than had we uniformly randomly selected a point $x$ and uniformly randomly selected a number $\kappa$ and applied $U(\0)$ with the target register in the eigenstate $\ket{\lambda_x(\kappa)}$. 

(This illustrates quite nicely the point made back in \S\ref{sect:QFTdefs} regarding the factorisations of a `classical group' unitary into a `permutation' part and a `diagonal' part.  For when the operator $U(\r)$ is factorised as at line~(\ref{eqn:EE}), the `diagonal' part comes first and is rather trivial, but if it were to be factored the other way round---as at line~(\ref{eqn:groupfactor}) with the `diagonal' part \emph{coming after} the `permutation' part---then the resulting `diagonal' part would be far more complicated.)

Then the $w$-bit string returned by the oracle will be such that amongst the first $m$ bits, the bias of the $j$th bit will be $\cos( 2\pi c_j \kappa/q )$, where $q$ is the length of a randomly chosen cycle of $f$ and $\kappa$ is a random integer in $[0..q-1]$.
We have therefore proved the following lemma~:
\medskip
\begin{lemma}[\emph{Cf.} \cite{lit:Kit9511}]  \label{lem:eigest1}
  For any uniform family $\{f\}$ of permutations with properties as above---provided the descriptions for constructing the circuits for $f^c$ are themselves uniform---there is a $\BPP^{\cFS[1]}$ subroutine which takes as input a description for such circuits and a description of a control schedule $(\{ c_j \}_{j=1}^m)$ and outputs a string the first $m$ bits of which have biases $\cos(2\pi c_j \kappa/q)$ for $j \in [1..m]$, where $q$ is the length of the orbit of a randomly chosen point in the domain of $f$, and $\kappa$ is a random integer, neither of which depends on $j$.
\end{lemma}
\begin{proof}
Overview as above, calling the $\cFS$ oracle with the $U(r)$ of line~(\ref{eqn:EE}), with additional explanatory details to be found in \cite{lit:Kit9511} and \cite{book:NandC}.
\end{proof}

Note that if the random string $\r$---used at line~(\ref{eqn:EE}) to select a random pattern of $Z$ gates---were set to be uniform over all strings other than the all-zero string, then this would diminish the probability of finding a case for which $\kappa=0$, but it would not eliminate that possibility altogether unless $f$ consisted of a single cycle of length~$2^n$.

\subsubsection*{Choosing a control schedule}

Next we consider how to choose the integer values $c_j$ that are used as indices on $f$ in the subroutine of Lemma~\ref{lem:eigest1}.  These values must be specified up front, not selected adaptively (for an algorithm in $\BPP^\cFS$).  In what follows, let $\phi = \kappa/q \in [0,1)$ denote the rational number that the algorithm is intended to find.  Since $\cos(2\pi c_j \phi) = \cos( 2\pi c_j (1-\phi) )$ for all integer $c_j$, we may as well take $\phi$ to be in $[0,\frac12]$, by symmetry.
The $c_j$ values then control the probabilities of the bits returned by the $\cFS$ oracle, with $p_0 = \cos^2( \pi c_j \phi )$ and $p_1 = \sin^2( \pi c_j \phi )$ being the probabilities of the $j$th returned bit being 0 or 1 respectively.  These bits are to be post-processed classically in order to learn the value $\phi$. 

Kitaev suggested taking the $c_j$ values of the form $2^\alpha$, repeating each several times in order to get an estimate of $2^\alpha \phi$ with a few bits of precision very accurately.  But as we have already seen, the bits being returned give information about $\sin^2( \pi c_j \phi )$, not about $c_j \phi$ directly.  
By way of example, suppose that $\phi$ were rather close to $\frac14$, say $\phi = \frac14 + \epsilon$; so close that measurement of $\sin^2( \pi 2^0 \phi ) = \sin^2( \pi (\frac14 + \epsilon) )$ to a few bits of precision would be unable to determine the sign of $\epsilon$ with any significant accuracy.  But then subsequent measurements of $\sin^2( \pi 2^\alpha \phi )$ for any integer $\alpha \ge 1$ will contain absolutely no information about the sign of $\epsilon$, because $\sin^2( \pi 2^\alpha (\frac14 + \epsilon) ) = \sin^2( \pi 2^\alpha (\frac14 - \epsilon) )$.  Thus a bit of information about $\phi$ would be inaccessible in this case, and without further modification, the algorithm would not work in the worst case with this control schedule.

Our proposed solution for a choice of control schedule---designed so as to interface with the subroutine of Lemma~\ref{lem:eigest1} without further modification---is to use integers of the form $2^\alpha3^\beta$.
\begin{lemma}  \label{lem:eigest2}
  For any $\eps>0$ there is a set of integers $\{c_j\}_{j=1}^m$ and an efficient classical algorithm that succeeds with probability $1-\eps$, such that for any rational $\phi \in [0,\frac12]$ with denominator $\le 2^n$, the algorithm outputs $\phi$ when input a sequence of $m$ independent bits with respective biases $\cos(2\pi c_j \phi)$ (as per Lemma~\ref{lem:eigest1}).  The algorithm takes time roughly linear in $m$, and $m = O( n^2 \log \frac{n}\eps )$, and the bit-length of each $c_j$ is $O(n)$.
\end{lemma}

\begin{proof}
First we give the control schedule $\{c_j\}_{j=1}^m$ precisely, and the algorithm, and then we argue for its correctness.

Let $\alpha$ take all integer values in the range $[0..t]$ and let $\beta$ take all integer values in the range $[0..d]$ where $t=2n$ and $d=O(n)$.  For each $\alpha,\beta$ pair, let there be $k$ different times when $c_j=2^\alpha3^\beta$, where $k=O( \log( \frac{n}{\eps} ) )$.  Then $m = k \cdot (t+1) \cdot (d+1)$ is the total size of the control register.

Denote the binary expansion of the unknown $\phi$ out to $t$ bits of precision as $\phi = 0.0\phi_2\phi_3\phi_4\ldots$.  Fix a parameter $\eta \approx 0.32$.
Process the random bits by letting $\mu_{2^\alpha3^\beta}$ be the average of those $k$ bits for which $c_j = 2^\alpha3^\beta$; that is to say
\begin{eqnarray}
  \mu_{c}  &\sim&  \frac1k \cdot Bin(~ k, \sin^2( \pi c \phi ) ~).
\end{eqnarray}

The pseudocode for the ensuing estimation procedure is given as follows, in Fig.~\ref{fig:alg1}.  

\begin{figure}[ht]
  \caption{\label{fig:alg1} Eigenvalue estimation~: classical post-processing.}
  \vbox{\texttt{
    \vskip 1mm
    Input~~: $\mu_{2^\alpha 3^\beta} ~\sim~ \frac1k B(~ k, \sin^2\pi 2^\alpha 3^\beta \phi ~)$, for $\alpha$ up to $t$, for $\beta$ up to $d$. 
    \\
    Output~: Estimate for hidden $\phi = 0.0\phi_2\phi_3\ldots$ up to $t$ bits of precision.
    \\
    Params~: $k$, $d$, and $\eta$ all affect worst-case failure probabilities.
    \\
    \begin{enumerate}
      \item  \quad  $\sigma \leftarrow 0$;
      \item  \quad  for $\alpha$ in $[0..t]$ do
      \item  \quad  \quad  $\phi_{\alpha+1} \leftarrow \sigma$,  ~$\sigma \leftarrow 1-\sigma$;
      \item  \quad  \quad  for $\beta$ in $[0..d]$ do
      \item  \quad  \quad \quad  $\sigma \leftarrow 1-\sigma$;
      \item  \quad  \quad \quad  if $\mu_{2^\alpha 3^\beta} < \eta$ then continue $\alpha$;
      \item  \quad  \quad \quad  if $\mu_{2^\alpha 3^\beta} > 1-\eta$ then $\sigma \leftarrow 1-\sigma$, ~continue $\alpha$;
      \item  \quad  \quad  continue $\beta$;
      \item  \quad  continue $\alpha$;
    \end{enumerate}
  }}
\end{figure}

Having obtained such an estimate for $\phi$ (assuming no errors), one can use the usual efficient technique of continued fractions \cite{lit:Kit9511} to recover the rational $\phi$ exactly.

The technique used in Fig.~\ref{fig:alg1} can be understood inductively.  Take the inductive hypothesis to be that, on entering the outer loop for the $\alpha+1$st time, the first $\alpha+1$ bits of $\phi$ are correctly known (\ie{} bits $\phi_1,\phi_2,\ldots,\phi_{\alpha+1}$).  In fact, line 3 of Fig.~\ref{fig:alg1} explicitly shows the bit $\phi_{\alpha+1}$ being stored~: at each point in the algorithm, the variable $\sigma$ holds the bit value deemed to be most likely for the next (unstored) output bit.
Then the inner loop examines successively more data, in a bid to determine the parity $\phi_{\alpha+1} \oplus \phi_{\alpha+2}$, from which the next output bit is learnt.

The role of the factor $3^\beta$ in $c_j$ is to deal with the case whereby the estimator for the next bit ($\mu_{2^\alpha} \approx \sin^2( \pi 2^\alpha \phi )$) happens to be inconclusive for further progress, \ie{} when it is too close to $\frac12$ to distinguish reliably between the two alternative hypotheses for the parity $\phi_{\alpha+1} \oplus \phi_{\alpha+2}$.

Estimators of the form $\mu_{2^\alpha 3^\beta}$ are ideal for that case.  This is because as soon as it is assured that $2^\alpha\phi \in (\frac18,\frac38)$, the estimator $\mu_{2^\alpha 3^1}$ is effective for `zooming in' to help determine on which side of $\frac14$ the value $2^\alpha\phi$ is more likely to lie (and likewise for the symmetrically opposite case, see Fig.~\ref{fig:circle} for a visual aid).  Larger values of $\beta$ then give improved `magnification'.

\begin{figure}[h]
\vspace{5mm}
  \begin{center}
     \includegraphics[width=80mm]{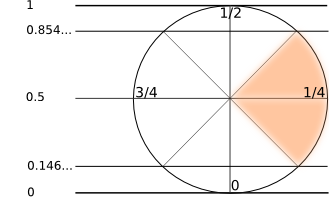}
     \caption{\label{fig:circle} Schematic for $\sin^2( \pi c \phi )$.  If a value $c \phi$ lies within the range $[\frac18, \frac38]$ (shaded area), then the best way to tell which side of $\frac14$ it lies on is to triple it (unshaded area), magnifying the resolution.}
  \end{center}
\end{figure}

The parameter $\eta$ is set to be well within the interior of the region $[0.146,0.5]$, so in terms of Fig.~\ref{fig:circle}, each round of the inner loop can be thought of as testing whether $c\phi$ lies well away from the `equatorial' neighbourhoods of $\frac14$ and $\frac34$, where the parity $\phi_{\alpha+1} \oplus \phi_{\alpha+2}$ would still be ambiguous.  While this ambiguity persists, the inner loop increases the value of $c$ by a factor of 3, driving $c\phi$ away from these neighbourhoods, until eventually the required parity bit is correctly learnt with high probability.
  
To understand the success probability of the algorithm and the role of the parameters $k, d, \eta$, it is appropriate to consider the probability of the algorithm making a mis-assignment of the $\alpha+2$nd bit of $\phi$, given that all prior assignments were correct.  (Note that if $2^{\alpha+2}\phi$ happens to be an integer precisely, then there are two equally valid possibilities for the $\alpha+2$st bit of $\phi$; and this causes no practical problem.)

There are four places where the algorithm shown can go wrong~: namely lines~6,~7,~8,~9 in the code of Fig.~\ref{fig:alg1}.  It makes an invalid assumption at line~6 if $\mu_c < \eta$ when in fact $\sin^2(\pi c \phi) > \frac12$.  The marginal probability of this happening on any particular $c$ is never worse than $\exp( -2k \cdot ( \frac12 - \eta)^2 )$, using the Chernoff bound.
Similarly, the algorithm makes an invalid assumption at line~7 if it sees $\mu_c > 1-\eta$ when really $\sin^2(\pi c \phi) < \frac12$.  Again, the marginal probability is as above, by symmetry.  
Thus the overall probability of either of these two kinds of error occurring at any point in the algorithm, regardless of $\phi$, is certainly bounded above by
\begin{eqnarray}  \label{eqn:errorbound}
  (t+1)\cdot(d+1)\cdot\exp\left( -2k \cdot ( \frac12 - \eta )^2 \right).
\end{eqnarray}
This is exponentially small in $k$, for fixed $\eta$.

The algorithm makes a potentially invalid assumption at line~8 if it is not justified in looking at the next $\beta$ value, because $\sin^2( \pi c \phi )$ lies outside the range $[\sin^2( \frac\pi8 )..\sin^2( \frac{3\pi}8 )]$ despite the fact that $\mu_c$ lies within $[\eta..1-\eta]$.  
As before, this probability is bounded above at any given point by a Chernoff bound of $\exp( -2k( \eta - \sin^2(\frac\pi8) )^2 )$.
So the overall probability of this kind of error occuring at any point in the algorithm is bounded above by 
\begin{eqnarray}  \label{eqn:errorbound2}
  (t+1)\cdot(d+1)\cdot\exp\left( -2k \cdot ( \eta - \sin^2( \frac\pi8 ) )^2 \right),
\end{eqnarray}
which is also exponentially small in $k$.

The fourth way in which the algorithm can make an error---line~9---is by `using up' all the available $\beta$ values for a given $\alpha$, without coming to a firm conclusion about the parity of the next bit.  This may be precluded for the rational $\phi$ that we consider by choosing $d$ so that $3^d/q \ge \frac14$.  This is surely achieved if $3^d > 2^n$, for example.
(Heuristically however, to avoid errors with exponentially good probability in the \emph{average} case, $d$ really only needs to be at least as long as the longest run of zeroes or ones in the first $t$ bits of $\phi$, which suggests asymptotically taking $d =\Omega( \log(t) )$.) 
\end{proof}

\subsection{Discrete logarithm over $\FF_{2^n}$}

Suppose we wish to find the discrete logarithm between two elements of a finite field of characteristic 2.  Let $g, h \in \FF_{2^n}^*$ be the two elements, so that we seek a solution to $g^s = h$.  (The group $\FF_{2^n}^*$ is isomorphic to one single cycle of known cardinality, \ie{} there are $2^n-1$ units in the finite field.)  In general, this problem is believed to be classically hard.

Kitaev's algorithm for the discrete logarithm involves learning $\kappa/q$ for two different permutations with respect to the same eigenvector.  The two permutations in this context are multiplication by $g$ and multiplication by $h$, respectively.  The map analogous to the $U(\0)$ of line~(\ref{eqn:U0}) is here to be defined so that the first half of the control bits (say bits $[1..\frac{m}2]$) control applications of powers of the \emph{first} permutation (\ie{} multiplication by $g^{c_j}$), while the second half of the control bits ($[\frac{m}2+1..m]$) control applications of powers of the \emph{second} permutation (\ie{} multiplication by $h^{c_j}$).  Then, adjusting Lemma~\ref{lem:eigest1} accordingly, we can arrange for an output string the first $\frac{m}2$ bits of which have biases of the form $\cos( 2\pi c_j \kappa/q )$ and the second $\frac{m}2$ bits of which have biases of the form $\cos( 2\pi c_j \kappa s/q )$, for the same $\kappa$ and $q$ (and for $c_j$ of our choosing).

\medskip
\begin{theorem}  \label{thm:Dlog}
The decision version of the discrete log problem over finite fields of characteristic 2 is in $\BPP^{\cFS[1]}$.
\end{theorem}
\begin{proof}
  To apply the ideas of this section to the discrete log problem requires that we have a concrete way of representing $\FF_{2^n}$ using $n$-bit strings, and also a way of implementing the appropriate permutations using $n$-bit wide permutation circuits, without ancill\ae{}.  The permutations in question are various powers of multiplication by $g$ or $h$ within the representation of $\FF_{2^n}$.  Now these various powers can all be precomputed (there are $m$ of them, and $m$ is bounded by a polynomial in $n$, \cf{} Lemma~\ref{lem:eigest2}), and each is simply multiplication by a constant.  Multiplication by a constant---in any standard representation of the field---is an $\FF_2$-linear transformation.  This means that over $\FF_2$ it can be represented as an $n$-by-$n$ matrix multiplication.  By performing Gaussian Elimination on such a matrix, one can factor it into a pair of triangular matrices, and thence construct a circuit of Not gates and C-Not gates, of quadratic complexity \cite{lit:Damm90} for implementing it \emph{in place}.  All this pre-processing can be rendered classically in polynomial time (and hence within the present framework), and thus a suitable input to the $\cFS$ oracle can be prepared, in the same way as was done for Lemma~\ref{lem:eigest1} (see line~(\ref{eqn:U0})).
  
  Finally, by learning $\kappa/q$ and $\kappa s/q$, for some random (non-zero) $\kappa$, it is easy (with classical post-processing) to recover $s$, which is the sought-after discrete log.
\end{proof}

\subsection{Extensions and future work}  \label{sect:other}

This section describes why it is not clear whether all related problems can be solved in $\BPP^{\cFS[1]}$, why it is not clear how well general problems in $\FH_2$ parallelise, what the mathematical relationship is between the algorithm of eigenvalue estimation in $\FH_2$ and the so-called \emph{Partition Problem}, and how the algorithm generalises to the continuous context.

\subsubsection*{Generalising to other Abelian Hidden Subgroup problems}

Suppose as before that $f$ is a permutation on the set of strings of length $n$, but now suppose that we need some clean ancilla space if for any integer $c$ we are to construct a circuit of size $O( poly(n) \cdot \log c )$ for implementing $f^c$, using gates from the permutation group.  In this (more realistic) case, Lemma~\ref{lem:eigest1} will not apply as it stands, because the only ancilla available in context of a $\cFS$ subroutine would be one prepared in the Hadamard basis.
Nonetheless, it is still possible to solve eigenvalue estimation problems within $\BPP^\cFS$ for such families of function, provided that a logarithmically big ancilla suffices, and provided that the language being decided is itself in $\NP$.  This is because---as we saw in Chapter~\ref{chap:PMC}---one can compute using depolarised qubits, provided one is prepared to amplify success proabilities via an \emph{outer loop} with a `verifier' (\ie{} run the algorithm many times and look at all answers before deciding whether to accept).
Just as at line~(\ref{eqn:EE}), we can prepare a depolarised ancilla by applying a random pattern of $Z$ gates to the Hadamard-basis ancilla.  There will then be a non-negligible probability of the overall algorithm behaving \emph{as though} a clean ancilla space had been provided.

So what happens when we try to solve \emph{Integer Factorisation} (equivalently, compute Euler's totient function), or \emph{Discrete Logarithm} over other finite fields?  Although such problems are clearly solved within $\FH_2$ and although they certainly belong to $\NP$ (and so have polynomial time verifiers for use in post-processing), it is still not apparent that they can be solved in $\BPP^{\cFS}$, because it is not clear that one can find efficient circuits for implementing the required permutations \emph{in place} even with log-sized ancilla space.  We note that it is clear that circuits of permutation gates must \emph{exist} for this kind of in-place permutation, but are they polynomial in size, and can they be compiled in polynomial time?  

For example, for Euler's totient function we would need to be able to design an efficient $n$-bit circuit of permutation gates for mapping integers $x \mapsto x \cdot m \pmod{N}$, where $N$ and $m$ are compile-time constants, $2^{n-1} < N < 2^n$, $0 < x < N$, $0<m<N$.  Now the arithmetic in question may require only a little ancilla space to compute the \emph{individual bits} of the output, but it is unclear where to write the bits of output as they are each computed so as not to corrupt the input before we are done using it!  (This problem did not exist in our previous example for discrete logs over $\FF_{2^n}^*$, because the underlying arithmetic there was seen to become remarkably simple after classical polynomial-time pre-processing.)

It remains as future work to determine which aspects of \emph{elementary arithmetic} can be performed in-place using permutation circuits with log-sized clean ancilla space.

\subsubsection*{Parallelisation for $\FH_2$ problems}

H\o{}yer and Spalek \cite{lit:Hoy02} have interesting results about reducing the \emph{depth} of certain quantum circuits to a constant, using so-called `fan-out' gates of arbitrary width, or equivalently `parity' gates of arbitrary width, both of which belong to the permutation group.  Arithmetic performed with such optimisations will certainly require substantial ancilla space, and therefore will not lead to algorithms of the form required for $\BPP^\cFS$.  Can it nonetheless lead to significant parallelisation of problems in the $\FH_2$ framework using these `wide gates'?

We did not find a way to make a good parallelisation, because the constructions of \cite{lit:Hoy02} additionally require gates of the form $\Lambda_a( e^{i \varphi Z_b} )$ for various angles $\varphi$, and to emulate such gates using permutation circuits requires not only having permutations with appropriate cycle structure, but also being able to construct the eigenvectors of these permutations.  It is therefore something of an open problem to find optimal circuits for arithmetic.  

In \cite{lit:DKRS04} a uniform method is given for computing \emph{integer addition} in place, in logarithmic depth, using permutation gates, with a linear-sized ancilla in the computational basis.  The so-called ``Carry-lookahead in-place adders'' given there can implement quantum addition of the form $\ket{a}\ket{b}\ket0 \mapsto \ket{a}\ket{a+b}\ket0$, or addition of a constant, of the form $\ket{b}\ket0 \mapsto \ket{a+b}\ket0$, in the ring of integers without modular reduction (the sum being represented in one more bit than the summands).  It is straightforward to adapt these circuits to render addition modulo a classical integer $N$ known at compile-time, without affecting depth or ancilla requirement by more than a constant factor.  Analogous results for \emph{integer multiplication} are not presently known.

\subsubsection*{Relation to Subset-Sum}

Lemma~\ref{lem:eigest1} expressed the probability of measuring some $m$-bit string in terms of a stochastically chosen angle $\phi$ of the form $\kappa/q$~: we found that 
\begin{eqnarray}
  \Pr( \y )  &=&  \frac{1}{q \cdot 2^m} ~\sum_{\kappa = 0}^{q-1} ~\prod_{j=1}^m
                 \Bigl( 1 + (-1)^{y_j} \cos( 2\pi c_j \kappa/q ) \Bigr),
\end{eqnarray}
where $y$ is the $m$-bit string returned from the control register, assuming that all cycles have length $q$ exactly.
But we can also compute this probability directly from the Born rule and Bayes's theorem, obtaining
\begin{eqnarray}
  \Pr( \y ) &=&  \sum_\z \sum_\r \Pr(\r) \cdot \Pr( \y,\z|\r ) \nonumber \\
            &=&  \sum_\z \sum_\r 2^{-n} ~\bra{\y,\z} Q_B \cdot U(\r) \cdot Q_B \ket{\0,\0}^2 
                 \nonumber \\
            &=&  \sum_\z \sum_\r 2^{-n}   \left[ 2^{-m-n} \sum_{\s,\s''} 
                 (-1)^{\s \cdot \y + f^{\s \cdot \c}(\s'') \cdot \z + \s'' \cdot \r} \right]^2 
                 \nonumber \\ 
            &=&  2^{-2m-n} \sum_{\s,\s',\s''} (-1)^{(\s+\s') \cdot \y} 
                 ~\{ f^{\s \cdot \c}(\s'') = f^{\s' \cdot \c}(\s'') \} \nonumber \\
            &=& \frac{1}{2^{2m}} \sum_{ \s,\s' \in \{0,1\}^m }  
                 (-1)^{ \y \cdot (\s+\s') } ~\{ \s \cdot \c \equiv_q \s' \cdot \c \},
\end{eqnarray}
where $\s \cdot \c$ is being used as a shorthand for $\sum_{j=1}^m s_j \cdot c_j$, and so on.
(The final bracket of the right side of the equation is a sort of Kronecker delta function~: highly discontinuous.  It takes the value 1 when the difference between $\s \cdot \c$ and $\s' \cdot \c$ is a multiple of $q$, and takes the value zero otherwise.)

Therefore, whenever $\c$ is an integer tuple, for all positive integers $q$, we have the following lemma~:
\begin{lemma}  \label{lem:pseudoeigenvectorthing}
For all $\y \in \{0,1\}^m$, 
\begin{eqnarray*}
  \frac{2^m}{q} ~\sum_{\kappa = 0}^{q-1} ~\prod_{j=1}^m
                 \Bigl( 1 + (-1)^{y_j} \cos( 2\pi c_j \kappa/q ) \Bigr)
  &=&  \!\!\!\!\!\!
  \sum_{ \s,\s' \in \{0,1\}^m }
                 (-1)^{ \y \cdot (\s+\s') } ~\{ \s \cdot \c \equiv_q \s' \cdot \c \}.
\end{eqnarray*}
\end{lemma}
\begin{proof}
  We have already seen that this formula holds approximately whenever there \emph{exists} a permutation on $2^n$ elements all of whose cycles have length $q$, and that the quality of the approximation increases without limit as the permutation considered tends more to be composed of length $q$ cycles (\cf{}~Lemma~\ref{lem:eigest1}).  Yet since neither side of the present equation depends on $n$, the formula must be exact.
\end{proof} 

This lemma then yields a corollary regarding (classical) randomized approximation schemes for counting solutions to modular partition problems (a variant of subset-sum), which may be of independent interest.
\begin{corollary}
  Let $q$ be any positive integer modulus, and let $\c = \{c_j\}_{j=1}^m$ be a set of integer weights.  Let $\s$ denote a uniformly selected $m$-bit string to select a subset of weights, and let $\bar{\s}$ denote its complement.  Let $\kappa$ denote a uniformly selected element of $\ZZ/q\ZZ$.  Then
\begin{eqnarray*}
  \Ex_\kappa \left[~ \prod_{j=1}^m \cos( 2 \pi c_j \kappa/q ) ~\right]
  &=&  \Ex_\s \left[~ \Bigl\{ \s \cdot \c  \equiv_q  \bar{\s} \cdot \c \Bigr\} ~\right].
\end{eqnarray*}
\end{corollary}
\begin{proof}
  This is shown simply by taking the Fourier transform of the result of the lemma above.
  It can also be seen directly by expanding each factor $\cos(2 \pi c_j \kappa/q)$ out to $\frac{\omega + \omega^*}2$ where $\omega = e^{2 \pi i c_j \kappa/q}$, and cancelling terms. 
\end{proof}

The right side of the equation counts the number of solutions to the modular partition problem, and involves $2^m$ terms, which could be prohibitive to exhaust over.  The left side however counts just $q$ terms, each a product of $m$ factors, which may be much smaller.  If we take $q$ to be the sum of all the weights, then we solve the ordinary (non-modular) partition problem in time $O( q \cdot m )$.  (This is still exponential if any weight is exponentially large, and therefore not a generally efficient solution to this $\NP$-complete problem.)
It is left for future work to integrate this idea properly into quantum algorithms for subset-sum problems.

\subsubsection*{Continuous problems and Pell's equation}

We close the Chapter with a proof that eigenvalue estimation within $\FH_2$ generalises to the context of continuous groups (\cf{}~\cite{lit:Halesthesis}), illustrated by reference to the number-theoretic problem of solving \emph{Pell's Equation} \cite{lit:Hal0205}.

\begin{theorem}  \label{thm:Pell}
  The decision version of the problem of solving Pell's equation lies within $\FH_2$.
\end{theorem}

The proof uses the following lemma~:
\begin{lemma}  \label{lem:otherpseudoeigenvectorthing}
  Let $\c = \{c_j\}_{j=1}^m$ be an $m$-tuple of integers.
For all real $\phi$, for all $\y \in \{0,1\}^m$,
\begin{eqnarray*}
  \prod_{j=1}^m \Bigl( 1 + (-1)^{y_j} \cos( c_j \phi ) \Bigr)
  &=&
  \frac1{2^m} \sum_{\s,\s' \in \{0,1\}^m} (-1)^{\y \cdot (\s+\s')} 
                 \cos( (\s-\s') \cdot \c ~\phi ).
\end{eqnarray*}
\end{lemma}
\begin{proof}
  Let $\d$ denote an $m$-bit string and let $\bar{\d}$ denote its complement.
By induction on $m$, with liberal use of the basic trigonometric identity
\begin{eqnarray}
  \frac12\bigl(~ \cos( A + B ) + \cos( A - B ) ~\bigr)  &=&  \cos(A) \cos(B),
\end{eqnarray}
it follows that 
\begin{eqnarray}  \label{eqn:trigid}
  \Ex_\d[~ \cos( (\d-\bar\d)\cdot\c~\phi ) ~]  &=&  \prod_{j=1}^m \cos( c_j \phi ).
\end{eqnarray}
Then we can write $\d := \s \oplus \s'$, and we can break $\s$ up into two parts---one part supported by $\d$ and one part supported by $\bar\d$---writing $\s = \a+\b$ where $\a \le \d$ and $\b \bot \d$.  Then $\s' = \s \oplus \d = \bar\a + \b$.

From line~(\ref{eqn:trigid}),
\begin{eqnarray}
  \prod_{j : d_j=1} \cos( c_j \phi )  
    &=&  \Ex_\a [~ \cos( (\a-\bar\a)\cdot\c~\phi ) ~]  \\
    &=&  \Ex_\a \Ex_\b [~ \cos( (\a+\b-\bar\a-\b)\cdot\c~\phi ) ~]  \nonumber \\
  \sum_\d (-1)^{\y \cdot \d} \prod_{j : d_j=1} \cos( c_j \phi )  
    &=&  \sum_\d (-1)^{\y \cdot \d} ~\Ex_\s [~ \cos( (\s-(\s \oplus \d))\cdot\c~\phi ) ~] 
    \nonumber \\ 
     =~~ \sum_\d \prod_{j : d_j=1} (-1)^{y_j} \cos( c_j \phi )  
    &=&  \frac1{2^m} \sum_{\s,\s'} (-1)^{\y \cdot (\s+\s')} \cos( (\s-\s')\cdot\c~\phi ),
    \nonumber
\end{eqnarray}
whence the statement of the lemma follows by factoring the left side.
\end{proof}

This identity is useful whenever it is natural to think about continuous examples of eigenvalue estimation (\cf~\cite{lit:Halesthesis}).  A good example would be Hallgren's method for solving Pell's equation efficiently \cite{lit:Hal0205}, which first estimates the real valued regulator, $R$, of a real quadratic number field, by working with a computable real pseudo-periodic function whose period is $R$.  Hallgren's method was developed as an extension of Shor's algorithm, so here we sketch a method for recasting it as an extension of Kitaev's eigenvalue estimation, so that it can be rendered within an $\FH_2$ framework (\ie{} using permutation gates and limiting Hadamard transforms to just two time-slices within the circuit) for our Theorem~\ref{thm:Pell}.

We take $h$ to be a map that `walks out' a prescribed distance along the metricated principal cycle of reduced principal ideals of the quadratic number field specified by the problem equation, and returns a representation of the reduced ideal thereby reached, together with an `overshoot' distance for how far `past' that ideal the walk-distance goes.  This can be achieved using a circuit that computes a series of giant steps and small steps to compute the new ideal, and then rounds off the remainder value.
(See \cite{lit:Joz0302} for a full discussion of the relevant number theory and algorithmics---but the notation here is a little different.)  

Identify $\RR$ with a covering of the principal cycle, so that pictorially speaking, the principal ideals $\{\mathfrak{i}_0, \mathfrak{i}_1, \ldots, \mathfrak{i}_{C-1}\}$ are laid out on a cycle of length $R$, each ideal having a unique representation.
The space on the cycle between successive ideals can be discretised to precision $\frac1N$, for some large integer $N$, so that any real number corresponds to a real point on the cycle, and is discretely approximated by quoting the (unique representation of the) first ideal `below' it on the cycle ($\mathfrak{i}(x/N)$) together with the approximate number of steps of size $\frac1N$ from that ideal up to the point in question ($k(x/N)$).  
We write $h_N$ for the function that approximates $h$ to precision $\frac1N$ in this sense~: if $h(x/N) = (\mathfrak{i}(x/N),k(x/N))$ then $h_N(x) = (\mathfrak{i}(x/N),\lfloor N \cdot k(x/N) \rfloor/N)$, where $x$ is restricted to integers.
Write $q = \lceil N \cdot R \rceil$ for the so-called pseudo-period of the cycle.
Here $h_N$ is a pseudo-periodic function of period $q$, because $h_N( q ) = h_N( 0 ) = (\mathfrak{i}_0, 0)$.  Because of the rounding down that takes place when computing $h_N$, the codomain of $h_N$ will likely contain more that $q$ distinct points.  But as explained in \cite{lit:Joz0302}, the `extra' points quickly become insignificant as $N$ becomes large enough.

We use an approximation of `pseudo-eigenvectors', together with the metaphor of \emph{state collapse}, to see how our standard method for eigenvalue estimation in $\FH_2$ still works in this continuous context, as follows.

\begin{proof}[Proof of Theorem \ref{thm:Pell}]
Consider this transformation---implementable unitarily using a polynomially sized permutation circuit \cite{lit:Hal0205}---where $\s$ is a string of $m$ bits, and $\c=\{c_j\}_{j=1}^m$ are appropriately chosen integers~:
\begin{eqnarray}  \label{eqn:discNoise}
    \ket{\s} ~\otimes~ \ket{ h_N( 0 ) }
    &\mapsto&  
    \ket{\s} ~\otimes~ \ket{ h_N( \s \cdot \c ) }.
\end{eqnarray}
Consider also the following set of `pseudo-eigenvectors' (\cf{}~line~(\ref{eqn:eigenvec}))~:
\begin{eqnarray}  \label{eqn:pseudeigenvec}
  \left\{~ \ket{\lambda(\kappa)} ~:=~ \frac{1}{\sqrt{q}}\sum_{j=0}^{q-1} \exp{\left( \frac{-2\pi i j \kappa}{NR} \right)}\ket{h_N(j)} ~\right\}_{\kappa=0}^{q-1}.
\end{eqnarray}
These are not quite orthogonal, but make a good approximation to an orthonormal basis for the space spanned by the vast majority of the computational basis of the second register.

As with the usual version of Kitaev's algorithm, one performs the above transformation in superposition and then imagines `collapsing the state' of the second register onto one of the vectors above, selected randomly. 
This gives---to good approximation---the superposition
\begin{eqnarray}
  \sum_\s \exp{ \left( \frac{2\pi i (\s \cdot \c + \epsilon_\s) \kappa}{NR} \right) } \ket\s,
\end{eqnarray}
the superposition being over those $\s$ for which there exists a $j$ such that $h_N( \s \cdot \c) = h_N( j )$, which will be at least a proportion $1-\frac1N$ of them.  Here $\epsilon_\s$ denotes some unknown `noise' term so that $\s \cdot \c + \epsilon_\s \equiv j \pmod{NR}$.  Therefore $\epsilon \in [0,1)$.

The first register is measured in the Hadamard basis as usual, and so the probability distribution for the returned string---conditioned on the stochastic choice of $\kappa$---will be 
\begin{eqnarray}
  \Pr(\y)  &\approx&  \frac1{2^{2m}} \sum_{\s,\s'} (-1)^{\y \cdot (\s+\s')}
    \cos\left( \frac{2\pi (\s \cdot \c - \s' \cdot \c + \epsilon)\kappa}{NR} \right),
\end{eqnarray}
where $\epsilon \in [0,2)$ (differing from term to term) arises from two combined noise terms.

There is an $O(1)$ probability that the $\kappa$ selected stochastically will be sufficiently small that the term $\epsilon \kappa$ in the log of the phase makes very little difference to any of these probabilities, and so Lemma~\ref{lem:otherpseudoeigenvectorthing} tells us that 
\begin{eqnarray}
  \Pr(\y)  &\approx& \frac1{2^m} \prod_{j=1}^m \left( 1 + (-1)^{y_j} \cos\left( \frac{2\pi c_j \kappa}{NR} \right) \right),
\end{eqnarray}
which means that with $O(1)$ probability the standard estimation algorithm of \S\ref{sect:eigest} will work as intended, recovering the real value $\phi = \kappa / NR$ for some random $\kappa$.  By running the algorithm twice, values can be recovered for two different $\kappa$s.  With good probability these will be coprime, and so a continued fractions analysis of the ratio of the two recovered $\phi$ values will likely yield the actual integer values of the $\kappa$s, whence the actual value of $NR$---and hence of $R$ itself---can be recovered.
\end{proof}

\cleardoublepage

\chapter{The Clifford-Diagonal Hierarchy}  \label{chap:IQP}

The theme for this chapter is ostensibly closely related to the previous one.  
Instead of taking polynomially-bounded `classical' circuits and Hadamard transforms as the building blocks for algorithms as we did for the Fourier hierarchy, here we take Clifford circuits and `diagonal' circuits.  

Diagonal circuits have the property that every gate commutes with every other gate, and so no `temporal complexity' can be encoded within a part of a circuit built exclusively from these.  We introduce ``$\IQP$ computing'' (see \S\ref{def:IQP}) as a particularly simple paradigm for understanding what kinds of probability distribution can be sampled using very little temporal complexity.
As before, we use a discrete TIME model and a finite-dimensional unitary space, but in this chapter complex phases will be used within the `diagonal' circuits.\footnotemark{}
\footnotetext{Much of the content has been extracted from a 2008 publication of mine with Michael Bremner \cite{me:IQC}.}

Many of the concepts introduced previously will be seen to translate into this framework (\cf{} circuit families~\S\ref{sect:cdqc}, adaption~\S\ref{sect:adaption}, algorithm hierarchies~\S\ref{sect:defFH}, oracles~\S\ref{sect:circuit-interfaces}...); but rather than showing how to cast well-known algorithms in the new paradigm, we instead use it to find a \emph{new} application for quantum algorithmics.  Our new application (see~\S\ref{sect:protocol}) consists in one side of a particular novel two-player interactive protocol, which we conjecture cannot be completed via purely classical means, but which completes using the idea of $\IQP$ computing.
At present, the only useful application of this protocol seems to be for demonstrations of computing power exceeding that of classical computation.

\section{Overview}

The Clifford-Diagonal (CD) hierarchy is, by analogy with the Fourier hierarchy, an arrangement of complexity classes, culminating in $\BQP$.  It has not been formally introduced in the literature, though it is certainly implicit in \cite{lit:Browne06}. 
Informally speaking, an algorithm is said to lie in the $k$th level of the CD hierarchy if it can be rendered using circuits that interleave $k$ layers of Clifford gates and `diagonal' gates (unitary maps whose description in the Hadamard basis is diagonal).  As with the Fourier hierarchy of Chapter~\ref{chap:FH} (in particular \S\ref{sect:adaption}), one can naturally define the CD hierarchy using \emph{mixed adaption}, or define a slightly stricter version using \emph{quantum adaption}, or define an `oracular' version using \emph{classical adaption}.  Our focus in this Chapter is only on the latter of these three options, described more fully in \S\ref{sect:Clifford}. 
It is somewhat surprising that the oracular definition might have any computational power exceeding that of a classical computer, because there is essentially no temporal `structure' encoded within such an oracle.  Although we have not found a decision language that can be decided more quickly with the help of the oracle, our main technical contribution is to identify a novel two-party `pseudo-cryptographic' protocol that seems to be efficient only when one of the parties can implement the oracle.  This is described in \S\ref{sect:protocol} and analysed in \S\ref{sect:heuristics} where the cryptographic analogy is emphasised.

To grasp the motivation behind the division of circuitry into `Clifford' parts and `diagonal' parts, it is necessary first to understand the computational capabilities of these parts in isolation.  
The Clifford group may be defined in terms of the Pauli group, which is itself defined by its action on \emph{qubits}.  Therefore our descriptions of quantum circuitry will be limited to qubit processing.  Many of the concepts required for discussing quantum circuitry on qubits have already been given in~\S\ref{chap:approach}.
Since the C-Not gate, the single-qubit Pauli gates, and the single-qubit Hadamard gates are all in the Clifford group, and since the single-qubit $\pi/8$ rotation gate forms a `diagonal' group, it is clear that the Clifford-Diagonal decomposition methodology can be seen as arising from the standard $\BQP$-universal gateset and hence allows for universal quantum computation in the limit of allowing arbitrarily many (\ie{} polynomially many) interwoven layers of circuitry.

We begin with some definitions (\S\ref{sect:IQPdefs}), then in \S\ref{sect:physics} we discuss architectures within which implementation of the lower levels of the CD hierarchy might be comparatively easy, before proceeding with a study of the mathematics and algorithms (though \emph{not} classes of decision languages) of the lowest levels of the hierarchy in \S\ref{sect:analysis}.

\section{Definitions}  \label{sect:IQPdefs}

We define Clifford circuits, X-programs, and the $\IQP$ oracle.  The $\IQP$ oracle can be thought of as standing in the same relationship to the CD hierarchy as the Fourier Sampling oracle stands to the Fourier hierarchy (\cf{}~Chapter~\ref{chap:FH}).  In this dissertation, the CD hierarchy itself is not considered directly, and so detailed definitions for it are omitted.

\subsection{Basic definitions}  \label{sect:Clifford}

\subsubsection*{Definition of the Clifford group}

\begin{definition}
  Within the algebra of unitary maps on some number $n$ of qubits, $C$ is in the Clifford group if for every $P$ in the Pauli group, $C \cdot P \cdot C^\dag$ is also in the Pauli group.
\end{definition}
Thus the Clifford group is the (discrete) group of unitary maps acting on qubits that stabilises the Pauli group.
\begin{eqnarray}
  \mbox{Pauli group}     &=&  \span{ X_j, ~Z_j ~:~ j \in [1..n] }  
\end{eqnarray}
has cardinality $4^{n+1}$, where $n$ counts the number of qubits under consideration, if a complex global phase change by $i$ is included for mathematical convenience.
If we remove the global phases of $i$, but leave the global $-1$ phases, we obtain the \emph{signed Pauli group}, of cardinality $2 \cdot 4^n$.
\begin{eqnarray}
  \mbox{Clifford Group}  &=&  \span{ \Lambda(Z), ~\sqrt{Z}, ~H }  
\end{eqnarray}
has cardinality ~$2^{n^2+2n+3} \cdot 3 \cdot 15 \cdot 63 \cdots (4^n-1)$,~ (\cf~\cite{lit:NRS0009}), where $n$ counts the number of qubits under consideration.  Global phase changes in multiples of \emph{eighth} roots of unity are automatically included by the definition above, constituting the centre of the Clifford group, so one should remove a factor of 8 from the cardinality if not wishing to count these.  Then this cardinality can be understood as arising from automorphisms of the signed Pauli group~: $2 \cdot (4^n-1)$ choices for the image of $X_1$, then $2 \cdot 2 \cdot 4^{n-1}$ choices for the image of $Z_1$, and so on down to $2 \cdot (4^1-1)$ choices for $X_n$, then $2 \cdot 2 \cdot 4^0$ choices for $Z_n$---the total cardinality being that of the Clifford group quotiented by its centre.

The Gottesman-Knill theorem provides that so-called \emph{stabilizer states}, which are the orbit of $\ket{\0}$ under the Clifford group, are efficiently representable, in such a way that the dynamics of the Clifford group \emph{acting on its own} are perfectly tractable classically~\cite{lit:AG0406}.  A particularly efficient classical simulation method is given in~\cite{lit:AB06}.
This kind of computation is not universal even for classical computation.  Clifford circuits acting on stabilizer states, with single-qubit measurements in the computational basis, generate the same kinds of computations in polynomial time as are available using just classical parity-circuits.  As regards decision languages, the class of languages decidable by a log-space Turing machine equipped with an oracle to analyse such circuits is none other than $\ParL$ (\cf~\S\ref{sect:egDQC1}).  This class, lying somewhere between $\L$ and $\P$, forms a natural `base' for many computing reductions.
(For recent work clarifying the `role' of $\ParL$ in quantum algorithmics, see \cite{lit:vdN08}.)

\subsubsection*{Definition of X-programs}

\begin{definition}
 An ``X-program'' on $n$ qubits (\cf{}~\cite{me:IQC}) is a list of pairs $(\theta_\p, \p) \in [0,2\pi] \times \FF_2^n$, so that $\theta_\p$ is an angle and $\p$ is a string of $n$ bits.  The order of the list is unimportant, but its length must scale polynomially in $n$, when considering uniform families.
\end{definition}

Each pair $(\theta_\p, \p) \in P$ is called an \emph{element} of the X-program.
To this X-program $P$ we associate an Hamiltonian, denoted
\begin{eqnarray}  \label{eqn:dist1b}
  \H_P  &:=&  \sum_\p ~\theta_\p \prod_{j:p_j=1}\! X_j,
\end{eqnarray}
having as many terms as there are elements in the list defining $P$.
That is to say, each string $\p$ indicates a subset of the $n$ qubits for Pauli $X$ to act upon, with `action'~$\theta_\p$. 

The `diagonal' unitary map produced by such an X-program $P$ is then taken to be 
\begin{eqnarray}  \label{eqn:ham}
  \exp(~ i \H_P ~)  &=&  
    \prod_\p ~\exp\left(~ i\theta_\p \prod_{j : p_j=1}\! X_j ~\right),
\end{eqnarray}  
which is indeed diagonal in the Hadamard basis.

When an Abelian group is being used for the gates within a circuit---as is the case here---it will be essentially devoid of temporal structure, since the order of the gates is immaterial.
It is convenient then to think in terms of the Hamiltonian, because every term of the Hamiltonian commutes with every other.  Each term, individually described by a pair $(\theta_\p,\p)$, acts on the qubits indicated by the string $\p$, with action\footnotemark{}~$\theta_\p$.
\footnotetext{Action is the temporal integral of work, in classical physical terms.}
So the \emph{circuit} nature of an implementation is not especially relevant to the unitary map associated to an X-program. 

Mathematically, one can think of $P$ as a function from $\FF_2^n$ to $\RR$, sending $\p$ to $\theta_\p$ if $(\theta_\p,\p)$ is an element, and sending $\p$ to $0$ if there is no corresponding element.  But since it is usually our intention that the number of elements will be far smaller than $2^n$, and since we will often take all the (non-zero) $\theta_\p$ values to be the same, it is also convenient to think of $P$ simply as a subset of $\FF_2^n$.  When we do not wish to refer to the $\theta_\p$ values, it is sometimes convenient to think of the subset $P$ alternatively as a binary matrix of width $n$, whose rows correspond to the elements of the X-program.

\subsection{Definition of $\IQP$ oracle}  \label{def:IQP}

An X-program is an explicitly `quantum' object, but it is also convenient to have a non-quantum description.  Analogously to definition~\ref{def:FSO} for the Fourier Sampling oracle, we define the $\IQP$ oracle.
Introduced in \cite{me:IQC}, the abbreviation denotes ``Instantaneous Quantum Polynomially-bounded''.  Here `instantaneous' refers to the explicit absence of temporal structure within the X-program description, though any given implementation may well require a non-trivial amount of time to run.
\medskip
\begin{definition}  \label{def:newIQP}
In the notation of line~(\ref{eqn:dist1b}), the \emph{$\IQP$ oracle} is a device which, on input a classical description $P$ of an X-program, returns a single sample string $\x$ from the probability distribution given by
\begin{eqnarray}  \label{eqn:dist1}
  \Pr(\X=\x)  &:=&  \left| \bra\x \exp\left(~ i\H_P~\right) \ket{\0} \right|^2 \!\!.  
\end{eqnarray}
\end{definition}
The \emph{output string} $\x$ is simply a measurement result, regarded as a (probabilistic) sample from the vector space $\FF_2^n$, where again $n$ counts the number of qubits mentioned in $\H_P$.

Our interest lies primarily \emph{not} in the decision languages that polytime-bounded machines can decide with access to such an oracle (\ie{} $\BPP^\IQP$), but in the wider notions of computing that go beyond mere decision languages, to encompass other computational concepts such as interactive games.  
As an historical aside, we note that Simon \cite{lit:Si97} wrote about algorithms that use nothing more than an oracle and an Hadamard transform, and which therefore could be described as `temporally unstructured'.  However, his notion of `oracle' was one tailored for a universal quantum architecture, being essentially an arbitrarily complex general unitary transformation, and since there is no natural notion of one of these within our `temporally unstructured' paradigm, the oracle of Simon's algorithm cannot be simulated by an $\IQP$ oracle.

\subsection{Adaption description}  \label{sect:CD_adaption}

We have said that an X-program contains no `temporal structure' because it makes no difference the order in which the Hamiltonian terms of line~(\ref{eqn:dist1b}) are applied.  It is perhaps interesting to note that if we reintroduce temporal structure by allowing the elements of an X-program to be subject to classical parity-control as discussed in \S\ref{sect:paritycontrol}, allowing also for intermediate computational-basis measurements as occurs within the measurement-based quantum computational paradigm of \cite{lit:Raus03,lit:RB01} \etc, then the full power of $\BPP$ classical computing can be recovered.
This is because there is a simple gadget\footnotemark{} on one qubit for emulating a classical \emph{And} gate~: see Fig.~\ref{fig:and1}.
\footnotetext{The design of this gadget is based on a similar concept developed by Daniel Browne, discussed in various recent conferences.} 
(Whereas a $\BPP$ circuit composed of gates of type \emph{Not} and \emph{C-Not} can be simulated entirely within the parity logic that feeds classical data forward from one X-program to the next, to simulate a $\BPP$ gate of type~\emph{And}, for example, it is necessary to apply a three-element single-qubit program where the three elements are controlled respectively by the two inputs to the \emph{And} gate and their parity, thereby potentially increasing the overall \emph{depth} of the simulation each time an \emph{And} gate is simulated.)

\begin{figure}[h]
  \begin{center}
     \includegraphics[width=90mm]{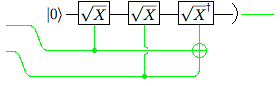}
     \caption{\label{fig:and1} Using `classical adaption', with parity-control of `diagonal' gates, to simulate a classical \emph{And} gate.  The green wires denote classical control signals to $\sqrt{X}$ gates and the classical measurement outcome.}
  \end{center}
\end{figure}

\subsubsection*{The power of post-selection} 

Recall that Aaronson \cite{lit:Aa04} showed that if one employs \emph{post-selection} (\cf~\S\ref{sect:postselection}) of the measurement results of a $\BQP$ circuit, the computational power is boosted enormously to encompass all of $\PP$.  Post-selection amounts to asking for some of the measurement outcomes to take specific values, even if those values are exponentially unlikely, before using the remaining measurement values to make a decision.  We note here that the same results hold true for circuits (or X-programs) merely implementing $\IQP$, that is, using one call to an $\IQP$ oracle.

\medskip
\begin{definition}  \label{propos:PostIQP}
  The class $\BPP^{IQP[1]}$ with post-selection is denoted $\mathbf{PostIQP}$.  It consists of all languages of the form $\Postb{\cP}{\cL_S \subseteq \cL_C}$ in the notation of Definition~\ref{def:postprob}, where $\cL_S$ and $\cL_C$ are both in $\BPP$ and where $\cP$ is a (uniform) family of probability distributions given by a uniform family of X-programs via Definition~\ref{def:newIQP}.
\end{definition}

\medskip
\begin{proposition}
  $\mathbf{PostIQP} = \mathbf{PP}$.
\end{proposition}
\begin{proof}
It suffices to show that $\mathbf{PostBQP} \subseteq \mathbf{PostIQP}$, the rest already being established.  To see this, consider any general $\BQP$ circuit that is composed of $H$ gates together with $Z$, $\Lambda (Z)$, and $\Lambda^2 (Z)$ gates (this set being $\BQP$-universal), some of whose qubits are output at the end and some of whose qubits are post-selected.  Assume without loss of generality that the input to the circuit is $\ket{0}^{\otimes n}$ and that the output and post-selection is in the computational basis.  Also assume without loss of generality that the first and last gate on every qubit is $H$.  The remaining $H$ gates which are neither first nor last on a qubit line can be replaced with the post-selection gadget of Fig.~\ref{fig:hadagadget}.  

\begin{figure}[h]
  \begin{center}
     \includegraphics[width=90mm]{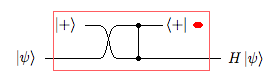}
     \caption{\label{fig:hadagadget} An Hadamard gadget, for replacing an Hadamard gate in a context where post-selection is admitted.  The lower qubit is a primal qubit, the upper one is an ancilla.  The red dot after the \emph{bra} symbol denotes post-selection of that outcome.}
  \end{center}
\end{figure}

This gadget, replacing an $H$ gate on a (primal) qubit acts as follows~: it introduces an ancilla qubit in state $\ket0$, applies an $H$ gate to it, swaps it with the primal qubit, applies a Controlled-$Z$ gate between the two, applies another $H$ gate to the ancilla, and then post-selects for that ancilla to be in state $\ket0$.  If this gadget is used everywhere to remove the `internal' $H$ gates, then we are left with a (post-selected) circuit having no inherent temporal structure.  This is because all the `internal' gates are now diagonal, and therefore mutually commutative.  (The \emph{swap} operations are to be regarded passively as relabelings, rather than as actively as gates.)  Regarding the remaining $H$ gates at the beginning and end of each qubit as passive changes of basis, it is then functionally equivalent to a circuit that can be rendered as an X-program (post-selected) in which all terms in the Hamiltonian $\H_P$ affect at most three qubits and have values $\theta$ that are some multiple of $\pi/8$.

Re-expressing this idea in the usual `calculus' of unitaries, let $d$ denote the qubit on which takes place the $H$ operation that we wish to remove, let $E$ denote the remaining qubits, and let $a$ denote the ancilla qubit that we introduce for emulating the Hadamard gate.  Then the transformation is given by
\begin{eqnarray*}
  \lefteqn{ Q_B \cdot V_{d,E} \cdot H_d \cdot U_{d,E} \cdot Q_B \cdot \ket0_d \ket\0_E } \\
  &\mapsto&  \bra{0}_a \cdot Q_B \cdot V_{d,E} \cdot \Lambda_a(Z_d) \cdot U_{a,E} \cdot Q_B \cdot \ket0_a \ket0_d \ket\0_E,
\end{eqnarray*}
where $U$ and $V$ denote the remaining parts of the circuit and $Q_B$ denotes the Hadamard gate on every qubit.
Post-selection then requires qubit $a$ to end up in state $\ket0_a$.

To see why this emulates an Hadamard transform, we need only compare $H_d\ket\psi_d$ with $\Lambda_a(Z_d) \ket\psi_a \ket+_d$ for arbitrary $\ket\psi = \alpha\ket0 + \beta\ket1$.
So $H_d\ket\psi_d \propto (\alpha + \beta, \alpha - \beta)$, written as a vector of amplitudes, while $\Lambda_a(Z_d) \ket\psi_a \ket+_d \propto (\alpha, \alpha, \beta, -\beta)$.  We then apply $H_a$ for transferring the ancilla back into the basis in which it is post-selected, obtaining the vector $(\alpha+\beta, \alpha-\beta, \alpha-\beta, \alpha+\beta)$, so that post-selection leaves $(\alpha+\beta, \alpha-\beta)$, as required.
\end{proof}

This line of reasoning indicates that it is unlikely that a classical computer would be able efficiently to tell us everything we might care to learn about the distribution of the ``$\IQP$ random variable'', $\X$, of line~(\ref{eqn:dist1}).
A similar line of reasoning is employed in~\cite{lit:FGHZ03}, where it is shown that \emph{exactly} simulating constant depth quantum circuits classically is hard, but that any family of constant-depth quantum circuits that \emph{decides} a language with \emph{zero failure probability} is efficiently classically simulable.  
We remark that the analogous Proposition (\ref{propos:post-FS}) was shown to hold for the $\cFS$ oracle of Chapter~\ref{chap:FH}, for essentially the same reason~: that post-selection supervenes classical adaption.  The proof is effectively formed by substituting Fig.~\ref{fig:hadamard1} for Fig.~\ref{fig:hadagadget}.

\section{Physical Considerations}  \label{sect:physics}

The lack of temporal structure within the Hamiltonian of an X-program, owing to the fact that its description is in terms of an Abelian group, leads one to wonder whether it might not be possible to solve the engineering challenge of implementing an $\IQP$ oracle using techniques that would be unsuitable for general purpose quantum computing.  If this were true, then the $\IQP$ framework could be useful for describing particularly `easy' quantum algorithms.  But for this idea to have any chance of making sense, it is preferable that an architecture exist whereby the \emph{physical} Hamiltonian required does not involve terms of more than two qubits.  General X-programs do not have this property.  For this reason, we also define Graph-programs.
Browne and Briegel \cite{lit:Browne06} wrote about \emph{CD-decomposition}, which is the first rigorous treatment that we know of that explicitly links graph state temporal depth with commutativity of Hamiltonian terms used to simulate a graph state computation.  It is from their terminology that we have chosen the Chapter title.

\subsection{Graph-programs}

Graph state computing architectures are popular candidates for scalable fully universal quantum processors \cite{lit:RB01, lit:Raus03}.  
Here, of course, we are concerned not with universal architectures \emph{per se}, but with the appropriate restriction to `unit time' computation~: the lowest levels of the CD hierarchy.
Another way to construct an $\IQP$ oracle uses so-called ``Graph-programs'', named since such a program is most easily described as the construction of a graph state followed by a series of measurements of the qubits in the graph state in various bases \cite{lit:Browne06}. 
A graph state has qubits that are initially devoid of information, but which are entangled together according to the pattern of some pre-specified graph, using the $G_\cG$ operator of~\S\ref{sect:clocks}.    
Unlike universal graph state computation, our Graph-programs do not admit any adaptive feed-forward, which is to say that all measurement angles must be known and fixed at compile-time, so that all measurements can be made simultaneously once the graph state has been built.  In this sense, the `depth' of a Graph-program is 1.
There is a sense in which one may regard such a program as scarcely involving \emph{dynamics} at all.

\begin{definition}
  A ``Graph-program'' is specified by giving an undirected graph $\cG(\cV,\cE)$ (usually bipartite), with labelled and distinguished vertices.  The vertex set is denoted $\cV$, of cardinality $n$, and for each $v \in \cV$ there is to be given an element of $SU(2)$; $R_v \in SU(2)$.  The edge set is denoted $\cE$.
We associate to it the probability distribution
\begin{eqnarray}  \label{eqn:dist2}
  \Pr(\X=\x)  &:=&
   \left| \bra\x ~\prod_{v \in \cV} R_v 
                 \cdot \!\!\!\!\prod_{(u,v) \in \cE}\!\Lambda_u(Z_v) 
                ~ \ket{+^n} \right|^2.  
\end{eqnarray}
To execute a Graph-program is to sample from this distribution.
\end{definition}

To implement the program, a qubit is associated with each vertex and is initialised to the state $\ket+$ in the Hadamard basis.  
Then a Controlled-$Z$ Pauli gate is applied between each pair of qubits whose vertices are a pair in $\cE$.  Since these Controlled-$Z$ gates commute, they may be applied simultaneously, at least in theory.  This process is equivalent to application of the $G_\cG$ constructor, discussed in \S\ref{sect:clocks}.
Finally, each vertex qubit $v$ is measured in the direction prescribed by its label $R_v$, returning a single classical bit.  Clearly the order of measurement doesn't matter, because the measurement direction is \emph{prescribed} rather than \emph{adaptive}.
Hence a sample from $\FF_2^n$ (a bit-string) is thus generated as the total measurement result.

\subsection{Emulation of X-programs}

We will show how Graph-programs can simulate the output of X-programs if a little trivial classical post-processing of the measurement results is allowed.
(Graph-programs would seem to be a little more general than X-programs~: it does not seem possible to emulate an arbitrary Graph-program with an X-program, because generally Graph-programs use gates that are neither in the Clifford group nor diagonal in the Hadamard basis.)

\medskip
\begin{proposition}
  Any X-program can be efficiently simulated by a Graph-program.  That is, a device for sampling from general distributions of the form at line~(\ref{eqn:dist2}) can emulate an $\IQP$ oracle, if classical post-processing is permitted, such that the size of the description of the Graph-program is polynomially bounded by the size of the description of the X-program.
\end{proposition}

\begin{proof}
Suppose we're given an X-program, $P$, thought of as a function from $\FF_2^n$ to $\RR$, and also as a list of elements $\{~ (\theta_\p, \p) ~:~ \p \in P \subseteq \FF_2^n ~\}$.  
Let $\cV$ be the disjoint union of $[1..n]$ and $P \subseteq \FF_2^n$, so that the graph state used to simulate the program will have one \emph{primal} qubit/vertex for each qubit being simulated (that is, $n$ of them), plus one \emph{ancilla} qubit/vertex for each program element $\p \in P$.  Write $\#P$ for the number of elements in $P$.  The cardinality of $\cV$ is then $n+\#P$.

We build a bipartite graph by connecting some of the primal vertices to some of the ancilla vertices.  For $j \in [1..n]$ and $\p \in P$, let $(j,\p) \in \cE$ exactly when the $j$th component of $\p$ is a 1.  Now let $R_j$ be the Hadamard element ($H$) for all primal qubits, so that all primal qubits are measured in the Hadamard basis.  And let $R_\p = \exp( i\theta_\p X )$, so that every ancilla qubit is measured in the $(YZ)$-plane at an angle specified by the corresponding program element.  See Fig.~\ref{fig:graphstate} for an example with $n=4$ primal qubits and $\#P=7$ ancill\ae{}.

\begin{figure}[h]
  \begin{center}
     \includegraphics[width=75mm]{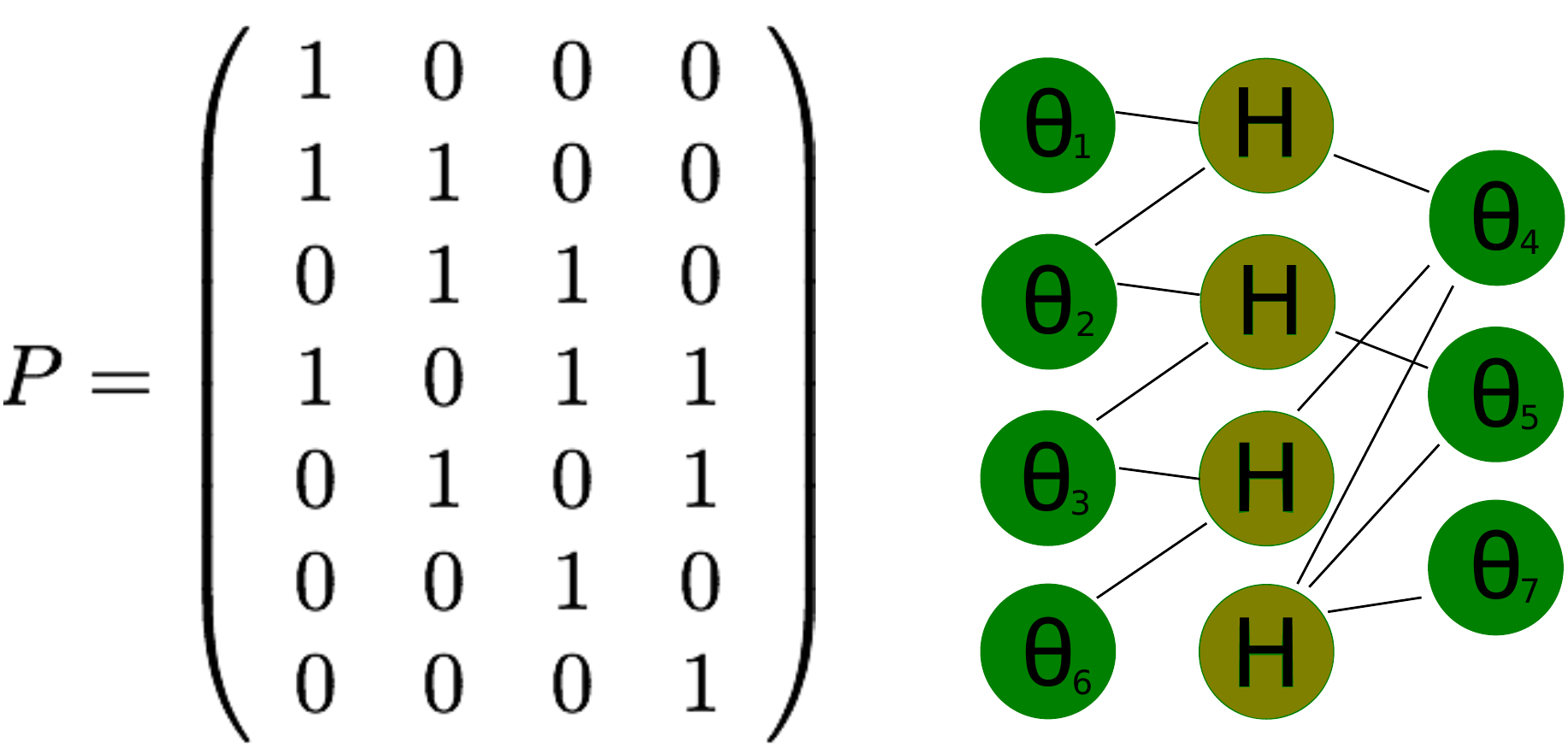}
     \caption{\label{fig:graphstate} On the left, $P$ is an example of a description of an X-program, given as a matrix (action values $\theta$ not shown in the matrix).  The graph state on the right---where the four lighter nodes represent primal qubits and the seven darker nodes represent ancilla qubits---can be used to simulate this X-program, as described in the text.}
  \end{center}
\end{figure}

If the resulting Graph-program is executed, it will return a sample vector $\x \in \FF_2^{n+\#P}$ for which the $n$ bits from the primal qubits are correlated with the $\#P$ bits from the ancill\ae{} in a fashion which captures the desired output, (though these two sets separately---\ie{} marginally---will look like flat random data).  To recover a sample from the desired distribution, we simply apply a classical C-Not gate from each ancilla bit to each neighbouring primal bit, according to $\cE$, and then discard all the ancilla bits.

One can use simple circuit identities to check that this produces the correct distribution of line~(\ref{eqn:dist1}) precisely.
For example, if we merge the post-processing C-Not gates into the quantum calculus description of the state, then we see
\begin{eqnarray*}
  \left( \prod_{(j,\p) \in \cE} \Lambda_\p(X_j) \right) \cdot
  \left( \prod_{\p \in P} e^{i \theta_\p X_\p} \right) \cdot
  \left( \prod_{j \in [1..n]} H_j \right) \cdot
  \left( \prod_{(j,\p) \in \cE} \Lambda_\p(Z_j) \right) ~\ket+^{\otimes (n+\#P)}_\cV,
\end{eqnarray*}
which is equivalent with 
\begin{eqnarray*}
  \left( \prod_{(j,\p) \in \cE} \Lambda_\p(X_j) \right) \cdot
  \left( \prod_{\p \in P} e^{i \theta_\p X_\p} \right) \cdot
  \left( \prod_{(j,\p) \in \cE} \Lambda_\p(X_j) \right) 
  ~\ket0^{\otimes n}_{[1..n]} ~\ket+^{\otimes \#P}_P.
\end{eqnarray*}
But $e^{i \theta_\p X_\p}$ conjugated by $\Lambda_\p(X_j)$ is simply $e^{i \theta_\p X_\p X_j}$, so the state above may be rewritten using the notation of line~(\ref{eqn:dist1b}) as
\begin{eqnarray*}
  \left( \prod_{\p \in P} e^{i X_\p \H_\p} \right)
  ~\ket0^{\otimes n}_{[1..n]} ~\ket+^{\otimes \#P}_P.
\end{eqnarray*}
But $\ket+_\p$ is an eigenvector of $X_\p$, so if we ignore the ancilla states (now separate anyway), we are left with
\begin{eqnarray*}
  \left( \prod_{\p \in P} e^{i \H_\p} \right) ~\ket0^{\otimes n}_{[1..n]},
\end{eqnarray*}
as required for the X-program.
\end{proof}

We note in passing that the kinds of graph called for in this particular reduction are not the usual cluster state graphs that correspond to a regular planar lattice arrangement as normally used in measurement-based quantum computation. 
The bipartite graphs described in the reduction here will usually be far from planar, for the X-programs that we'll be considering, having a relatively high genus.

\subsection{Constructing graph states}

By decomposing the graph constructor $G_\cG$ into individual gates (\cf{} \S\ref{sect:clocks}), it is clear that graph states may be constructed from polynomially many single-qubit rotations and two-qubit interactions.  
Moreover, a graph state may be constructed without \emph{inherent} temporal complexity, because there is no essential reason requiring one edge of the graph (one aspect of entanglement) to be prepared before any other.

One might argue that a \emph{physical} implementation of a graph state construction process could require time on the order of the valency of the graph in question, because it might be impractical to have an individual qubit engage in more than one entangling gate at a time.
However, even if this latter argument turns out to be relevant for architectures of interest (\cf{} \cite{lit:Hoy02}), it is still the case that the circuit-depth of graph state construction is merely \emph{logarithmic} in the valency of the graph.  This is because for the cost of some extra ancilla qubits, one can employ a binary tree of C-Not gates to `fan out' each qubit vertex of the graph onto $n$ `identical' physical qubits that \emph{together} represent the qubit associated to a graph vertex.  Then the logical state $\ket+$ is rendered physically as $\ket{00\ldots0} + \ket{11\ldots1}$ on these qubits.  The entangling operation $E_\cG$ can then be rendered using a circuit of depth one, since each Controlled-$Z$ gate impinging on the vertex can now use a distinct physical qubit.  Of course, the `fan-out' procedure must then be reversed after the vertices have been entangled and before they are measured, again using a binary tree of C-Not gates, at a cost of logarithmic circuit depth.  This trick was pointed out by Moore and Nilsson in \cite{lit:MN98}, who introduced the class $\QNC$, thereby effectively showing that $\IQP$ circuits are renderable in quantum logarithmic parallel time~:
\begin{proposition}
  ~~$\BPP^\IQP \subseteq \BPP^{\QNC^1}$.
\end{proposition}

\section{Mathematical Analysis}  \label{sect:analysis}

This section provides mathematical background for analysing X-programs, and hence $\IQP$ computing, for the later algorithmic constructions.

\subsection{Computational paths}

Using the idea of counting computational paths, we can simplify the expression for an $\IQP$ output probability distribution as follows.

\medskip
\begin{lemma}  \label{lem:thing}
The probability distribution given at line~(\ref{eqn:dist1}) (repeated below) is equivalent to the one at line~(\ref{eqn:distpaths}) given below.  
(Here $P$ denotes an X-program having $k$ elements, and we use the same symbol $P$ to denote the $k$-by-$n$ binary matrix whose rows are the $\p$ vectors of the X-program under consideration.)
\begin{eqnarray}
  \Pr(\X=\x)  &:=&  \left| \bra\x \exp\left(~ i\H_P~\right) \ket{\0} \right|^2 \!\!
  \nonumber \\  
                 \label{eqn:distpaths}
               &=& \left|~  \sum_{ \a \in \FF_2^{k} ~:~ \a \cdot P = \x} 
                      ~~\prod_{\p ~:~ a_\p = 0}   \cos \theta_\p 
                        \prod_{\p ~:~ a_\p = 1} i \sin \theta_\p ~\right|^2.
\end{eqnarray}
\end{lemma}

\begin{proof}
Using the fact that the Hamiltonian terms in an X-program all commute, we can think of the quantum amplitudes arising in an X-program implementation as a sum over paths, 
\begin{eqnarray}
  \lefteqn{ \bra{\x}~ \prod_\p 
            \left( \cos \theta_\p ~+~ i \sin \theta_p \prod_{j:p_j=1} X_j \right) 
            ~\ket{\0} }  \nonumber \\
  &=&
  \bra{\x} ~\sum_{ \a \in \FF_2^k} 
    ~\prod_{\p ~:~ a_\p = 0}   \cos \theta_\p 
     \prod_{\p ~:~ a_\p = 1} i \sin \theta_\p 
    ~\prod_{j=1}^n X_j^{(\a \cdot P)_j} ~~\ket{\0},
\end{eqnarray}
and hence derive a new form for the probability distribution accordingly.
\end{proof}

\subsection{Binary matroids and linear binary codes}

Before proceeding further, it behoves us to establish the link that these formul\ae{} have with the (closely related) theories of binary matroids and linear binary codes. 

\subsubsection*{Codes}

\begin{definition}
A linear binary code, $\Code$, of length $k$ is a (linear) subspace of the vector space $\FF_2^k$, represented explicitly.
The elements of $\Code$ are called \emph{codewords}, and the Hamming weight $wt(c) \in [0..k]$ of some $c \in \Code$ is defined to be the number of 1s it has.  The rank of $\Code$ is its rank as a vector space.
\end{definition}

Linear binary codes are frequently presented using \emph{generator matrices}, where the columns of the generator matrix form a basis for the code.  If $P$ is a generator matrix for a rank $n$ code $\Code$, then $P$ has $n$ columns and the codewords are $\{ P \cdot \d^T ~:~ \d \in \FF_2^n \}$.  Fig.~\ref{fig:graphstate} includes an example of a rank $n=4$ code of length $k=7$.

Another nice way to conceptualise the $\IQP$ oracle is as a device that forms a uniform coherent superposition over codewords of a code, before measuring that state using a locally skewed basis.

\subsubsection*{Matroids}

There are many different, isomorphic, definitions for matroids, (see \cite{book:Ox92}).  
We shall adopt the following definition.

\medskip
\begin{definition}
A $k$-point binary matroid is an equivalence class of matrices defined over $\FF_2$, where each matrix in the equivalence class has exactly $k$ rows, and two matrices are equivalent (written $M_1 \sim M_2$) when for some ($k$-by-$k$) permutation matrix $Q$, the column-echelon reduced form of $M_1$ is the same as the column-echelon reduced form of $Q \cdot M_2$.  Here we take column-echelon reduction to delete empty columns, so that the result is full-rank.  Hence the rank of a matroid is the rank of any of its representatives.  
\end{definition}

Less formally, this means that a binary matroid is like a matrix over $\FF_2$ that doesn't notice if you rearrange its rows, if you add one of its columns into another (modulo 2), or if you duplicate one of its columns.  This means that a matroid is like the generator matrix for a linear binary code, but it doesn't mind if it contains redundancy in its spanning set (\ie{} has more columns than its rank) and it doesn't care about the actual order of the zeroes and ones in the individual codewords.
To be clear, when thinking of a matrix $P$, we are simultaneously thinking of its \emph{columns} as the elements of a spanning set for a \emph{code}, and its \emph{rows} as the points of a corresponding \emph{matroid}, the elements of an X-program.
Because one cannot express a matroid independently of a representation, we consistently conflate notation for the matrix $P$ with the matroid $P$ that it represents.
 
There is a definition in the literature for \emph{weighted matroids}, which in this context would correspond to allowing different $\theta$ values for different terms in the Hamiltonian of an X-program.  While mathematically (and physically) natural, such considerations would not help with the clarity of our presentation, and in most of what follows we are concerned only with X-programs for which all the (non-zero) $\theta_\p$ values are the same.

\subsubsection*{Weight enumerator polynomials}
 
Perhaps the main structural feature of a binary matroid is its \emph{weight enumerator polynomial}.  

\medskip
\begin{definition}  \label{def:WEP}
If the $k$ rows of binary matrix $P$ establish the points of a $k$-point matroid, then the weight enumerator of the matroid is defined to be the weight enumerator of the $k$-long code $\Code$ spanned by the columns of $P$, which in turn is defined to be the bivariate polynomial
\begin{eqnarray}
  WEP_\Code(x,y)  &=&  \sum_{\c \in \Code} x^{wt(\c)} y^{k-wt(\c)}.
\end{eqnarray}
\end{definition}

This is well-defined, because the effect of choosing a different matrix $P$ that represents the same binary matroid simply leads to an isomorphic code that has the same weight-enumerator polynomial as the original code $\Code$.  The exact evaluation of an arbitrary weight-enumerator is hard for the polynomial hierarchy~: see \cite{lit:Vya03} for more on the computational complexity of approximating weight-enumerators.  This suggests that a search for any $\P^\IQP$ computing method for evaluating arbitrary weight-enumerators might lead one day to a way to put $\IQP$---and hence also $\BQP$---outside of the polynomial hierarchy.

\subsubsection*{Bias in probability distributions}

\begin{definition}  \label{def:bias}
If $\X$ is a random variable taking values in $\FF_2^n$, and $\s$ is any element of $\FF_2^n$, then the \emph{bias}\footnotemark{} of $\X$ in direction $\s$ is simply the probability that $\X \cdot \s^T$ is zero, \ie{} the probability of a sample being orthogonal to $\s$.
\end{definition}
\footnotetext{Note that this definition of \emph{bias}---used throughout this Chapter---is a little different from the definition used in~\S\ref{sect:eldef}; simply to help shorten some of the formulas.}
 
Let us now consider an X-program on $n$ qubits that has constant action value $\theta$, whose Hamiltonian terms are specified by the rows of matrix $P$, as discussed earlier.
Then we can use Lemma~\ref{lem:thing} with the definition above to obtain the following expression of bias, for any binary vector $\s \in \FF_2^n$~:
\begin{eqnarray}  \label{eqn:walshpaths}
  \Pr(\X \cdot \s^T=0) 
    &=&  \sum_{\x ~:~ \x \cdot \s^T = 0} ~
         \left|~ \sum_{ \a ~:~ \a \cdot P = \x} ~
             (\cos \theta)^{k-wt(\a)}  (i \sin \theta)^{wt(\a)} 
        ~\right|^2.
\end{eqnarray}

Since it would be nice to interpret this expression as the evaluation of a weight enumerator polynomial, we are led to define $P_\s$ to be the submatrix of $P$ obtained by deleting all rows $\p$ for which $\p \cdot \s^T = 0$, leaving only those rows for which $\p \cdot \s^T = 1$.  We call the number of rows remaining $n_\s$.  
(Note, $n_\s$ is here being used for the length of the code $\Code_\s$ in deference to the usual practice of reserving the letter $n$ for code lengths.  This $n_\s$ is counting a number of rows, and should not be confused with the $n$ used earlier for counting a number of columns.)
This in turn leads to the code $\Code_\s$ being the span of the columns of $P_\s$, and likewise a submatroid (also called a \emph{matroid minor}) is correspondingly defined.

\medskip
\begin{theorem}  \label{thm:bias}
  For constant-action X-programs, the bias expression $\Pr(\X \cdot \s^T=0)$ for the random variable $\X$ of line~(\ref{eqn:dist1}) depends only on the action value $\theta$ and (the weight enumerator polynomial of) the $n_\s$-point matroid $P_\s$, as defined above. 
Moreover,  if $\Code_\s$ is a binary code representing the matroid $P_\s$, then the following formula expresses the bias~:
\begin{eqnarray}  \label{eqn:walshcode}
    \Pr(\X \cdot \s^T=0) 
    &=& \Ex_{\c \sim \Code_\s} \left[~ \cos^2\Bigl(~ \theta( n_\s ~-~ 2 \cdot wt(\c) ) ~\Bigr) ~\right].
\end{eqnarray}
\end{theorem}

\begin{proof}
To derive line~(\ref{eqn:walshcode}) from line~(\ref{eqn:dist1}) in the case that the value $\theta$ is constant, proceed as follows.
Throughout, the variable $\p$ ranges over the rows of the binary matrix $P$, which are the program elements of an X-program. 

\begin{eqnarray} 
  \Pr(\X=\x)
  &=& \left| \bra\x ~\exp\left(~\sum_\p i\theta_\p \prod_{j:p_j=1} X_j~\right)~ \ket{\0^n} \right|^2  \nonumber \\
  &=& \left| 2^{-n}\sum_\a (-1)^{\x \cdot \a^T}\bra\a ~~\exp\left(~\sum_\p i\theta_\p \prod_{j:p_j=1} Z_j~\right) 
                   \sum_\b \ket{\b} \right|^2 \nonumber \\
  &=& \left| ~\Ex_{\a}~ \left[ (-1)^{\x \cdot \a^T} ~\exp\left(~i\theta \sum_{\p} (-1)^{\p \cdot \a^T} ~\right) 
                              \right] ~\right|^2 \nonumber \\
  &=& \Ex_{\a,\d} \left[ (-1)^{\x \cdot \d^T} \exp\left(~ i\theta \sum_{\p} (-1)^{\p \cdot \a^T} 
                                                           \Bigl(1 - (-1)^{\p \cdot \d^T}\Bigr) ~\right) \right].  
\end{eqnarray}
On the second line we made a change of basis, so as to replace the Pauli $X$ operators with Pauli $Z$ ones.

\begin{eqnarray}  \label{eqn:wooo}
  \Pr(\X \cdot \s^T = 0)
    &=& 2^{n} ~\Ex_{\x} \left[~ \{\x\cdot\s^T=0\} \cdot \Pr( \X = \x ) ~\right] \nonumber\\
    &=& 2^{n} ~\Ex_{\a,\d,\x} 
        \left[ \frac{(1+(-1)^{\x\cdot\s^T})}2 (-1)^{\x \cdot \d^T}
        e^{ i\theta \sum_\p (-1)^{\p \cdot \a^T} \Bigl(1 - (-1)^{\p \cdot \d^T}\Bigr) } \right] \nonumber \\
    &=& 2^{n} ~\Ex_{\a,\d} 
        \left[ \frac{\Bigl( \{\d=\0\} + \{\d=\s\} \Bigr)}2 
             ~e^{ i\theta \sum_\p (-1)^{\p \cdot \a^T} \Bigl(1 - (-1)^{\p \cdot \d^T}\Bigr) } \right] \nonumber \\
    &=& \frac12\left( 1 ~+~ \Ex_\a \left[ e^{ i\theta \sum_\p (-1)^{\p \cdot \a^T} \Bigl(1 - (-1)^{\p \cdot \s^T}\Bigr) } \right]
        \right).
\end{eqnarray}
These transformations are conceptually simple but notationally untidy.

\begin{eqnarray}
  2 \cdot \Pr(\X \cdot \s^T = 0) - 1 
      &=& \sum_j e^{ij\theta} ~\Ex_{\a,\phi} 
        \left[ e^{i\phi\left( -j ~+~ \sum_\p (-1)^{\p \cdot \a^T} \Bigl(1 - (-1)^{\p \cdot \s^T}\Bigr) ~\right)} \right] \nonumber \\
      &=& \sum_j e^{ij\theta} ~\Pr_{\a} 
        \left(~ j ~=~ 2\!\!\!\!\!\!\sum_{\p~:~\p\cdot\s^T=1} \!\!(-1)^{\p \cdot \a^T} ~\right)  \nonumber \\
      &=& \sum_j e^{ij\theta} ~\Pr 
        \left(~ j = 2 (~ n_\s - 2 \cdot wt( \c ) ~) ~|~ \c \sim \Code_\s ~\right) \nonumber \\
      &=& \sum_w \cos( 2\theta(n_\s - 2 w) ) \cdot \Pr\left(~ w = wt( \c ) ~|~ \c \sim \Code_\s ~\right).
\end{eqnarray}
Here we have used the standard Fourier decomposition of a periodic function, and used the fact that the function is known to be real.  The variable substitution at the third line was $\c = P_\s \cdot \a^T$, understood in the correct basis.  At the fourth line it was $w = (2n_\s-j)/4$. 
Then
\begin{eqnarray}
  \Pr(\X \cdot \s^T = 0)
    &=& \sum_{w=0}^{n_\s} \cos^2(~ \theta(n_\s - 2 w) ~) \cdot \Pr\left(~ w = wt( \c ) ~~|~~ \c \sim \Code_\s ~\right) \nonumber \\ 
    &=& \Ex_{\c \sim \Code_\s} \left[~ \cos^2\Bigl(~ \theta( n_\s ~-~ 2 \cdot wt(\c) ) ~\Bigr) ~\right]. 
\end{eqnarray}
\end{proof}

To recap, this means that if we run an X-program using the action value $\theta$ for all program elements, then the probability of the returned sample being orthogonal to an $\s$ of our choosing depends only on $\theta$ and on the (weight enumerator polynomial of the) linear code obtained by writing the program elements $\p$ as rows of a matrix and ignoring those that are orthogonal to~$\s$.

We emphasise at this point the value of Theorem~\ref{thm:bias}~: it means that for any direction $\s \in \FF_2^n$, the bias of the output probability distribution from an X-program $(P,\theta)$ in the direction $\s$ depends \emph{only} on $\theta$ and the rows of $P$ that are \emph{not} orthogonal to $\s$, and not at all on the rows of $P$ that \emph{are} orthogonal to $\s$. 
Moreover, the bias in direction $\s$ depends \emph{only} on the \emph{matroid} $P_\s$, and not on the particular \emph{matrix} $P_\s$ that represents it.  That is, directional bias (definition~\ref{def:bias}) is a matroid invariant.

Note that whenever $A$ is an $n$-by-$n$ invertible matrix over $\FF_2$, then 
\begin{eqnarray}
  \p \cdot \s^T  &=&  \p \cdot A \cdot A^{-1} \cdot \s^T  ~~=~~  
                       (\p \cdot A) \cdot (\s \cdot A^{-T})^T,
\end{eqnarray} 
so any invertible column operation on matrix $P$ accompanies an invertible change of basis for the set of directions of which $\s$ is a member.  Note also that appending or removing an all zero column to $P$ has the effect of including or excluding a qubit on which no unitary transformations are performed.
Thus if $P_\s$ is a submatroid of $P$ by point-deletion, as described earlier, then if the invertible column transformation $A$ is applied to the matrix $P$ that represents the matroid $P$, then the same \emph{matroid} that was formerly called $P_\s$ is still a submatroid, but now it is represented by the matrix $P_{\s \cdot A^{-T}}$.
Likewise, appending or removing a column of zeroes to $P$ necessitates an extra zero be appended or removed from any $\s$ that serves as a direction for indicating a submatroid.
This is purely an issue of representation, and we consider that intuition about these objects is aided by taking an `abstractist' approach to the geometry, thinking of the matroid as the fundamental object.

\subsection{Matroids, unitaries, Hamiltonians, probability distributions}  \label{sect:three}

We have seen that a $k$-point matroid $P$ with a constant action value $\theta$ defines a Hamiltonian $\H_P$ on $k$ qubits (up to qubit ordering), which in turn defines a unitary $\exp( i\H_P )$, and thence a random variable, $\X$ on $\FF_2^k$, having an interesting probability distribution.  Yet it is possible that two different matroids could give rise to the same probability distribution, because two different Hamiltonians can give rise to the same unitary map.

Consider the case whereby $\theta = \pi/8$.  Notate the Pauli $X$ gate alternatively as $1-2x$, so that $x = (1-X)/2$ is represented by an integral matrix in the diagonal basis.  Then any term $\frac{\pi}8 X_aX_b \cdots X_c$ of $\H_P$ can be expanded into many terms by multiplying out the expression $\frac{\pi}8 (1-2x_a)(1-2x_b) \cdots (1-2x_c)$.  We need only keep the monomial terms of degree 3 or less in the $x$ variables, since the higher order terms will have coefficients a multiple of $2\pi$, and will therefore not be `seen' in the resulting unitary map.  Note that this expansion and truncation will cause the number of terms to `explode' only polynomially, not exponentially.  Now, rewriting each monomial back in terms of $\pi/8$ and $X$ variables, we end up with a Hamiltonian that has 3-qubit interactions at worst.  The resulting matroid is possibly larger than the initial one, but it possesses a `sparse' representative matrix whose every \emph{row} has Hamming weight at most 3.

This sort of trick can be useful in understanding the complexity of $\IQP$ algorithms, and in tailoring designs to particular architectures.  It also puts an equivalence class structure on the set of all (unweighted) binary matroids, which may be of independent interest.

\subsection{Entropy, and trivial cases}  \label{sect:entropy}

Because it will be useful later, we will define the R\'enyi entropy (collision entropy) of a random variable, before exemplifying Theorem~\ref{thm:bias} and proceeding with the main construction of this Chapter.

\medskip
\begin{definition}
The collision entropy, $S_2$, of a discrete random variable, $\X$, measures the randomness of the sampling process by measuring the likelihood of two (independent) samples being the same.  It is defined by 
\begin{eqnarray} \label{eqn:Renyi}
  2^{-S_2}  &=&  \sum_\x \Pr(\X=\x)^2
           ~~=~~  \Ex_\s \left[~ \Bigl(~ 2\Pr( \X \cdot \s^T = 0 ) - 1 ~\Bigr)^2 ~\right].
\end{eqnarray}
\end{definition}

And so there are a few `easy cases' for our $\X$ random variable of Lemma~\ref{lem:thing} that should be highlighted and dismissed up front~:

\medskip
\begin{lemma} 
For a constant-action X-program, if $\theta$ is\ldots

\begin{itemize}
  \item
\ldots a multiple of $\pi$, then the returned sample will always be $\0$. 
The collision entropy will be zero.
  \item
\ldots an odd multiple of $\pi/2$, then the returned sample will always be $\sum_{\p \in P} \p$.  
The collision entropy is zero.  
  \item
\ldots an odd multiple of $\pi/4$, then the collision entropy need not be zero, but the probability distribution will be classically simulable to full precision.
\end{itemize}
\end{lemma}

\begin{proof}
In the first case, considering line~(\ref{eqn:distpaths}), there is then a $\sin(\pi)=0$ factor in every term of the probability, except where $\x=\0$.

In the second case, considering again line~(\ref{eqn:distpaths}), there is then a $\cos(\pi/2)=0$ factor in every term, except where all the $\p$ vectors are summed together to give $\x$.  The same can also be deduced from Theorem~\ref{thm:bias}, which implies that $\x$ will be surely orthogonal to $\s$ exactly when $n_\s$ is even, \ie{} exactly when an even number of rows of $P$ are \emph{not} orthogonal to $\s$, \ie{} exactly when $\sum_{\p \in P} \p$ \emph{is} orthogonal to $\s$.

For the third case, if $\theta$ is an odd multiple of $\pi/4$, then all the gates in the program would be Clifford gates.  By the Gottesman-Knill theorem there is then a classically efficient method for sampling from the distribution, by tracking the evolution of the system using stabilisers, \etc
\end{proof}

For other sufficiently different values of the action parameter, classical intractibility becomes a plausible conjecture (\cf{} \cite{lit:SWVC08}). 
In particular, the remainder of this Chapter will specialise to the case $\theta = \pi/8$, since we are able to make all our points about the utility of $\IQP$ computing even with this restriction.

\medskip
\begin{conjecture}  \label{conj:entropy}
The expected collision entropy of the probability distribution of a randomly selected X-program of width $n$, with constant action $\pi/8$, scales as $n - O(1)$.
\end{conjecture}
\medskip

This conjecture is perhaps not directly relevant to the `hardness' of the $\IQP$ paradigm itself, but it is implicitly relevant to the design of the kind of hypothesis test that can legitimately be used to constitute the final part of the interactive proof game discussed next.  It is future work to prove this conjecture and gain a better understanding of random matroids in the context of quantum computation.

\section{Interactive Protocol}  \label{sect:protocol}

One would naturally like to find some `use' for the ability to sample from the probability distribution that arises from a temporally unstructured quantum computation; a `task' or `proof' that can be completed using \eg{} an X-program, which could presumably not be completed by purely classical means.
In this section we develop our main construction towards that goal~: a two-player interactive protocol game, with classical message passing, in which a Prover uses an $\IQP$ oracle simply to demonstrate that he does have access to an $\IQP$ oracle.

Perhaps such algorithms constructed in $\BPP^\IQP$ will be found to be the \emph{simplest} algorithms for demonstrating quantum computing, provided it is believed that they cannot be efficiently simulated classically.  Certainly such algorithms will stand a good chance of being much simpler, and requiring far fewer qubits, than are the algorithms in $\BPP^\cFS$ (or $\FH_2$, \cf{} Chapter~\ref{chap:FH}) which are for solving reasonably hard instances of certain $\NP$ problems.

\subsection{At a glance}

This section gives a brief overview of our protocol.
\emph{Alice} plays the role of the Challenger/Verifier, while \emph{Bob} plays the role of the Prover.
There are three aspects of design involved in specifying an actual ``Alice \& Bob'' game~:

\begin{itemize}
  \item[A)]   a code/matroid construction, for Alice to select a problem $P$, to send to Bob;
  \item[B)]   an architecture or technique by which Bob is able to take samples from the $\IQP$ distribution of the challenge he receives, to send back to Alice;
  \item[A')]  an hypothesis test for Alice to use to verify (or reject) Bob's attempt.
\end{itemize}

Alice uses secret random data to obfuscate a `causal' matroid $P_\s$ inside a larger matroid $P$, and the latter she publishes (as a matrix) to Bob.  Bob interprets matrix $P$ as an X-program to be run several times, with $\theta = \pi/8$.  He collects the returned samples, and sends them to Alice.  Alice then uses her secret knowledge of `where' in $P$ the special $P_\s$ matroid is hidden, in order to run a statistical test on Bob's data, to validate or refute the notion that Bob has the ability to run X-programs.

This application is perhaps the simplest known protocol, requiring (say) $\sim200$ qubits, that could be expected to convince a skeptic of the existence of some \emph{computational} quantum effect.  The reason for this is that there seems to be no classical method to fake even a \emph{classical transcript} of a run of the interactive game between Challenger and Prover, without actually \emph{being} (or subverting the secret random data of) the classical Challenger.  In this sense, the verification may be said to be ``device-independent''.

In \S\ref{sect:heuristics} we make an analysis\footnotemark{} of some (best-known) classical cheating strategies for Bob, though these are shown to be insufficient in general.
\footnotetext{The details of precisely how to make a good hypothesis test are omitted from this work for the sake of brevity, but sourcecode is available.}

\subsection{More details}

Consider therefore the following game, played between Alice and Bob.
Alice, also called the Challenger/Verifier, is a classical player with access to a private random number generator.
Bob, also called the Prover, is a supposedly quantum player, whose goal is to convince Alice that he can access an $\IQP$ oracle, \ie{} run X-programs.  The rules of this game are that he has to convince her simply by sending classical data, and so in effect Bob offers to act as a remote $\IQP$ oracle for Alice, while Alice is initially skeptical of Bob's true $\IQP$ abilities.

\subsubsection*{Alice's challenge}

The game begins with Alice choosing some code $\Code_\s$ that has certain properties amenable to her analysis. 
In particular, she chooses the code $\Code_\s$ in such a way that all the known classical cheating strategies of~\S\ref{sect:heuristics} are defeated.  Details are given in~\S\ref{sect:recommend}.

She then finds a matrix $P_\s$ whose columns generate the code (not necessarily as a basis), and ensures that there is some $\s$ that is not orthogonal to any of the rows of $P_\s$.  The vector $\s$ should be thought of not as a structural property of the code $\Code_\s$, but as a secret `locator' that she can use to `pinpoint' $P_\s$ even after it has later been obfuscated. 

\emph{Obfuscation} of $P_\s$ is achieved by appending arbitrary rows that \emph{are} orthogonal to $\s$.  This gives rise to matrix $P$.  The matroid $P$ has $P_\s$ as a submatroid, in the sense that removal of the correct set of rows will recover $P_\s$.
Alice publishes to Bob a representation of matroid $P$ that hides the structure that she has embedded.  Random row permutations are appropriate, and reversible column operations likewise leave the matroid invariant (though the latter will affect $\s$ and must therefore be tracked by Alice).

\subsubsection*{Bob's proof}

Bob, by hypothesis being capable of sampling from an $\IQP$ distribution, may interpret the published $P$ as an X-program, to be run with the (constant) action set to $\theta = \pi/8$ (say).  He will be able to generate random vectors which independently have the correct bias in the direction (unknown to him) $\s$, \ie{} the correct probability of being orthogonal to Alice's secret $\s$, in accordance with Theorem~\ref{thm:bias}.  Although he may still be entirely unable to recover this $\s$ from such samples, he nonetheless can send to Alice a list of these samples as proof that he is `$\IQP$-capable'.  
Note that Bob's strategy is error-tolerant, because if each run of the $\IQP$ oracle were to use a `noisy' $\theta$ value, then the overall proof that he generates will still be valid, providing the noise is small and unbiased and independent between runs.  
Note also that Bob can manage several runs in one oracle call, if desired, simply by concatenating the matrix $P$ with itself diagonally.  That is to say, we even avoid \emph{classical} temporal structure (\emph{adaptive feed-forward}) on Bob's part, and so can regard his part of the protocol as lying within $\P^{\IQP[1]}$.

\subsubsection*{Alice's verification}

Since Alice knows the secret value $\s$, and can presumably compute the value $\Pr(\X \cdot \s^T=0)$ from the code's weight enumerator polynomial (see Theorem~\ref{thm:bias} and recall that she is free to choose any $\Code_\s$ that suits her purpose), it is not hard for her to use a hypothesis test to confirm that the samples Bob sends are \emph{commensurate} with having been sampled independently from the same distribution that an X-program generates.
That is to say, Alice will not try to test whether Bob's data \emph{definitely fits the correct $\IQP$ distribution,} but she will ensure that it has the particular characteristic of a strong bias in the secret direction $\s$.  This enables her to test the null hypothesis that Bob is cheating, from the alternative hypothesis that Bob has non-trivial quantum computational power.  

\emph{This requires belief in several conjectures on Alice's part.}
She must believe that there is a classical separation between quantum and classical computing; in particular that an $\IQP$ oracle is not classically efficiently approximately simulable---at least she must believe that Bob doesn't know any good simulation tricks.
And she must believe that her problem is hard---at least she should believe that the problem of identifying the location of $P_\s$ within $P$ is not a $\BPP$ problem---on the assumption that the matroid $P_\s$ is known.

If he passes her hypothesis test, Bob will have `proved' to Alice that he ran a quantum computation on her program, provided she is confident that there is no feasible way for Bob to simulate the `proof' data classically efficiently, \ie{} provided she has performed her hypothesis test correctly against a plausibly best null hypothesis.

Since Alice will test to see whether Bob's data has a strong bias in direction~$\s$ (known only to her), she should first of all ensure that Bob's data does not have a strong bias in many directions simultaneously.  This is easily done by removing all  `short circuits' (\ie{} all the empty rows and all the duplicate rows) from Bob's data, before testing it.  Bob's data would not be expected to contain short circuits if the collision entropy of the distribution were high, and so Conjecture~\ref{conj:entropy} is relevant in this sense~: we believe that the collision entropy of an $\IQP$ distribution formed as described in~\S\ref{sect:recommend} will indeed be large.

\subsubsection*{Significance}

This kind of interactive game could be of much significance to validation of early quantum computing architectures, since it gives rise to a simple way of `tomographically ascertaining' the actual presence of at least \emph{some} quantum computing, modulo some relatively basic complexity assumptions, in a `device-independent' fashion.  
In this sense it is to quantum computation what Bell violation experiments are to quantum communication.
(We have serendipitously identified a construction for which the probability gap---quantum $85.4\%$ over classical $75\%$---precisely matches the gap available from Bell's inequalities.  See Lemma~\ref{lem:QRCode}, \S\ref{sect:recommend}.)

Note that this `testing concept' does not use the $\IQP$ paradigm to compute any data that is unknown to \emph{everyone} (since Alice must know $\s$ if her verification is to work), nor does it directly provide Bob with any `secret' data that could be used as a witness to validate an $\NP$ language membership claim (Bob doesn't really `learn' anything from his experiments).  Its only effect is to provide Bob with data that he \emph{can't} apparently use for any purpose other than to pass on to Alice as a `proof' of $\IQP$-capability.  It is an open problem to find something more commonly associated with computation---perhaps deciding a decision language, for example---that can be achieved specifically with $\IQP$ oracle calls.

\subsection{Recommended construction method}  \label{sect:recommend}

This section covers a specific example of a construction methodology (with implicit test methodology) for Alice, which we conjecture to be asymptotically secure (against cheating Prover) and efficient (for both Prover and Verifier).
We emphasise here that it seems not unreasonable for Alice to \emph{believe} that Bob can have no classical cheating strategy \emph{so long as} none such has been published nor proven to exist, and so our protocol may still serve as a demonstration (if not a proof) of a genuinely quantum computing phenomenon, despite the lack of proof of any simulation conjecture.

\subsubsection*{Recipe for codes}

The family of codes that we suggest Alice should employ within the context of the game outlined above are the \emph{quadratic residue codes}.  These will be shown to have the significant property that there is a non-negligible gap between the quantum- and best-known-classical-approximation expectation values for the bias in the secret direction, both of which are significantly below unity.  

Consider a quadratic residue code over $\FF_2$ with respect to the prime $q$, chosen so that $q+1$ is a multiple of eight.  The rank of such a code is $(q+1)/2$, and the length is $q$.
A quadratic residue code is a cyclic code, and can be specified by a single cyclic generator.  There are several ways of defining these, but the simplest definition for our purposes is as follows.
\begin{definition}
  The \emph{quadratic residue code} (QR-Code) of prime length $q$, where $8$ divides $q+1$, is a cyclic code over $\FF_2$ generated by the codeword that has a 1 in the $j$th place if and only if the Legendre symbol $\left(\frac{j}{q}\right)$ equals 1 (\ie{} if and only if $j$ is a non-zero quadratic residue modulo $q$).
\end{definition}
For example, if $q=7$ (the smallest example) then the non-zero quadratic residues modulo $q$ are $\{1,2,4\}$, and so the quadratic residue code in question is the rank-$4$ code spanned by the various rotations of the generator $(0,1,1,0,1,0,0,0)^T$.  

\medskip
\begin{lemma}  \label{lem:QRCode}
When $q$ is a prime and 8 divides $q+1$, then the quadratic residue code $\Code$ of length $q$ has rank $(q+1)/2$, and it satisfies
\begin{eqnarray}  \label{eqn:QRCodestats1}
  \Ex_{\c \sim \Code} \left[~ \cos^2\Bigl(~ \frac\pi8 ( q ~-~ 2 \cdot wt(\c) ) ~\Bigr) ~\right]
  &=& \cos^2( \pi/8 ) ~~=~~ 0.854\ldots
\end{eqnarray}
Moreover, it also satisfies 
\begin{eqnarray}  \label{eqn:QRCodestats2}
  \Pr\Bigl(~ \c_1^T \cdot \c_2 = 0 ~|~ \c_1, \c_2 \sim \Code ~\Bigr)
  &=& 3/4 ~~=~~ 0.75,
\end{eqnarray}
which is relevant to certain classical strategies (described in \S\ref{sect:heuristics}).
\end{lemma}

\begin{proof}
The proof for the rank of the code is a well established result from classical coding theory (see~\cite{book:MandS}).  Other classical results of coding theory include that quadratic residue codes are a parity-bit short of being self-dual and doubly even.  That is, the extended quadratic code, with length $q+1$, obtained by appending a single parity bit to each codeword, has every codeword weight a multiple of 4 and every two codewords orthogonal.  

For line~(\ref{eqn:QRCodestats1}) this means that the (unextended) code has codeword weights which, modulo 4, are half the time 0 and half the time $-1$.  On putting these values into the left side of the formula, we immediately obtain the right side.  
For line~(\ref{eqn:QRCodestats2}) this means that in the (unextended) code, any two codewords are non-orthogonal if and only if they are both odd-parity, which happens a quarter of the time; whence the formula follows.
\end{proof}

The corollary here is that if Alice uses one of these codes for her `causal' $\Code_\s$, then if Bob runs a series of X-programs (with constant $\theta = \pi/8$) described by the (larger) matrix $P$, the data samples he recovers should be orthogonal to the hidden $\s$ about $85.4\%$ of the time (\cf{} Theorem~\ref{thm:bias}); whereas if Bob tries to cheat using the classical strategy outlined in \S\ref{sect:heuristics}, then his data samples will tend to be orthogonal to the hidden $\s$ only about $75\%$ of the time (\cf{} Lemma~\ref{lem:classprob}).  
Alice's hypothesis test therefore basically consists in measuring this single characteristic, \emph{after having filtered duplicate and null data samples from Bob's dataset.}
We conjecture that Bob has no pragmatic way of boosting these signals, at least not without feedback from Alice or by expending exponential computing resources.
(In fact, it would suffice for Alice to take any singly-punctured doubly-even code for this application.)

Note that \emph{with} exponential time on his hands, Bob could choose to simulate classically an $\IQP$ oracle, in order to obtain a dataset with a bias in direction $\s$ that is approximately $85.4\%$.  Alternatively, he could consider every possible $\s$ in turn, and test to see whether the matroid obtained by deleting rows orthogonal to his guess is in fact correspondent to a quadratic residue code, assuming he knew that this had been Alice's strategy.  For these reasons, $q$ should be fairly large in any practical example (say a few hundred), to preclude such exhaustive cheating strategies.

\subsubsection*{Recipe for obfuscation}

Having chosen $q$ as outlined above, and constructed a $q$-by-$(q+1)/2$ binary matrix generating a quadratic residue code, Alice needs to obfuscate it.  The easiest way to manage this process is not to start with a particular secret $\s$ in mind, but rather to recognise the obfuscation problem as a \emph{matroid} problem, proceeding as follows~:
\begin{itemize}
  \item 
Append a column of 1s to the matrix~: this does not change the code spanned by its columns since the all-ones (full-weight) vector is always a codeword of a quadratic code.
Other redundant column codewords may also be appended, if desired.  
  \item
Append many (say $q$) extra rows to the matrix, each of which is random, subject to having a zero in the column lately appended.  This gives rise to a $2q$-point matroid, and  ensures that there now \emph{is} an $\s$ such that the causal submatroid (quadratic residue matroid) is defined by non-orthogonality of the rows to that $\s$.
  \item
Reorder the rows randomly.  This has no effect on the matroid that the matrix represents, nor on the hidden causal submatroid.  Nor does it affect $\s$, the `direction' in which the sumbatroid is hidden.
  \item
Now column-reduce the matrix.  There is no (desirable) structure within the particular form of the matrix before column-reduction, nothing that affects either codes or matroids.  Echelon-reduction provides a canonical representative for the overall matroid, while stripping away any redundant columns that would otherwise cost an unnecessary qubit, when interpreted as an X-program.  By providing a canonical representative, it closes down the possibility that information in Alice's original construction of a basis for her causal code might leak through to Bob, which might be useful to him in guessing $\s$.  Rather more importantly, this reduction actually serves to \emph{hide} $\s$.  (We can be sure by zero-knowledge reasoning that this hiding process is random~: echelon reduction is canonical and therefore supervenes any column-scrambling process, including a random one.)
  \item
Finally, one might sort the rows, though this is unnecessary.  The resulting matrix is the one to publish.  It will have at least $(q+1)/2$ columns, since that is the rank of the causal submatroid hidden inside.  
\end{itemize}

\subsection{Mathematical problem description} \label{sect:mathythingo}

What this method of obfuscation amounts to---mathematically speaking---is a situation whereby for each suitable prime $q$, we start by acknowledging a particular (public) $q$-point binary matroid $Q$, \emph{viz} the one obtained from the QR-Code of length~$q$.  Then an `instance' of the obfuscation consists of a published $2q$-point (say) binary matroid $P$; and there is to be a hidden `obfuscation' subset $O$ such that $Q = P\backslash O$; and the practical instances occur with $P$ chosen effectively at random, subject only to these constraints.  (One could choose to make $O$ bigger than $q$ points if that were desired.)  This has the feel of a fairly generic hidden substructure problem, so it seems likely that it should be \textbf{NP}-hard to determine the location of the hidden $Q$, given $P$ and the appropriate promise of $Q$'s existence within.  
More syntactically, we should like to prove that it is \textbf{NP}-complete to decide the related matter of \emph{whether or not} $P$ is of the specified form, given only a matrix for $P$.  Clearly this problem is in \textbf{NP}, since one could provide $Q$ \emph{in the appropriate basis} as an explicit witness.  We conjecture this problem to be \textbf{NP}-complete.

\medskip
\begin{conjecture}  \label{conj:NPc}
The language of matroids $P$ that contain a quadratic-residue code submatroid $Q$ \emph{by point deletion}, where the size of $Q$ is at least half the size of $P$, is \textbf{NP}-complete under polytime reductions. 
\end{conjecture}
\medskip

These sorts of conjecture are apparently independent of conjectures about hardness of classical efficient $\IQP$ simulation, since they indicate that \emph{actually identifying the hidden data} is hard, even (presumably) for a universal quantum computer.
And even should this conjecture prove false, we know of no reason to think that a quantum computer would be much better than a classical one at finding the hidden $Q$, notwithstanding Grover's quadratic speed-up for exhaustive search. 

One might compare the structure of Conjecture~\ref{conj:NPc} to that of the following important \emph{theorem} from graph theory~:

\medskip
\begin{proposition}
The language of graphs $G$ that contain a complete graph $K$ \emph{by vertex deletion}, where the size of $K$ is at least half the size of $G$, is \textbf{NP}-complete under polytime reductions. 
\end{proposition}

This is a classic result.  See \eg{} \cite{book:Papa}, where the problem in different guises is called `Clique' and `Independent Set' and `Node Cover'.  Nonetheless, we know of no way to adapt the proof to fit Conjecture~\ref{conj:NPc}.

\subsubsection*{Challenge}  \label{sect:challenge}

It seems reasonable to conjecture that, using the methodology described, with a QR-code having a value $q \sim 500$, it is very easy to create randomised Interactive Game challenges for $\IQP$-capability, whose distributions have large entropy, which should lead to datasets that would be easy to validate and yet infeasible to forge without an $\IQP$-capable computing device (or knowledge of the secret $\s$ vector).  We propose such challenges as being appropriate `targets' for early quantum architectures, since such challenges\footnotemark{} would essentially seem to be the simplest ones available (at least in terms of inherent temporal structure and number of qubits) that can't apparently be classically met.  
\footnotetext{Accordingly, Michael Bremner and I have posted on the internet a challenge problem of size $q=487$ (\href{http://quantumchallenges.wordpress.com}{http://quantumchallenges.wordpress.com}), to help motivate further study.
This challenge website includes the source code (C) used to make the challenge matrix, and also the source code of the program that we would use to check candidate solutions, excluding only the secret seed value that we used to randomise the problem.}

\section{Heuristics}  \label{sect:heuristics}

The idea behind this two-party protocol is essentially a cryptographic one.  There is an analogy to Public Key Cryptography, if one thinks of $P$ as a public key, $\s$ as a secret key, and $\IQP$ as a kind of `computational trapdoor'.  In this section, we attempt to push the analogy a little further, describing the best-known classical `attack' strategies, and also give an accounting of our failure to find a \emph{decision language} for proving the worth of $\IQP$.

It is tempting to think that it would be desirable to have $\Pr( \X \cdot \s^T = 0 )=1$, so that Bob stands a chance of finding many vectors that are \emph{surely} orthogonal to $\s$, thereby allowing for actually learning $\s$ via Gaussian elimination, thus genuinely  computing something non-trivial.  But we shall see (Theorem~\ref{thm:unitbound}) that this is precisely the condition that makes $\s$ efficiently learnable using the classical techniques described below.
This is why the code selected for the construction in~\S\ref{sect:recommend} gave a bias of $0.854\ldots$, well below 1.
For the same reason, it seems hard to find decision languages that plausibly lie in the difference $\BPP^\IQP \backslash \BPP$.

\subsubsection*{Directional derivatives}

Suppose we wish to construct a probability distribution that arises from some purely classical methods, which can be used to approximate our $\IQP$ distribution.  Our motivation here is to check whether any purported application for an $\IQP$ oracle might not be efficiently implemented without any quantum technology.  We proceed using the relatively \emph{ad hoc} methods of linear differential cryptanalysis.

For the case $\theta = \pi/8$, we will need to consider only second-order derivatives.  The same sort of method will apply to the case $\theta = \pi/2^{d+1}$ using $d$th order derivatives, but the presentation would not be improved by considering that general case here.

In terms of a binary matrix/X-program $P$, proceed by defining 
\begin{eqnarray}  \label{eqn:def_f}
  f     &:&  \FF_2^n ~\rightarrow~ \ZZ/16\ZZ, \nonumber \\
  f(\a) &\equiv& \sum_{\p \in P} (-1)^{\p \cdot \a^T} \pmod{16},
\end{eqnarray}
and notate discrete directional derivatives as 
\begin{eqnarray} \label{eqn:deriv}
  f_\d(\a) &\equiv& f(\a) - f(\a\oplus\d) \pmod{16}.  
\end{eqnarray}

Consider also the \emph{second} derivatives of $f$, given by
\begin{eqnarray}  \label{eqn:deriv2}
  f_{\d,\e}(\a) 
    &\equiv&  f_\e(\a) ~-~ f_\e(\a\oplus\d) \pmod{16} \nonumber \\
    &\equiv&  2\sum_{\p \in P_\e} (-1)^{\p \cdot \a^T} 
              \left( 1 ~-~ (-1)^{\p \cdot \d^T} \right) \pmod{16} \nonumber \\
    &\equiv&  4\!\!\!\!\sum_{\p \in P_\d \cap P_\e} (-1)^{\p \cdot \a^T} \pmod{16} \nonumber \\
    &\equiv&  4\!\!\!\!\sum_{\p \in P_\d \cap P_\e} ~\prod_{j:p_j=1} \Bigl(~ 1 ~-~ 2a_j ~\Bigr) \pmod{16} \nonumber \\
    &\equiv&  \sum_{\p \in P_\d \cap P_\e} \left(~ 4 ~+~ 8\!\!\!\sum_{j~:~p_j=1} \!\!a_j ~\right) \pmod{16},
\end{eqnarray}
each of which is quite patently a linear function in the bits $(a_1, \ldots, a_n)$ of $\a$, as a function with codomain the ring $\ZZ/16\ZZ$, regardless of the choice of directions~$\d,\e$.

\medskip
\begin{lemma}
With $f$ defined as per line~(\ref{eqn:def_f}), and $\X$ the random variable of Lemma~\ref{lem:thing}, for all~$\s$,
\begin{eqnarray}  \label{eqn:piby8}
  \Pr( \X \cdot \s^T = 0 ) 
    &=&  \Ex_\a \left[ \cos^2\Bigl(~ \frac\pi{16} \cdot f_\s(\a) ~\Bigr) \right],
\end{eqnarray}
and so the $\IQP$ probability distribution (in the case $\theta=\pi/8$) may be viewed as a function of $f$ rather than as a function of $P$.
\end{lemma}

\begin{proof}
Starting from the proof of Theorem~\ref{thm:bias}, line~(\ref{eqn:wooo}), 
\begin{eqnarray}
  \Pr(\X \cdot \s^T = 0)
    &=& \frac12\left(~ 1 ~+~ \Ex_\a \left[ e^{ i\theta \sum_\p (-1)^{\p \cdot \a^T} \Bigl(1 - (-1)^{\p \cdot \s^T}\Bigr) } \right]
        ~\right) \nonumber \\
    &=& \frac12\left(~ 1 ~+~ \Ex_\a \left[ \exp\left( \frac{i\pi}8 \bigl( f(\a) - f(\a\oplus\s) \bigr) \right) \right] ~\right) \nonumber \\
    &=& \frac12\left(~ 1 ~+~ \Ex_\a \left[ \cos\Bigl(~ \frac\pi8 \cdot f_\s(\a) ~\Bigr) \right] ~\right).
\end{eqnarray}
The second line above is obtained immediately from the first, using the definition of $f$.  
The third line follows because the expression is real-valued.  
The conclusion follows from a basic trigonometric identity, and linearity of the expectation operator.
\end{proof}

And so \emph{if} there is a hidden $\s$ such that $\Pr( \X \cdot \s^T = 0 )=1$, \emph{then} that implies $f_\s(\a) \equiv 0 \pmod{16}$ for all $\a$.
This is essentially a non-oracular form of the kind of function that arises in applications of Simon's Algorithm (\cf{} \cite{lit:Si97}), with $\s$ playing the role of a \emph{hidden shift}.
One could find linear equations for such an $\s$ if it exists, because it would follow immediately that $f_\a(\s) \equiv f_\a(\0)$ for all $\a$, and hence $f_{\d,\e}(\s) \equiv f_{\d,\e}(\0)$ for any directions $\d,\e$, which---by line~(\ref{eqn:deriv2})---is equivalent with 
\begin{eqnarray}  \label{eqn:whatever}
  \left( \sum_{\p \in P_\d \cap P_\e} \!\!\p \right) \cdot \s^T &=& 0.
\end{eqnarray}

\subsubsection*{Classical sampling}

To make use of this specific second-order differential property, we need to analyse the probability distribution that a classical player can generate efficiently from it.
Proceed by defining a new probability distribution for a new random variable $\Y$, as follows~:
\begin{eqnarray}  \label{eqn:classprob}
  \Pr(\Y=\y) 
    &:=&  \Pr_{\d,\e} \left(~ \sum_{\p \in P_\d \cap P_\e} \!\!\p ~=~ \y ~\right).
\end{eqnarray}
This may be classically rendered, simply by choosing $\d,\e \in \FF_2^n$ independently with a uniform distribution, and then returning the sum of all rows in $P$ that are not orthogonal to either $\d$ or~$\e$.

\medskip
\begin{lemma}  \label{lem:classprob}
The classical simulable distribution on the random variable $\Y$ defined in line~(\ref{eqn:classprob}) satisfies
\begin{eqnarray}  \label{eqn:classfromcode}
  \Pr( \Y \cdot \s^T = 0 ) 
    &=&  \Pr\Bigl(~ \c_1^T \cdot \c_2 = 0 ~~|~~ \c_1,\c_2 \sim \Code_\s ~\Bigr) \\
    &=&  \frac12\left(~ 1 ~+~ 2^{-rank(~ P_\s^T \cdot~ P_\s ~)} ~\right),
\end{eqnarray}
and so the bias of $\Y$ in direction $\s$ is a function of the matroid $P_\s$.
\end{lemma}

\begin{proof}
Starting from line~(\ref{eqn:classprob}),
\begin{eqnarray}
  \Pr( \Y \cdot \s^T = 0 )
    &=& \sum_{\y ~:~ \y \cdot \s^T = 0}  ~~ \Pr_{\d,\e} \left(~ \sum_{\p \in P_\d \cap P_\e} \!\!\p ~=~ \y ~\right) \\
    &=& \Pr_{\d, \e} \left(~ \sum_{\p \in P_\d \cap P_\e} \!\!\p \cdot \s^T ~=~ 0 ~\right) \nonumber \\
    &=& \Pr_{\d, \e} \left(~ wt(~ P\cdot\d^T ~\wedge~ P\cdot\e^T ~\wedge~ P\cdot\s^T ~) \equiv 0 \pmod2 ~\right) \nonumber \\
    &=& \Pr_{\d, \e} \left(~ wt(~ P_\s\cdot\d^T ~\wedge~ P_\s\cdot\e^T ~) \equiv 0 \pmod2 ~\right) \nonumber \\
    &=& \Pr_{\d, \e} \left(~ \d\cdot P_\s^T \cdot P_\s\cdot\e^T = 0 ~\right). \nonumber
\end{eqnarray}
The \emph{wedge operator} $\wedge$ here denotes the logical \emph{And} between binary column-vectors.

The first line of the Lemma follows from the direct substitutions $\c_1 = P_\s \cdot \d^T$, $\c_2 = P_\s \cdot \e^T$.
The second line follows because unimodular actions on the left or right of a quadratic form (such as $(P_\s^T \cdot P_\s)$) affect neither its rank nor the probabilities derived from it; so it suffices to consider the cases where it is in Smith Normal Form, \ie{} diagonal, which are trivially verified.    
Since this expression is patently invariant under invertible linear action on the right and permutation action on the left of $P_\s$, it too is a matroid invariant.  
\end{proof}

\subsubsection*{Correlation}

Thus we have established some kind of correlation between random variables $\X$ and~$\Y$.

\medskip
\begin{theorem}  \label{thm:unitbound}
In the established notation, for X-programs with fixed $\theta=\pi/8$,
\begin{eqnarray}  \label{eqn:unitbound}
  \Pr( \X \cdot \s^T = 0 ) = 1 
    ~~&\Rightarrow&~~  \Pr( \Y \cdot \s^T = 0 ) = 1.
\end{eqnarray}
\end{theorem}  

\begin{proof}
By Theorem~\ref{thm:bias}, the antecedent gives, for all $\c \in \Code_\s, ~n_\s \equiv 2 wt(\c) \pmod8$, where $n_\s$ is again the length of the code $\Code_\s$.
This entails that every codeword in $\Code_\s$ has the same weight modulo 4, including the null codeword, so $\Code_\s$ must be doubly even (which means every codeword has a weight a multiple of 4).
It is easy to see that doubly even linear codes are self-dual (which means that a word is a codeword if and only if it is orthogonal to every codeword).
By Lemma~\ref{lem:classprob}, the consequent is obtained.
\end{proof}

The only counterexamples to the \emph{converse} implication seem to occur in the trivial cases whereby the binary matroid $P_\s$ has circuits of length 2, \ie{} where $P_\s$ has repeated rows. 

This random variable $\Y$ is the `best classical approximation' that we have been able to find for $\X$.  
(The intuition is that it captures all of the `local' information in the function $f$, which is to say all the `local' information in the matroid $P$, so that the only data left unaccounted for and excluded from use within building this classical distribution is the `non-local' matroid information, which is readily available to the quantum distribution via the magic of quantum superposition.) 
There seems to be no other sensible way of processing $P$ (or $f$) classically, to obtain useful samples efficiently, though it also seems hard to make any rigorous statement to that effect.

\medskip
\begin{conjecture}  \label{conj:best}
The classical method defined in this section, yielding random variable $\Y$, is asymptotically classically optimal (when comparing average-case behaviour and restricting to polynomial time) for the simulation of $\IQP$ distributions arising from constant-action $\theta=\pi/8$ X-programs.
\end{conjecture}
\medskip

This conjecture lends credence to the design methodology of \S\ref{sect:protocol}.  It means that if Bob wishes to cheat, using classical techniques only and not expending sufficient time to search exhaustively for $\s$, then so far as we are aware, the best he can realistically hope to do is to use this random variable $\Y$ to make data items $75\%$ of which \emph{ought} to be orthogonal to $\s$, while hoping that in fact surprisingly many of them will turn out to be orthogonal to 1, thereby `fooling' Alice's hypothesis test.  His chances of succeeding naturally depend on how much data Alice requires for her hypothesis test, and how she trades off the probability of making a \emph{Type I error} (accepting data sampled classically from $\Y$, for example) versus the probability of making a \emph{Type II error} (rejecting data despite its having been sampled from $\X$ using $\IQP$ methods).  As far as we are aware, neither random variable $\X$ nor $\Y$ seems to be particularly useful for actually learning $\s$ for sure.

\section{Summary}

We have made a thorough study of the simplest (`temporally unstructured') part of the Clifford-Diagonal hierarchy, by considering the mathematical structures that underpin the notion of an X-program or $\IQP$ oracle.  We have looked at some different methods for conceptually implementing such a computational process and, using ideas from graph state computing (\S\ref{sect:physics}), have found an implementation for $\IQP$ within logarithmic quantum-circuit depth (\cf{}~\cite{lit:MN98}).   

Specialising to constant-action X-programs with $\theta=\pi/8$, we have shown that modulo post-selection (\S\ref{sect:postselection}) they are as powerful as $\BQP$ computation, despite the fact that they contain no temporal structure, and can always be rewritten as 3-local X-programs (\S\ref{sect:three}).  We have proposed a family of easily described challenge problems (\S\ref{sect:protocol}) that seems to capture well the complexity of this kind of problem, in context of a two-party protocol, exploiting a natural cryptographic analogy (\S\ref{sect:heuristics}).

We have given several conjectures of an open-ended nature, to indicate directions for possible future work.  We might also recommend the further study of matroid invariants through quantum techniques, or perhaps the invariants of \emph{weighted} matroids, since they seem to be the natural objects of $\IQP$ computation as hitherto circumscribed.
This would seem to be fertile ground for developing examples of things that only genuine quantum computers can achieve.

Note that if it weren't for the correlation described in Theorem~\ref{thm:unitbound}, then it would be possible to conceive of a mechanism whereby an $\IQP$-capable device could compute an actual secret or witness to something (\eg{} learn $\s$), so that the computation wouldn't require two rounds of player interaction to achieve something non-trivial.  Yet as it stands, it is an open problem to suggest tasks for this paradigm involving no communication nor multi-party concepts.




\clearpage
\pagestyle{empty}
\cleardoublepage

This version of the document was typeset from the source-file on \today.

\bigskip
Contact the author~:
\href{mailto:djs@cantab.net}{djs\_at\_cantab.net}

\end{document}